\documentclass[11pt]{article}

\usepackage{fullpage}
\usepackage{hdb_macros}
\usepackage{tikz}
\usepackage{bm}
\usetikzlibrary{positioning}
\usetikzlibrary{arrows.meta}
\usetikzlibrary{calc}
\usetikzlibrary{shapes}

\def\gapCVP{\mathrm{CVP}}
\def\coGapCVP{\mathrm{coCVP}}
\def\gapSVP{\mathrm{SVP}}
\def\BDD{\mathrm{BDD}}
\def\LWE{\mathrm{LWE}}
\def\SIS{\mathrm{SIS}}
\def\sDGS{\mathrm{DGS}}

\renewcommand{\SVP}{\gapSVP}
\renewcommand{\CVP}{\gapCVP}
\renewcommand{\GapSVP}{\gapSVP}
\renewcommand{\GapCVP}{\gapCVP}

\def\ball{\mathcal{B}}

\newcommand{\Vcap}{\ensuremath{V^{\textrm{cap}}}}
\newcommand{\Vball}{\ensuremath{V^{\textrm{ball}}}}
\newcommand{\coAM}{\cc{coAM}}
\newcommand{\coMA}{\cc{coMA}}
\newcommand{\coAMTIME}{\cc{coAMTIME}}
\newcommand{\Eclass}{\cc{E}}
\newcommand{\Par}{\mathcal{P}}

\DeclarePairedDelimiter\inner{\langle}{\rangle}

\mathtoolsset{centercolon}
\DeclareMathOperator*{\expect}{\mathbb{E}}
\newcommand{\intd}{\,{\rm d}}
\newcommand{\basis}{\mathbf{B}}

\title{Lattice Problems Beyond Polynomial Time}
\author{
Divesh Aggarwal\\National University of Singapore\\ \texttt{divesh@comp.nus.edu.sg} \and 
Huck Bennett\\Oregon State University\\ \texttt{huck.bennett@oregonstate.edu} \and 
Zvika Brakerski\\Weizmann Institute of Science\\ \texttt{zvika.brakerski@weizmann.ac.il}\and 
Alexander Golovnev\\Georgetown University\\\texttt{alexgolovnev@gmail.com} \and 
Rajendra Kumar\\ Weizmann Institute of Science \\ \texttt{rjndr2503@gmail.com} \and 
Zeyong Li\\ National University of Singapore\\ \texttt{li.zeyong@u.nus.edu} \and 
Spencer Peters\\Cornell University\\ \texttt{sp2473@cornell.edu}  \and 
Noah Stephens-Davidowitz\\Cornell University \\ \texttt{noahsd@gmail.com}   \and 
Vinod Vaikuntanathan\\MIT \\ \texttt{vinodv@csail.mit.edu} 
}
\date{\today}

\begin{document}

\pagenumbering{roman}

\maketitle

\begin{abstract}
We study the complexity of lattice problems in a world where algorithms, reductions, and protocols can run in superpolynomial time. Specifically, we revisit four foundational results in this context---two protocols and two worst-case to average-case reductions. We show how to improve the approximation factor in each result by a factor of roughly $\sqrt{n/\log n}$ when running the protocol or reduction in $2^{\eps n}$ time instead of polynomial time, and we show
a novel protocol with no polynomial-time analog.
Our results are as follows.
\begin{enumerate}
    \item We show a worst-case to average-case reduction proving that secret-key cryptography (specifically, collision-resistant hash functions) exists if the (decision version of the) Shortest Vector Problem (SVP) cannot be approximated to within a factor of $\widetilde{O}(\sqrt{n})$ in $2^{\eps n}$ time for any constant $\eps > 0$. This extends to our setting Ajtai's celebrated polynomial-time reduction for the Short Integer Solutions problem (SIS) [STOC, 1996], which showed (after improvements by Micciancio and Regev [FOCS, 2004; and SIAM J. Computing, 2007]) that secret-key cryptography exists if SVP cannot be approximated to within a factor of $\widetilde{O}(n)$ in polynomial time.
    \item We show another worst-case to average-case reduction proving that \emph{public-key} cryptography exists if SVP cannot be approximated to within a factor of $\widetilde{O}(n)$ in $2^{\eps n}$ time.
    This extends Regev's celebrated polynomial-time reduction for the Learning with Errors problem (LWE) [STOC, 2005; and J. ACM, 2009], which achieved an approximation factor of $\widetilde{O}(n^{1.5})$. In fact, Regev's reduction is quantum, but we prove our result under a classical reduction, generalizing Peikert's polynomial-time classical reduction [STOC, 2009], which achieved an approximation factor of $\widetilde{O}(n^2)$.
        \item We show that the (decision version of the) Closest Vector Problem (CVP) with a constant approximation factor has a $\mathsf{coAM}$ protocol with a $2^{\eps n}$-time verifier. This generalizes the celebrated polynomial-time protocol due to Goldreich and Goldwasser [STOC 1998; and J. Comp. Syst. Sci., 2000]. It follows that the recent series of $2^{\eps n}$-time and even $2^{(1-\eps)n}$-time hardness results for CVP cannot be extended to large constant approximation factors $\gamma$ unless AMETH is false. We also rule out $2^{(1-\eps)n}$-time lower bounds for any constant approximation factor $\gamma > \sqrt{2}$, under plausible complexity-theoretic assumptions. (These results also extend to arbitrary norms, with different constants.)
        \item We show that $O(\sqrt{\log n})$-approximate SVP has a $\mathsf{coNTIME}$ protocol with a $2^{\eps n}$-time verifier.  Here, the analogous (also celebrated!) polynomial-time result is due to Aharonov and Regev [FOCS, 2005; and J. ACM, 2005], who showed a polynomial-time protocol achieving an approximation factor of $\sqrt{n}$. This result implies similar barriers to hardness, with a larger approximation factor under a weaker complexity-theoretic conjectures (as does the next result).
        \item Finally, we give a novel $\mathsf{coMA}$ protocol for constant-factor-approximate CVP with a $2^{\eps n}$-time verifier. Unlike our other results, this protocol has no known analog in the polynomial-time regime.
    \end{enumerate}
All of the results described above are special cases of more general theorems that achieve time-approximation factor tradeoffs. In particular, the tradeoffs for the first four results smoothly interpolate from the polynomial-time results in prior work to our new results in the exponential-time world.
\end{abstract}

\thispagestyle{empty}
\newpage
\setcounter{tocdepth}{2}
\tableofcontents
\newpage
\pagenumbering{arabic}

\section{Introduction}

A lattice $\lat \subset \R^n$ is the set of all integer linear combinations of linearly independent basis vectors $\vec{b}_1,\ldots, \vec{b}_n \in \R^n$,
\[
    \lat = \lat(\vec{b}_1,\ldots, \vec{b}_n) = \{ z_1 \vec{b}_1 + \cdots + z_n \vec{b}_n \ : \ z_i \in \Z\}
    \; .
\]

 The most important computational problem associated with lattices is the $\gamma$-approximate Shortest Vector Problem ($\gamma$-$\SVP$), which is parameterized by an approximation factor $\gamma \geq 1$. Given a basis for a lattice $\lat \subset \R^n$,  $\gamma$-SVP asks us to approximate the length of the shortest non-zero vector in the lattice up to a factor of $\gamma$. The second most important problem is the $\gamma$-approximate Closest Vector Problem ($\gamma$-$\CVP$), in which we are additionally given a target point $\vec{t} \in \Q^n$, and the goal is to approximate the minimal distance between $\vec{t}$ and any lattice point, again up to a factor of $\gamma$.\footnote{These problems are sometimes referred to as $\gamma$-$\mathrm{GapSVP}$ and $\gamma$-$\mathrm{GapCVP}$, when one wishes to distinguish them from the associated search problems. In this paper, we are only interested in the decision problems and we will therefore refer to these problems simply as $\gamma$-SVP and $\gamma$-CVP, as is common in the complexity literature.} Here, we define length and distance in terms of the $\ell_2$ norm (though, in the sequel, we sometimes work with arbitrary norms).

These two problems are closely related. In particular, $\CVP$ is known to be at least as hard as $\SVP$ in quite a strong sense, as there is a simple efficient reduction~\cite{GMSS99} from $\SVP$ to $\CVP$ that preserves the approximation factor $\gamma$ and rank $n$~ (as well as the norm). Moreover, historically, it has been much easier to find algorithms for $\SVP$ than for $\CVP$ and much easier to prove hardness results for $\CVP$.

Both $\SVP$ and $\CVP$ have garnered much attention over the past twenty-five years or so, after Ajtai proved two tantalizing results. First, he constructed a cryptographic (collision-resistant) hash function and proved that it is secure if $\gamma$-$\SVP$ is hard for some approximation factor $\gamma = \poly(n)$~\cite{ajtaiGeneratingHardInstances1996,GGHCollisionFreeHashingLattice2011}. This in particular implies that secret-key encryption exists under this assumption. To prove his result, he showed the first worst-case to average-case reduction in this context. Specifically, he showed that a certain \emph{average-case} lattice problem called the Short Integer Solutions problem ($\SIS$, corresponding to the problem of breaking his hash function) was as hard as $\gamma$-$\SVP$, a \emph{worst-case} problem. Second, Ajtai proved the NP-hardness of exact $\SVP$, i.e., $\gamma$-$\SVP$ with $\gamma = 1$ (under a randomized reduction)~\cite{ajtaiShortestVectorProblem1998}, answering a long-standing open question posed by van Emde Boas~\cite{vanemdeboasAnotherNPCompleteProblem1981}.

Ajtai's two breakthrough papers led to \emph{many} follow-ups. In particular, there followed a sequence of works showing the hardness of $\gamma$-$\SVP$ for progressively larger approximation factors $\gamma$~\cite{caiApproximatingSVPFactor1999,Mic01svp,khotHardnessApproximatingShortest2005,havivTensorbasedHardnessShortest2012}, leading to the current state of the art: NP-hardness (under randomized reductions) for any constant $\gamma$ and hardness for $\gamma = n^{c/\log \log n}$ under the assumption that $\mathsf{NP} \not \subset \mathsf{RTIME}[2^{n^{o(1)}}]$. A different, but related, line of work showed hardness of $\gamma$-$\CVP$ for progressively larger approximation factors $\gamma$, culminating in NP-hardness for $\gamma = n^{c/\log \log n}$~\cite{dinurApproximatingCVPAlmostpolynomial2003}.

A separate line of work improved upon Ajtai's worst-case to average-case reduction. Micciancio and Regev showed  that Ajtai's hash function is secure if $\widetilde{O}(n)$-$\SVP$ is hard~\cite{MR04}, improving on Ajtai's large polynomial approximation factor.  
Regev also improved on Ajtai's results in another (very exciting!) direction, showing a \emph{public-key} encryption scheme that is secure under the assumption that $\widetilde{O}(n^{1.5})$-$\SVP$ is hard for a \emph{quantum} computer~\cite{regevLatticesLearningErrors2009}. To do so, Regev defined an average-case lattice problem called Learning with Errors ($\LWE$), constructed a public-key encryption scheme whose security is (essentially) equivalent to the hardness of $\LWE$, and showed a \emph{quantum} worst-case to average-case reduction for $\LWE$. Peikert later showed how to prove \emph{classical} hardness of $\LWE$ in a different parameter regime, showing that secure public-key encryption exists if $\widetilde{O}(n^2)$-$\SVP$ is hard, even for a \emph{classical} computer~\cite{peikertPublickeyCryptosystemsWorstcase2009}. (The ideas in these works have since been extended to design \emph{many} new and exciting cryptographic primitives. See~\cite{peikertDecadeLatticeCryptography2016} for a survey.)

One might even hope that continued work in this area would lead to one of the holy grails of cryptography: a cryptographic construction whose security can be based on the (minimal) assumption that $\mathsf{NP} \not\subseteq \mathsf{BPP}$. Indeed, in order to do so, one would simply need to decrease the approximation factor achieved by one of these worst-case to average-case reductions and increase the approximation factor achieved by the hardness results until they meet! However, two seminal works showed that this was unlikely. First, Goldreich and Goldwasser showed a $\mathsf{coAM}$ protocol for $\sqrt{n/\log n}$-$\CVP$, and therefore also for $\sqrt{n/\log n}$-$\SVP$~\cite{GG00}.  Second, Aharonov and Regev showed a $\mathsf{coNP}$ protocol for $\sqrt{n}$-$\CVP$ (and therefore also for $\sqrt{n}$-$\SVP$)~\cite{aharonovLatticeProblemsNP2005}. These results are commonly interpreted as barriers to proving hardness, since they imply that if $\sqrt{n/\log n}$-$\SVP$ (or even $\sqrt{n/\log n}$-$\CVP$) is NP-hard, then the polynomial hierarchy would collapse to the second level, and that the hierarchy would collapse to the first level for $\gamma = \sqrt{n}$. It seems very unlikely that we will be able to build cryptography from the assumption that $\gamma$-$\SVP$ is hard for some $\gamma = o(\sqrt{n})$, and so these results are typically interpreted as ruling out achieving such a ``holy grail'' result via this approach.

Indeed, the state of the art has been stagnant for over a decade now (in spite of much effort), in the sense that no improvement has been made to the approximation factors achieved by (1) (Micciancio and Regev's improvement to) Ajtai's worst-case to average-case reduction; (2) Regev's worst-case to average-case quantum reduction for public-key encryption or Peikert's classical reduction; (3) the best known hardness results for $\SVP$ (or $\CVP$); (4) Goldreich and Goldwasser's $\mathsf{coAM}$ protocol; \emph{or} (5) Aharonov and Regev's $\mathsf{coNP}$ protocol. (Of course, much progress has been made in other directions!)

However, all of the above results operate in the polynomial-time regime, showing hardness against polynomial-time algorithms and protocols that run in polynomial time (formally, protocols with polynomially bounded communication and polynomial-time verifiers). It is of course conventional (and convenient) to work in this polynomial-time setting, but as our understanding of computational lattice problems and lattice-based cryptography has improved over the past decade, the distinction between polynomial and superpolynomial time has begun to seem less relevant. Indeed, the fastest algorithms for $\gamma$-$\SVP$ run in time that is \emph{exponential} in $n$, even for $\gamma = \poly(n)$, and it is widely believed that no $2^{o(n)}$-time algorithm is possible for $\gamma = \poly(n)$. This belief plays a key role in the study of lattice-based cryptography.

In particular, descendants of Regev's original public-key encryption scheme are nearing widespread use in practice. One such scheme was even recently standardized by NIST~\cite{ABD+CRYSTALSKyberVersion022021,NIST2022}, with the goal of using this scheme as a replacement for the number-theoretic cryptography that is currently used for nearly all secure communication.\footnote{The number-theoretic cryptography that is currently in use is known to be broken by a sufficiently large quantum computer. In contrast, lattice-based cryptography is thought to be secure not only against classical computers, but also against quantum computers, which is why it has been standardized. See~\cite{NIST2022} for more discussion.}  In practice, these schemes rely for their security not only on the polynomial-time hardness of $\SVP$, but on \emph{very} precise assumptions about the hardness of $\gamma$-$\SVP$ as a function of $\gamma$. (E.g., the authors of~\cite{ABD+CRYSTALSKyberVersion022021} rely on sophisticated simulators that attempt to predict the optimal behavior of heuristic $\gamma$-$\SVP$ algorithms, which roughly tell us that $n^k$-$\SVP$ cannot be solved in time much better than $2^{0.29 n/(2k+1)}$ for constant $k \geq 0$.)

Therefore, we are now more interested in the \emph{fine-grained, superpolynomial} complexity of $\gamma$-$\SVP$ and $\gamma$-$\CVP$. I.e., we are not just interested in what is possible in polynomial time, but rather we are interested in precisely what is possible with different superpolynomial running times, with a particular emphasis on algorithms that run in $2^{Cn}$ time for different constants $C$. And, the specific approximation factor really matters quite a bit, as the running time $2^{C_\gamma n}$ of the best known $\gamma$-$\SVP$ algorithms  for polynomial approximation factors $\gamma = \poly(n)$ depends quite a bit on the specific polynomial $\gamma$. (This is true both for heuristic algorithms and those with proven correctness. E.g., the best known proven running time for approximation factor $n^c$ is roughly $2^{O(n/(c+1))}$ for constant $c \geq 0$. See~\cite{ALSTimeAlgorithmSqrtn2021} for the current state of the art.)

Indeed, a recent line of work has extended \emph{some} of the seminal polynomial-time results described above to the fine-grained superpolynomial setting~\cite{bennettQuantitativeHardnessCVP2017,ASGapETHHardness2018,aggarwalNoteConcreteHardness2021,ABGSFinegrainedHardnessCVP2021,bennettImprovedHardnessBDD2022}. Specifically, these works show exponential-time lower bounds for $\SVP$ and $\CVP$, both in their exact versions with $\gamma = 1$ and for small constant approximation factors $\gamma = 1+\eps$ (under suitable variants of the Exponential Time Hypothesis). These results can be viewed as fine-grained generalizations of Ajtai's original hardness result for $\SVP$ (or, perhaps, of the subsequent results that showed hardness for small approximation factors, such as~\cite{caiApproximatingSVPFactor1999,Mic01svp}), and they provide theoretical evidence in favor of the important cryptographic assumption that (suitable) lattice-based cryptography cannot be broken in $2^{o(n)}$ time.

However, there are no known non-trivial generalizations of the other major results listed above to the regime of superpolynomial running times. For example, (in spite of much effort) it is not known how to extend the above fine-grained hardness results to show exponential-time lower bounds for approximation factors $\gamma$ substantially larger than one---say, e.g., large constants $\gamma$ (let alone the polynomial approximation factors that are relevant to cryptography)---in analogy with the celebrated hardness of approximation results that are known against polynomial-time algorithms. And, prior to this work, it was also not known how to extend the worst-case to average-case reductions and protocols mentioned above to the superpolynomial setting in a non-trivial way (i.e., in a way that improves upon the approximation factor).

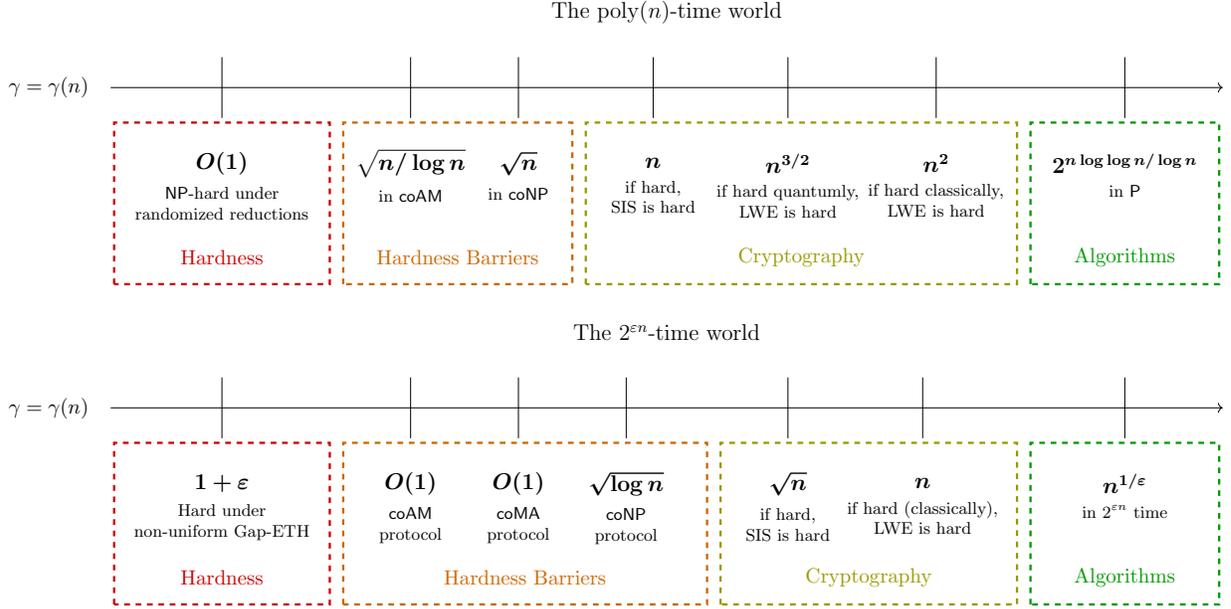
\begin{figure}[t]
    \centering
    {\renewcommand{\arraystretch}{3}
    \begin{tabular}{c}
\resizebox{\textwidth}{!}{
\begin{tikzpicture}

\pgfmathsetmacro{\boxheight}{3}
\pgfmathsetmacro{\boxspacing}{0.25}

\pgfmathsetmacro{\factorspacing}{3}
\pgfmathsetmacro{\smallsidespacing}{1.25}
\pgfmathsetmacro{\largesidespacing}{1.75}
\pgfmathsetmacro{\factorheight}{\boxheight - 0.75}

\pgfmathsetmacro{\labheight}{0.5}

\pgfmathsetmacro{\hardnessboxwidth}{2*2 + (1 - 1)*\factorspacing}
\pgfmathsetmacro{\protocolsboxwidth}{2*\smallsidespacing + 1.75}
\pgfmathsetmacro{\cryptographyboxwidth}{2*\smallsidespacing + 5.5}
\pgfmathsetmacro{\algorithmsboxwidth}{2*\largesidespacing + (1 - 1)*\factorspacing}

\pgfmathsetmacro{\hardnessshift}{0}
\pgfmathsetmacro{\protocolsshift}{\hardnessshift + \hardnessboxwidth + \boxspacing}
\pgfmathsetmacro{\cryptographyshift}{\protocolsshift + \protocolsboxwidth + \boxspacing}
\pgfmathsetmacro{\algorithmsshift}{\cryptographyshift + \cryptographyboxwidth + \boxspacing}

\pgfmathsetmacro{\totalwidth}{\hardnessboxwidth+\protocolsboxwidth+\cryptographyboxwidth+\algorithmsboxwidth+3*\boxspacing}

\definecolor{darkred}{rgb}{0.8,0,0}
\definecolor{darkorange}{RGB}{204,102,0}
\definecolor{darkyellow}{RGB}{153,153,0}
\definecolor{darkgreen}{rgb}{0,0.6,0}

\begin{scope}[darkred]
\draw[dashed,very thick] (0, 0) rectangle (\hardnessboxwidth, \boxheight);
\node at (0.5*\hardnessboxwidth, \labheight) (hardness-box) {Hardness};

\begin{scope}[black]
\footnotesize
\node (NP-hardness-factor) at (0.5*\hardnessboxwidth, \factorheight) {\large $\bm{O(1)}$};
\node[below = 0 of NP-hardness-factor, align=center] (NP-hardness-label) {$\NP$-hard under\\randomized reductions};
\end{scope}
\end{scope}

\begin{scope}[shift = {(\protocolsshift, 0)}, darkorange]
\draw[dashed,very thick] (0, 0) rectangle (\protocolsboxwidth, \boxheight);
\node at (0.5*\protocolsboxwidth, \labheight) {Hardness Barriers};
\pgfmathsetmacro{\smallshift}{1}

\begin{scope}[black]
\footnotesize
\node (coAM-factor) at (\smallsidespacing, \factorheight) {\large $\bm{\sqrt{n/\log n}}$};
\node[below = 0 of coAM-factor] (coAM-label) {in $\coAM$};

\node (coNP-factor) at (\smallsidespacing + 2, \factorheight) {\large $\bm{\sqrt{n}}$};
\node[below = 0 of coNP-factor] (coNP-label) {in $\coNP$};
\end{scope}
\end{scope}

\begin{scope}[shift = {(\cryptographyshift, 0)}, darkyellow]
\draw[dashed,very thick] (0, 0) rectangle (\cryptographyboxwidth, \boxheight);
\node at (0.5*\cryptographyboxwidth, \labheight) {Cryptography};

\begin{scope}[black]
\footnotesize
\node (SIS-factor) at (\smallsidespacing , \factorheight) {\large $\bm{n}$};
\node[below = 0 of SIS-factor,align=center] (SIS-label) {if hard,\\ $\SIS$ is hard};

\node (qLWE-factor) at (\smallsidespacing + 2.5, \factorheight) {\large $\bm{n^{3/2}}$};
\node[below = 0 of qLWE-factor,align=center] (qLWE-label) {if hard quantumly,\\ $\LWE$ is hard};

\node (cLWE-factor) at (\smallsidespacing + 5.25, \factorheight) {\large $\bm{n^2}$};
\node[below = 0 of cLWE-factor,align=center] (cLWE-label) {if hard classically,\\ $\LWE$ is hard};
\end{scope}
\end{scope}

\begin{scope}[shift = {(\algorithmsshift, 0)}, darkgreen]
\draw[dashed,very thick] (0, 0) rectangle (\algorithmsboxwidth, \boxheight);
\node at (0.5*\algorithmsboxwidth, \labheight) {Algorithms};

\begin{scope}[black]
\footnotesize
\node (BKZ-factor) at (0.5*\algorithmsboxwidth, \factorheight) 
{\large $\bm{2^{n \log \log n/\log n}}$};
\node[below = 0cm of BKZ-factor,align=center] (BKZ-label) {in $\P$};
\end{scope}
\end{scope}

\foreach \factor in {NP-hardness-factor, coAM-factor, coNP-factor, SIS-factor, qLWE-factor, cLWE-factor, BKZ-factor} {
\node (above\factor) at ($(\factor)+(0,0.7)$) {};
\node (wayabove\factor) at ($(\factor)+(0,2.1)$) {};
\draw (above\factor) -- (wayabove\factor);
}

\node (left-of-arrow) at (-0.2, \factorheight+1.4) {};
\node (right-of-arrow) at (\totalwidth + 0.2, \factorheight+1.4) {};
\draw[->] (left-of-arrow) -- (right-of-arrow);
\node[left = 0cm of left-of-arrow] (arrow-label) {$\gamma = \gamma(n)$};

\node at (0.5*\totalwidth, \factorheight+2.8) {\large The $\poly(n)$-time world};
\end{tikzpicture}}\\
\resizebox{\textwidth}{!}{
\begin{tikzpicture}

\pgfmathsetmacro{\boxheight}{3}
\pgfmathsetmacro{\boxspacing}{0.25}

\pgfmathsetmacro{\factorspacing}{3}
\pgfmathsetmacro{\smallsidespacing}{1.25}
\pgfmathsetmacro{\largesidespacing}{1.75}
\pgfmathsetmacro{\factorheight}{\boxheight - 0.75}

\pgfmathsetmacro{\labheight}{0.5}

\pgfmathsetmacro{\hardnessboxwidth}{2*2 + (1 - 1)*\factorspacing}
\pgfmathsetmacro{\protocolsboxwidth}{2*\smallsidespacing + 4.25}
\pgfmathsetmacro{\cryptographyboxwidth}{2*\smallsidespacing + 3}
\pgfmathsetmacro{\algorithmsboxwidth}{2*\largesidespacing + (1 - 1)*\factorspacing}

\pgfmathsetmacro{\hardnessshift}{0}
\pgfmathsetmacro{\protocolsshift}{\hardnessshift + \hardnessboxwidth + \boxspacing}
\pgfmathsetmacro{\cryptographyshift}{\protocolsshift + \protocolsboxwidth + \boxspacing}
\pgfmathsetmacro{\algorithmsshift}{\cryptographyshift + \cryptographyboxwidth + \boxspacing}

\pgfmathsetmacro{\totalwidth}{\hardnessboxwidth+\protocolsboxwidth+\cryptographyboxwidth+\algorithmsboxwidth+3*\boxspacing}

\definecolor{darkred}{rgb}{0.8,0,0}
\definecolor{darkorange}{RGB}{204,102,0}
\definecolor{darkyellow}{RGB}{153,153,0}
\definecolor{darkgreen}{rgb}{0,0.6,0}

\begin{scope}[darkred]
\draw[dashed,very thick] (0, 0) rectangle (\hardnessboxwidth, \boxheight);
\node at (0.5*\hardnessboxwidth, \labheight) (hardness-box) {Hardness};

\begin{scope}[black]
\footnotesize
\node (NP-hardness-factor) at (0.5*\hardnessboxwidth, \factorheight) {\large $\bm{1 + \eps}$};
\node[below = 0 of NP-hardness-factor, align=center] (NP-hardness-label) {Hard under\\non-uniform Gap-ETH};
\end{scope}
\end{scope}

\begin{scope}[shift = {(\protocolsshift, 0)}, darkorange]
\draw[dashed,very thick] (0, 0) rectangle (\protocolsboxwidth, \boxheight);
\node at (0.5*\protocolsboxwidth, \labheight) {Hardness Barriers};
\pgfmathsetmacro{\smallshift}{1}

\begin{scope}[black]
\footnotesize
\node (coAM-factor) at (\smallsidespacing, \factorheight) {\large $\bm{O(1)}$};
\node[below = 0 of coAM-factor,align=center] (coAM-label) {$\coAM$\\protocol};

\node (coMA-factor) at (\smallsidespacing + 2, \factorheight) {\large $\bm{O(1)}$};
\node[below = 0 of coMA-factor,align=center] (coMA-label) {$\cc{coMA}$\\protocol};

\node (coNP-factor) at (\smallsidespacing + 4, \factorheight) {\large $\bm{\sqrt{\log n}}$};
\node[below = 0 of coNP-factor,align=center] (coNP-label) {$\coNP$\\protocol};
\end{scope}
\end{scope}

\begin{scope}[shift = {(\cryptographyshift, 0)}, darkyellow]
\draw[dashed,very thick] (0, 0) rectangle (\cryptographyboxwidth, \boxheight);
\node at (0.5*\cryptographyboxwidth, \labheight) {Cryptography};

\begin{scope}[black]
\footnotesize
\node (SIS-factor) at (\smallsidespacing , \factorheight) {\large $\bm{\sqrt{n}}$};
\node[below = 0 of SIS-factor,align=center] (SIS-label) {if hard,\\ $\SIS$ is hard};

\node (cLWE-factor) at (\smallsidespacing + 2.5, \factorheight) {\large $\bm{n}$};
\node[below = 0 of cLWE-factor,align=center] (cLWE-label) {if hard (classically),\\ $\LWE$ is hard};
\end{scope}
\end{scope}

\begin{scope}[shift = {(\algorithmsshift, 0)}, darkgreen]
\draw[dashed,very thick] (0, 0) rectangle (\algorithmsboxwidth, \boxheight);
\node at (0.5*\algorithmsboxwidth, \labheight) {Algorithms};

\begin{scope}[black]
\footnotesize
\node (BKZ-factor) at (0.5*\algorithmsboxwidth, \factorheight) 
{\large $\bm{n^{1/\eps}}$};
\node[below = 0cm of BKZ-factor,align=center] (BKZ-label) {in $2^{\eps n}$ time};
\end{scope}
\end{scope}

\foreach \factor in {NP-hardness-factor, coAM-factor, coMA-factor, coNP-factor, SIS-factor, cLWE-factor, BKZ-factor} {
\node (above\factor) at ($(\factor)+(0,0.7)$) {};
\node (wayabove\factor) at ($(\factor)+(0,2.1)$) {};
\draw (above\factor) -- (wayabove\factor);
}

\node (left-of-arrow) at (-0.2, \factorheight+1.4) {};
\node (right-of-arrow) at (\totalwidth + 0.2, \factorheight+1.4) {};
\draw[->] (left-of-arrow) -- (right-of-arrow);
\node[left = 0cm of left-of-arrow] (arrow-label) {$\gamma = \gamma(n)$};

\node at (0.5*\totalwidth, \factorheight+2.8) {\large The $2^{\eps n}$-time world};
\end{tikzpicture}}
    \end{tabular}}
    \caption{\small This figure shows the current state of the art of the complexity of $\gamma$-$\SVP$ for different approximation factors $\gamma$ in two different regimes. The top row shows polynomial-time results (polynomial-time hardness, protocols, worst-case to average-case reductions, and algorithms, respectively). The bottom row shows $2^{\eps n}$-time results. Note that the scales are rather extreme, and are certainly not the same in the two rows. The hardness barriers and cryptography results in the bottom row are the five new results in this paper. We have omitted some constants for simplicity. (This figure is based on a similar one appearing in~\cite{bennett/svp-survey22}.) \label{fig:svp-complexity}}
\end{figure}

\subsection{Our results}

At a high level, our results can be stated quite succinctly. We generalize to the superpolynomial setting (1) Ajtai's worst-case to average-case reduction for secret-key cryptography; (2) Regev's worst-case to average-case quantum reduction for public-key cryptography and Peikert's classical version; (3) Goldreich and Goldwasser's $\mathsf{coAM}$ protocol; and (4) Aharonov and Regev's $\mathsf{coNP}$ protocol. In all of these results, in the important special case when the reductions or protocols are allowed to run in $2^{\eps n}$ time, we improve upon the polynomial-time approximation factor by a factor of roughly $\sqrt{n/\log n}$ (and a factor of $\widetilde{O}(n)$ for Peikert's classical worst-case to average-case reduction). We also show a novel $\mathsf{coMA}$ protocol that has no known analog in the polynomial-time regime. 

See \cref{fig:svp-complexity} for a diagram showing the current state of the art for both the polynomial-time regime and the $2^{\eps n}$-time regime for arbitrarily small constants $\eps > 0$. Below, we describe the results in more detail and explain their significance. We describe the protocols first, as our worst-case to average-case reductions are best viewed in the context of our protocols.

\subsubsection{Protocols for lattice problems}

\paragraph{A $\mathsf{coAM}$ protocol.} Our first main result is a generalization of Goldreich and Goldwasser's $\mathsf{coAM}$ protocol, as follows.

\begin{theorem}[Informal, see \cref{sec:coAM}]
    \label{thm:coAM_intro}
    For every $\gamma = \gamma(n) \geq 1$, there is a $\coAM$ protocol for $\gamma$-$\GapCVP$ running in time $2^{O(n/\gamma^2)}$.
    
    Furthermore, for every constant $\eps > 0$, there exists a $\delta > 0$ such that there is a two-round private-coin (honest-verifier perfect zero knowledge) protocol for $(\sqrt{2} + \eps)$-$\coGapCVP$ running in time $2^{(1/2-\delta) n}$.
\end{theorem}

See \cref{sec:coAM} for the precise result, which is also more general in that it also applies to arbitrary norms $\|\cdot \|_K$ (with different constants), just like the original theorem of \cite{GG00}.

This theorem is a strict generalization of the original polynomial-time result of Goldreich and Goldwasser~\cite{GG00}. And, just like~\cite{GG00} was viewed as a barrier to proving polynomial-time hardness results for approximation factors $\gamma \geq \sqrt{n/\log n}$, our result can be viewed as a barrier to proving superpolynomial hardness for smaller approximation factors $\gamma$. In particular, the theorem rules out the possibility of using a fine-grained reduction from $k$-$\mathrm{SAT}$ to prove, e.g., $2^{\Omega(n)}$ hardness for large constants $\gamma$ or $2^{(1-\eps)n}$-time hardness for any constant $\gamma > \sqrt{2}$ (assuming AMETH and IPSETH respectively, in a sense that is made precise in \cref{sec:limitations}).\footnote{It might seem strange that we describe a roughly $2^{n/2}$-time protocol as ruling out roughly $2^n$ hardness. This is because $k$-$\mathrm{coSAT}$ is known to have a roughly $2^{n/2}$-time two-round \emph{protocol}~\cite{W16} (and even an $\mathsf{MA}$ protocol), but is not known to have a $2^{(1-\eps)n}$-time \emph{algorithm} (for sufficiently large $k$). The assumption that $k$-$\mathrm{SAT}$ has no $2^{(1-\eps)n}$-time protocol for sufficiently large $k$ is called SETH, while the assumption that $k$-$\mathrm{coSAT}$ does not have a $2^{(1/2 - \eps)n}$-time two-round protocol for sufficiently large $k$ is called IPSETH. So, to prove $2^{(1-\eps) n}$-time hardness of $\gamma$-$\CVP$ under SETH, it would suffice to give a ($2^{\eps n}$-time, Turing) reduction from $k$-$\mathrm{SAT}$ on $n$ variables to $\gamma$-$\CVP$ on a lattice with rank $n + o(n)$. But, for constant $\gamma > \sqrt{2}$, such a reduction together with \cref{thm:coAM_intro} would imply a significantly faster protocol for $k$-$\mathrm{coSAT}$ than what is currently known, and would therefore violate IPSETH. See \cref{sec:prelims_ETH_stuff} for more discussion of fine-grained complexity and related hypotheses and \cref{sec:limitations} for formal proofs ruling out such reductions under various hypotheses.} We place a particular emphasis on the running time of $2^{(1-\eps)n}$ because (1) the fastest known algorithm for CVP runs in time $2^{n+o(n)}$; and (2) we know a $2^{(1-\eps)n}$-time lower bound for $(1+\eps')$-$\CVP$~\cite{bennettQuantitativeHardnessCVP2017,ABGSFinegrainedHardnessCVP2021} (under variants of SETH---though, admittedly, only in $\ell_p$ norms where $p$ is not an even integer, so not for the $\ell_2$ norm). Therefore, this protocol provides an explanation for why fine-grained hardness results for $\CVP$ are stuck at small constant approximation factors. See \cref{sec:limitations} for a precise discussion of these barriers to proving hardness and their relationship to known hardness results.

As we explain in more detail in \cref{sec:techniques_AM}, our protocol is a very simple and natural generalization of the original beautiful protocol due to Goldreich and Goldwasser. And, as we explain below, the same simple ideas behind this protocol are also used in our worst-case to average-case reduction for $\LWE$.

\paragraph{A co-non-deterministic
protocol.} Our second main result is a variant of Aharonov and Regev's $\mathsf{coNP}$ protocol for $\sqrt{n}$-$\CVP$, as follows.

\begin{theorem}[Informal, see \cref{thm:coNP}]
    \label{thm:coNP_intro}
    For every $\gamma = \gamma(n) \geq 1$, there is a co-non-deterministic protocol for $\gamma$-$\gapSVP$ that runs in time $n^{O( n/\gamma^2)}$. In particular, there is a $2^{\eps n}$-time protocol for $O_\eps(\sqrt{\log n})$-$\gapSVP$.
\end{theorem}

This result is almost a strict generalization of~\cite{aharonovLatticeProblemsNP2005}, except that Aharonov and Regev's protocol works for $\CVP$, while ours only works for $\SVP$. 

Again, this result can be viewed as a barrier to proving hardness of $\gamma$-$\SVP$ (assuming NETH; see \cref{sec:prelims_ETH_stuff,sec:limitations}). And, just like how~\cite{aharonovLatticeProblemsNP2005} gives a stronger barrier against proving polynomial-time hardness than~\cite{GG00} (collapse of the polynomial hierarchy to the first level, as opposed to the second) at the expense of a larger approximation factor $\gamma$, our \cref{thm:coNP_intro} gives a stronger barrier against proving superpolynomial hardness (formally, a barrier assuming NETH rather than AMETH) than \cref{thm:coAM_intro}, at the expense of a larger approximation factor. See \cref{sec:limitations}.

As we discuss more in \cref{sec:techniques_NP}, our protocol is broadly similar to the original protocol in \cite{aharonovLatticeProblemsNP2005}, but the details and the analysis are quite different---requiring in particular careful control over the higher moments of the discrete Gaussian distribution.

 We note that we originally arrived at this protocol in an attempt to solve a different (and rather maddening) open problem. In~\cite{aharonovLatticeProblemsNP2005}, Aharonov and Regev speculated that their protocol could be improved to achieve an approximation factor of $\sqrt{n/\log n}$ rather than $\sqrt{n}$, therefore matching in $\mathsf{coNP}$ the approximation factor achieved by~\cite{GG00} in $\mathsf{coAM}$. And, there is a certain sense in which they came tantalizingly close to achieving this (as we explain in \cref{sec:techniques_NP}). It has therefore been a long-standing open problem to close this $\sqrt{\log n}$ gap. 

We have \emph{not} successfully closed this gap between~\cite{aharonovLatticeProblemsNP2005} and~\cite{GG00}. Indeed, for all running times, the approximation factor in \cref{thm:coNP_intro} remains stubbornly larger than that in \cref{thm:coAM_intro} by a factor of $\sqrt{\log n}$, so that in some sense the gap persists even into the superpolynomial-time regime! But, we \emph{do} show that a suitable modification of the Aharonov and Regev $\mathsf{coNP}$ protocol can achieve approximation factors less than $\sqrt{n}$, at the expense of more running time. This in itself is already quite surprising, as the analysis in~\cite{aharonovLatticeProblemsNP2005} seems in some sense tailor-made for the approximation factor $\sqrt{n}$ and no lower. For example, prior to our work, it was not even clear how to  achieve an approximation factor of, say, $\sqrt{n}/10$ in co-non-deterministic time less than it takes to simply solve the problem deterministically. We show how to achieve, e.g., an approximation factor of $\sqrt{n}/C$ for any constant $C$ in polynomial time.

\paragraph{A $\mathsf{coMA}$ protocol.} Our third main result is a $\mathsf{coMA}$ protocol for $\CVP$, as follows.

\begin{theorem}[Informal; see \cref{thm:coMA}]
    \label{thm:coMA_intro}
  There is a $\mathsf{coMA}$ protocol for $\gamma$-$\gapCVP$ that runs in time $2^{O(n/\gamma)}$. In particular, there is a $2^{\eps n}$-time protocol for $O_\eps(1)$-$\gapCVP$.
\end{theorem}

Unlike our other protocols, the protocol in \cref{thm:coMA_intro} has no known analog in prior work. Indeed, the result is only truly interesting for running times larger than roughly $2^{\sqrt{n}}$, since for smaller running times it is completely subsumed by \cite{aharonovLatticeProblemsNP2005}. It is therefore unsurprising that this result was not discovered by prior work that focused on the polynomial-time regime.

This protocol too can be viewed as partial progress towards improving the approximation factor achieved by~\cite{aharonovLatticeProblemsNP2005} by a factor of $\sqrt{\log n}$. In particular, notice that in the important special case of $2^{\eps n}$ running time, the approximation factor achieved in \cref{thm:coMA_intro} is better than that achieved by \cref{thm:coNP_intro} by a $\sqrt{\log n}$ factor. (Indeed, since the approximation factor is constant in this case, it is essentially the best that we can hope for.) So, in the $2^{\eps n}$-time world, there is no significant gap between the approximation factors that we know how to achieve in $\mathsf{coMA}$ and $\mathsf{coAM}$, in contrast to the polynomial-time world.

As a barrier to proving exponential-time hardness of lattice problems, the $\mathsf{coMA}$ protocol in \cref{thm:coMA_intro} lies between the co-non-deterministic protocol in \cref{thm:coNP_intro} and the $\mathsf{coAM}$ protocol in \cref{thm:coAM_intro}, since a co-non-deterministic protocol implies a $\mathsf{coMA}$ protocol, which implies a $\mathsf{coAM}$ protocol (though at the expense of a constant factor in the exponent of the running time; see \cref{sec:prelims_ETH_stuff}). In particular, for $2^{\eps n}$ running time, the approximation factor is (significantly) better than \cref{thm:coNP_intro} but (just slightly) worse than \cref{thm:coAM_intro}. But, the complexity-theoretic assumption needed to rule out hardness in this case (MAETH) is weaker than for \cref{thm:coAM_intro} (AMETH) but stronger than for \cref{thm:coNP_intro} (NETH).

In fact, our $\mathsf{coMA}$ protocol is perhaps best viewed as a ``mixture'' of the two beautiful protocols from~\cite{GG00} and~\cite{aharonovLatticeProblemsNP2005}. As we explain in \cref{sec:techniques_MA}, we think of this $\mathsf{coMA}$ protocol as taking the best parts from~\cite{GG00} and~\cite{aharonovLatticeProblemsNP2005}, and we therefore view the resulting ``hybrid'' protocol as quite natural and elegant.

\subsubsection{Worst-case to average-case reductions}

\paragraph{Worst-case to average-case reductions for $\SIS$.} Our fourth main result is a generalization beyond polynomial time of (Micciancio and Regev's version of) Ajtai's worst-case to average-case reduction, as follows.

\begin{theorem}[Informal; see \cref{thm:SIS}]
    \label{thm:SIS_intro}
    For any $\gamma = \gamma(n) \geq 1$, there is a reduction from $\gamma$-$\gapSVP$ to $\SIS$ that runs in time $2^{n^2 \cdot \polylog(n)/\gamma^2}$. In particular, (exponentially secure) secret-key cryptography exists if $\widetilde{O}(\sqrt{n})$-$\gapSVP$ is $2^{\Omega(n)}$ hard.
\end{theorem}

This is a strict generalization of the previous state of the art, i.e., the main result in~\cite{MR04}, which only worked in the polynomial-time regime, i.e., for $\gamma = \widetilde{\Theta}(n)$. 
(In fact, our reduction is also a generalization of the reduction due to Micciancio and Peikert~\cite{micciancioHardnessSISLWE2013}, which itself generalizes~\cite{MR04} to more parameter regimes. Specifically, our result holds in the ``small modulus'' regime, like that of \cite{micciancioHardnessSISLWE2013}. But, in this high-level description where we have not even defined the modulus, we ignore this important distinction.)

We are particularly interested in the special case of our reduction for $\gamma = \widetilde{\Theta}(\sqrt{n})$. Indeed, as we mentioned earlier, it is widely believed that $\gamma$-$\SVP$ is $2^{\Omega(n)}$ hard for \emph{any} approximation factor $\gamma \leq \poly(n)$, and even stronger assumptions are commonly made in the literature on lattice-based cryptography (both in theoretical and practical work---and even in work outside of lattice-based cryptography~\cite{BSVHardnessAveragecaseSUM2021}). Therefore, we view the assumption that $\widetilde{O}(\sqrt{n})$-$\SVP$ is $2^{\Omega(n)}$ hard to be quite reasonable in this context. Indeed, if one assumes (as is common in the cryptographic literature) that the best known (heuristic) algorithms for $\gamma$-$\SVP$ are essentially optimal, then this result implies significantly better security for lattice-based cryptography than other worst-case to average-case reductions.

In fact, \cref{thm:SIS_intro} follows from an improvement to just one step in Micciancio and Regev's reduction. Specifically, to achieve the best possible approximation factor, Micciancio and Regev essentially used their $\SIS$ oracle to generate the witness used in Aharonov and Regev's $\mathsf{coNP}$ protocol.\footnote{There are simpler ways to use a $\SIS$ oracle to solve $\gapSVP$ that achieve a worse approximation factor---e.g., by using $\mathrm{SIVP}$ as an intermediate problem. But Micciancio and Regev's clever use of the~\cite{aharonovLatticeProblemsNP2005} protocol yields the $\widetilde{O}(n)$ approximation factor that has remained the state of the art since a preliminary version of~\cite{MR04} was published in 2004.} Our generalization of Aharonov and Regev's protocol uses (a larger version of) the same witness, so that we almost get our generalization of~\cite{MR04} for free once we have generalized~\cite{aharonovLatticeProblemsNP2005}. There are, however, many technical details to work out, as we describe in~\cref{sec:techniques_SIS}.

(To get the best approximation factor that we can, we actually use our $\mathsf{coMA}$ protocol in some parameter regimes and our co-non-deterministic protocol in others. This works similarly because the witness is the same for the two protocols.)

\paragraph{Worst-case to average-case reductions for $\LWE$.} Our fifth and final main result is a generalization of both Regev's quantum worst-case to average-case reduction for $\LWE$~\cite{regevLatticesLearningErrors2009} and Peikert's classical version~\cite{peikertPublickeyCryptosystemsWorstcase2009}. Since $\LWE$ comes with many parameters, in this high-level overview we simply present the special case of the result for the hardest choice of parameters that is known to imply public-key encryption.

\begin{theorem}[Informal; see \cref{thm:LWE_classical,thm:LWE_quantum}]
    \label{thm:LWE_intro}
  For any $\gamma = \gamma(n) \geq 1$, public-key encryption exists if $\gamma$-$\SVP$ is $2^{n^2 \polylog(n) /\gamma}$ hard for a classical computer or $2^{n^3 \polylog(n)/\gamma^2}$ hard for a quantum computer. In particular, (exponentially secure) public-key cryptography exists if $\widetilde{O}(n)$-$\gapSVP$ is $2^{\Omega(n)}$ hard, even for a classical computer.
\end{theorem}

Again, this is a strict generalization of the prior state of the art, which matched the above result for polynomial running time. And, again, we stress that $2^{\Omega(n)}$-hardness of $\widetilde{O}(n)$-$\gapSVP$ is a widely believed conjecture. Indeed, if one assumes (as is common in the cryptographic literature) that the best known (heuristic) algorithms for $\gamma$-$\SVP$ are essentially optimal, then this result implies significantly better security for lattice-based public-key cryptography than prior worst-case to average-case reductions.

In particular, notice that in the important special case of running time $2^{\eps n}$, our quantum reduction and classical reduction achieve essentially the same approximation factor. (Indeed, they differ by only a constant factor.) So, perhaps surprisingly, there is no real gap between classical and quantum reductions in the exponential-time regime, unlike in the polynomial-time regime. 

We note that behind this result is a new generalization of the polynomial-time reduction from $\gapSVP$ to the Bounded Distance Decoding problem ($\BDD$). This polynomial-time reduction was implicit in~\cite{peikertPublickeyCryptosystemsWorstcase2009} and made explicit in~\cite{conf/crypto/LyubashevskyM09}, and it can be viewed as a version of the~\cite{GG00} $\mathsf{coAM}$ protocol in which Merlin is simulated by a $\BDD$ oracle. We (of course!) generalize this by allowing the reduction to run in more time in order to achieve a better approximation factor, using the same ideas that we used to generalize the \cite{GG00} $\mathsf{coAM}$ protocol. (See \cref{sec:gapsvp-bdd}.) 

Furthermore, to obtain the best possible approximation factor in the classical result (and, in particular, an approximation factor that matches the quantum result in the $2^{\eps n}$-time setting), we also observe that Peikert's celebrated classical reduction from $\BDD$ to $\LWE$ can be made to work for a wider range of parameters if it is allowed to run in superpolynomial time. At a technical level, this involves combining basis reduction algorithms (e.g., from ~\cite{gamaFindingShortLattice2008}) with the discrete Gaussian sampling algorithm from~\cite{gentryTrapdoorsHardLattices2008,brakerskiClassicalHardnessLearning2013}. The resulting improved parameters results in a significant savings in the approximation factor, and even a small savings in the polynomial-time setting. (E.g., in the exponential-time setting, this saves us a factor of $\sqrt{n}$.)

Both of these observations follow relatively easily from combining known techniques. But, they might be of independent interest.

\subsection{Our techniques}
\label{sec:techniques}

\subsubsection{A \texorpdfstring{$\mathsf{coAM}$}{coAM} protocol}
\label{sec:techniques_AM}

At a high level, our $\mathsf{coAM}$ protocol uses the following very elegant idea due to Goldreich and Goldwasser~\cite{GG00}. Recall that our goal is to describe a protocol between all-powerful Merlin and computationally bounded Arthur in which Merlin (for whatever mysterious reason) wishes to convince Arthur that $\vec{t}$ is far from the lattice. In particular, if $\dist(\vec{t},\lat) > 2$ (the FAR case), Merlin should be able to convince Arthur that $\vec{t}$ is far from the lattice. On the other hand, if $\dist(\vec{t},\lat) \leq d$ (the CLOSE case, where $d < 2$ will depend on Arthur's running time), then even if all-powerful Merlin tries his best to convince Arthur that $\vec{t}$ is far from the lattice, Arthur should correctly determine that Merlin is trying to trick him with high probability.

To that end, consider the set
\[
    S_{\vec0} := \bigcup_{\vec{y} \in \lat} (\ball  + \vec{y} )
    \; ,
\]
which is the union of balls of radius $1$ centered around each lattice point, and the set
\[
    S_{\vec{t}} := \bigcup_{\vec{y} \in \lat} (\ball + \vec{y} - \vec{t}) = S_{\vec{0}} - \vec{t}
    \; ,
\]
which instead consists of balls centered around lattice points shifted by $\vec{t}$.
See \cref{figure:distributions}.

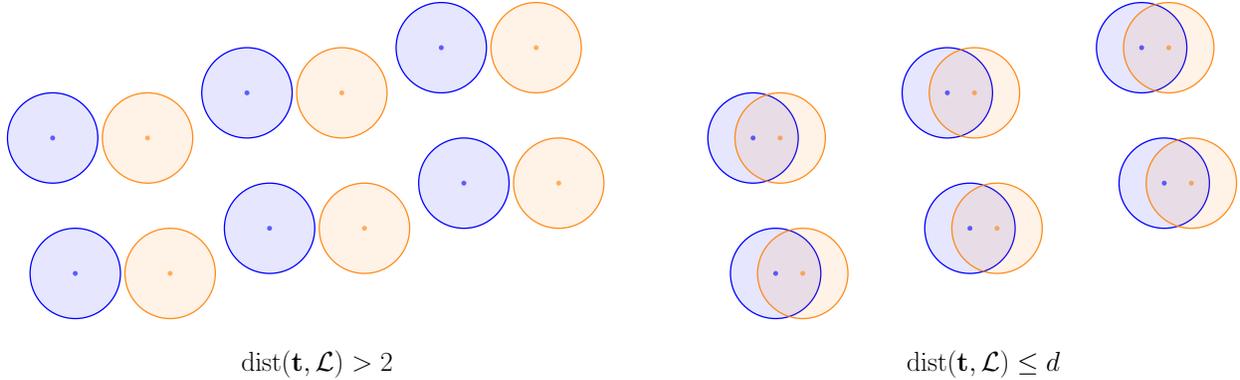
\begin{figure}
    \centering
    \resizebox{\textwidth}{!}{
    \begin{tikzpicture}[
    solidcircle/.style={circle, minimum size=2cm, draw=blue, thick, fill=blue, fill opacity=0.1},
    dashedcircle/.style={circle, minimum size=2cm, draw=orange, draw opacity=0.9, thick, fill=orange, fill opacity=0.1},
    point/.style={circle, fill=blue, scale=0.3, fill opacity=0.6},
    cross/.style={circle, fill=orange, scale=0.3, fill opacity=0.6}
    ]
    \coordinate (origin1) at (-2, 0);
    \coordinate (origin2) at (13.5, 0);
    \coordinate (shift1) at (2.1,0);
    \coordinate (shift2) at (0.6, 0);
    \coordinate (b1) at (4.3, 1);
    \coordinate (b2) at (-0.5, 3);

    \foreach\shift\origin\sometext in 
    {shift1/origin1/{\LARGE$\mathrm{dist}(\mathbf{t}, \mathcal{L}) > 2$}, 
    shift2/origin2/{\LARGE$\mathrm{dist}(\mathbf{t}, \mathcal{L}) \leq d$}}
    {
    \path let \p1 = (b1) in let \p2 = (\shift) in
    node at ($(\origin) + (\x1 + 0.5*\x2, -2)$) {\sometext};

    \foreach\a in {0, 1, 2}{
    \foreach\b in {0, 1}{
    \coordinate (current) at ($\a*(b1) + \b*(b2) + (\origin)$);
    \node[solidcircle] at (current)  {};
    \node[dashedcircle] at ($(current) + (\shift)$) {};
    \node[point] at (current) {};
    \node[cross] at ($(current) + (\shift)$) {};
    }}}

    \end{tikzpicture}
    }
    \caption{Comparison of the sets $S_{\vec{0}}$ and $S_{\vec{t}}$ in the FAR case and in the CLOSE case.
    }
    \label{figure:distributions}
\end{figure}

Notice that $\dist(\vec{t},\lat) > 2$ (i.e., the FAR case) if and only if $S_{\vec0}$ and $S_{\vec{t}}$ are disjoint (ignoring the distinction between open and closed balls). On the other hand, if $\dist(\vec{t},\lat) \leq d < 2$ (the CLOSE case), then the two sets must overlap, with more overlap if $d$ is smaller. Specifically, the intersection of the two sets will contain at least a
\[
    p_d \approx (1-d^2/4)^{n/2} \approx e^{-d^2 n}
\]
fraction of the total volume of $S_{\vec0}$. (See \cref{lem:normalized-l2-ball-intersection-vol} for the precise statement.)

So, Arthur first flips a coin. If it comes up heads, he samples a point $\vec{x} \sim S_{\vec0}$ uniformly at random from $S_{\vec0}$. Otherwise, he samples $\vec{x} \sim S_{\vec{t}}$.\footnote{In fact, there is no uniformly random distribution over $S_{\vec0}$ or $S_{\vec{t}}$, since they have infinite volume. In reality, we work with these sets \emph{reduced modulo the lattice}. But, in this high-level description, it is convenient to pretend to work with the sets themselves.} He then sends the result to Merlin. Arthur then simply asks Merlin ``was my coin heads or tails?'' In other words, Arthur asks whether $\vec{x}$ was sampled from $S_{\vec0}$ or $S_{\vec{t}}$. If we are in the FAR case where $\dist(\vec{t},\lat) > 2$, then Merlin (who, remember, is all powerful) will be able to unambiguously determine whether $\vec{x}$ was sampled from $S_{\vec0}$ or $S_{\vec{t}}$, since they are disjoint sets. On the other hand, if $\dist(\vec{t},\lat) \leq d$, then with probability at least $p_d$, $\vec{x}$ will lie in the intersection of the two sets. When this happens, even all-powerful Merlin can do no better than randomly guessing Arthur's coin. 

Arthur and Merlin can therefore play this game, say, $n/p_d$ times. If we are in the FAR case, then an honest Merlin will answer correctly every time, and Arthur will correctly conclude that $\vec{t}$ is far from the lattice. If we are in the CLOSE case, then no matter what Merlin does, he is likely to guess wrong at least once, in which case Arthur will correctly conclude that Merlin is trying to fool him.

This yields a \emph{private-coin} (honest-verifier perfect zero knowledge) protocol that runs in time roughly
$1/p_d \approx e^{d^2 n}$ for $\gamma$-$\CVP$ with $\gamma \approx 1/d$. Similarly to the polynomial-time setting, one can then use standard generic techniques to convert any private-coin protocol into a true public-coin, two-round protocol (i.e., a true $\mathsf{coAM}$ protocol), at the expense of increasing the constant in the exponent.

\subsubsection{A co-non-deterministic protocol}
\label{sec:techniques_NP}

Our co-non-deterministic protocol (as well as our $\mathsf{coMA}$ protocol) is based on the beautiful protocol of Aharonov and Regev~\cite{aharonovLatticeProblemsNP2005}. The key tools are the \emph{periodic Gaussian function} and the \emph{discrete Gaussian distribution}. For $\vec{x} \in \R^n$, we define
\[
    \rho(\vec{x}) := e^{-\pi \|\vec{x}\|^2}
    \; ,
\]
and for a lattice $\lat \subset \R^n$ and target vector $\vec{t} \in \R^n$, we extend this definition to the lattice coset $\lat - \vec{t}$ as
\[
    \rho(\lat - \vec{t}) := \sum_{\vec{y} \in \lat} \rho(\vec{y} - \vec{t})
    \; .
\]
We can then define the periodic Gaussian function as 
\[
    f(\vec{t}) := \frac{\rho(\lat - \vec{t})}{\rho(\lat)}
    \; .
\]
Very roughly speaking, we expect $f(\vec{t})$ to be a smooth approximation to the function $e^{-\pi \dist(\vec{t},\lat)^2}$, or at least to be relatively large when $\vec{t}$ is close to the lattice and relatively small when $\vec{t}$ is far from the lattice. See \cref{fig:periodic_gaussian}.

\begin{figure}
    \centering
    \includegraphics[width=0.4\textwidth]{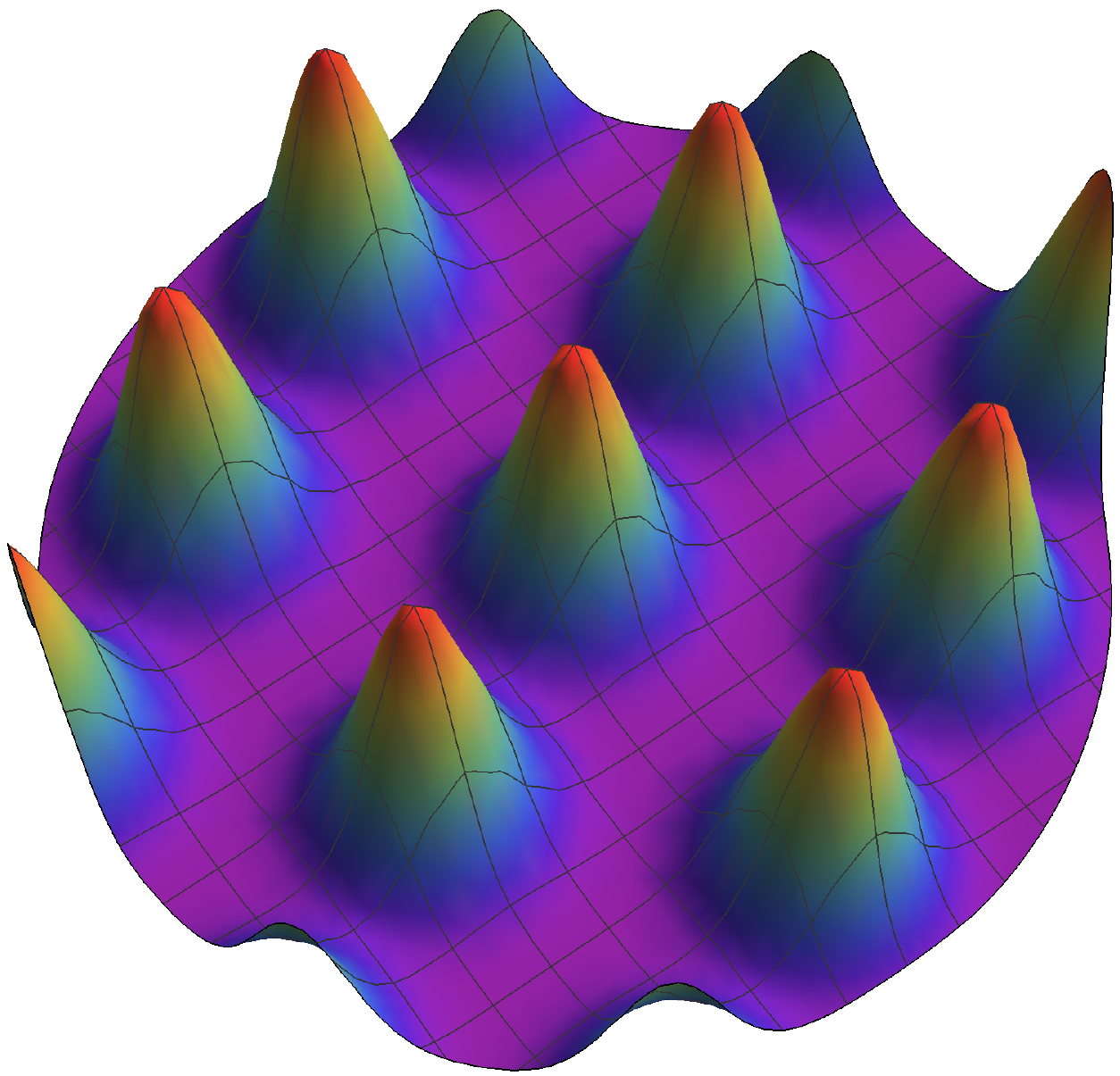} \qquad \qquad \includegraphics[width=0.4\textwidth]{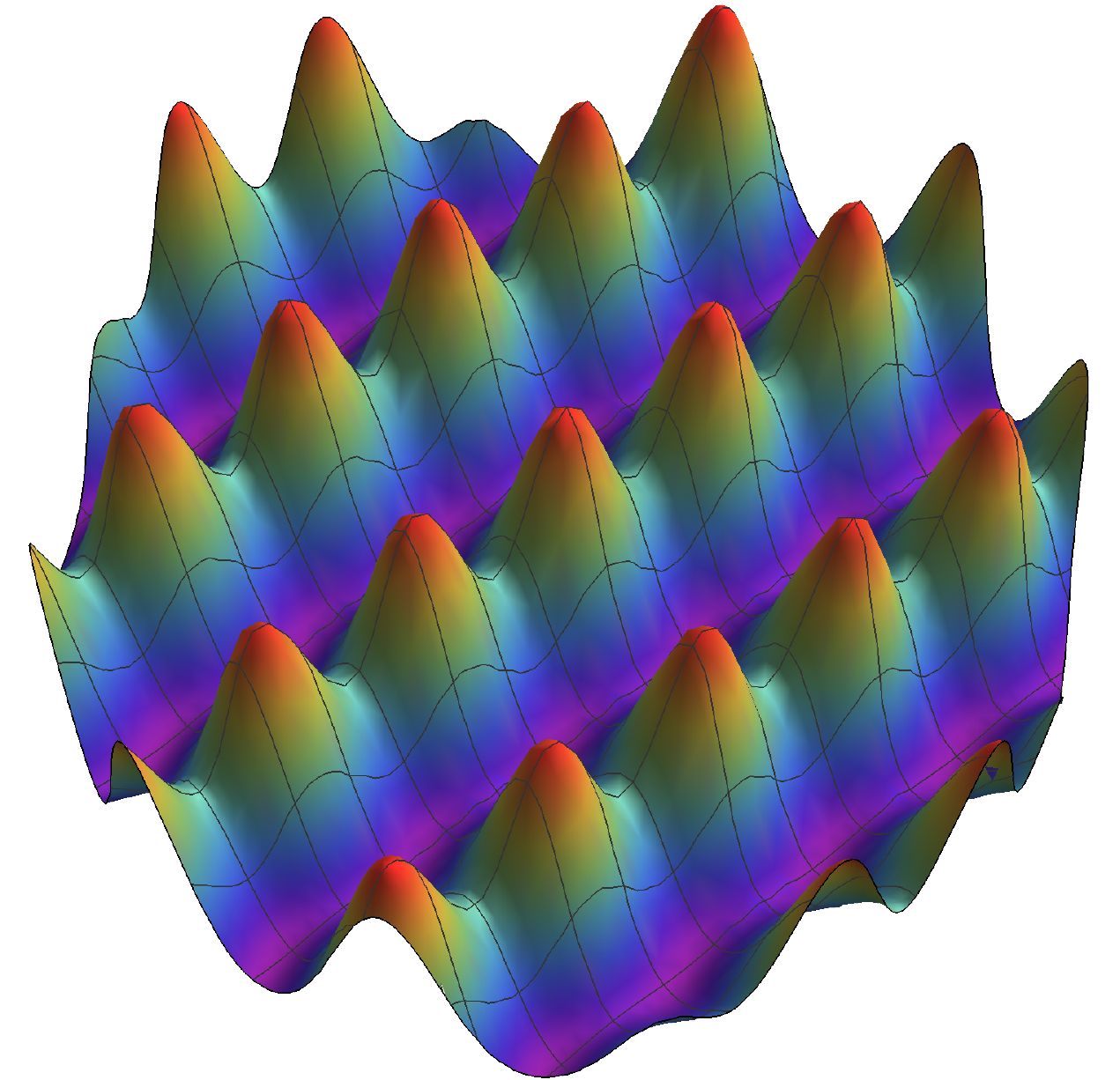}
    \caption{The periodic Gaussian function $f(\vec{t})$ for two different two-dimensional lattices $\lat$.}
    \label{fig:periodic_gaussian}
\end{figure}

Banaszczyk proved a number of important and beautiful results about the periodic Gaussian function~\cite{banaszczykNewBoundsTransference1993}. In particular, he showed that
\[
    e^{-\pi\, \dist(\vec{t},\lat)^2} \leq f(\vec{t}) \leq 1
     \; .
\]
So, if $\vec{t}$ is close to the lattice, then $f(\vec{t})$ cannot be too small. On the other hand, if $\dist(\vec{t},\lat) \geq \sqrt{n}$, then Banaszczyk proved that $f(\vec{t}) < 2^{-n}$. So, if we could somehow approximate $f(\vec{t})$ up to an additive error of $\delta \in (2^{-n},1)$, then we could distinguish between the case when $\dist(\vec{t},\lat) \lesssim \sqrt{\log(1/\delta)}$ and the case when $\dist(\vec{t}, \lat) \geq \sqrt{n}$, and therefore solve $\gamma$-$\CVP$ for $\gamma \approx \sqrt{n/\log(1/\delta)}$.

Of course, it is not immediately clear \emph{how} to approximate $f(\vec{t})$, even with additional help from an all-powerful prover. However, Aharonov and Regev observed that suitably chosen short vectors from the \emph{dual lattice} $\lat^*$ can be used for this purpose. Specifically, they recalled from the Poisson summation formula that
\begin{equation}
    \label{eq:PSF_intro}
    f(\vec{t}) = \expect_{\vec{w} \sim D_{\lat^*}}[\cos(2\pi \langle \vec{w},\vec{t} \rangle)]
    \; ,
\end{equation}
where $D_{\lat^*}$ is the \emph{discrete Gaussian distribution}, defined by
\[
   \Pr_{\vec{w} \sim D_{\lat^*}}[\vec{w} = \vec{z}] := \frac{\rho(\vec{z})}{\rho(\lat^*)}
   \; 
\]
for any $\vec{z} \in \lat^*$. So, Aharonov and Regev had the prover provide the verifier with $W := (\vec{w}_1,\ldots, \vec{w}_N)$ sampled independently from $D_{\lat^*}$. The verifier can then compute
\[
    f_W(\vec{t}) := \frac{1}{N} \sum_{i=1}^N \cos(2\pi \langle \vec{w}, \vec{t} \rangle)
    \; .
\]
I.e., $f_W$ is the sample approximation of \cref{eq:PSF_intro}.
By the Chernoff-Hoeffding bound, $f_W(\vec{t})$ will provide an approximation of $f(\vec{t})$ up to an error of roughly $1/\sqrt{N}$.
So, this \emph{almost} yields a roughly $N$-time non-deterministic protocol for distinguishing the FAR case when $\dist(\vec{t}, \lat) \geq \sqrt{n}$ from the CLOSE case when $\dist(\vec{t}, \lat) \lesssim \sqrt{\log N}$, i.e., a protocol for $\sqrt{n/\log N}$-$\CVP$.

The one (rather maddening) issue with this protocol is that it is not clear how to maintain soundness against a cheating prover in the case when $\vec{t}$ is close to the lattice. I.e., suppose that the prover provides vectors $W := (\vec{w}_1,\ldots, \vec{w}_N)$ that are \emph{not} sampled from the discrete Gaussian distribution. Then, $f_W(\vec{t})$ will presumably no longer be a good approximation to $f(\vec{t})$, and the verifier might therefore be fooled into thinking that $\vec{t}$ is far from the lattice when it is in fact quite close.

It seems that what we need is some sort of ``test of Gaussianity'' to ``check that $W$ looks like it was sampled from $D_{\lat^*}^N$.'' Or, more accurately, we need some efficiently testable set of properties that (1) are satisfied by honestly sampled vectors $W = (\vec{w}_1,\ldots, \vec{w}_N) \in \R^{n \times N}$ with high probability in the FAR case; and (2) are enough to imply that $f_W(\vec{t})$ is not too small in the CLOSE case when $\vec{t}$ is relatively close to $\lat$. One crucial observation is that, as long as the $\vec{w}_i$ are dual lattice vectors, then it suffices in the CLOSE case to consider $f_W(\vec{u})$ for $\vec{u}$ that are relatively short. This is because the function $f_W$ is \emph{periodic} over the lattice, so that $f_W(\vec{t}) = f_W(\vec{u})$ where $\vec{u} := \vec{t} - \vec{y}$ for $\vec{y} \in \lat$ a closest lattice vector to $\vec{t}$. (It is crucial to remember that $\vec{u}$ is only used for the analysis. In particular, the verifier cannot compute $\vec{u}$ efficiently.)

To create a sound protocol, Aharonov and Regev therefore studied the second-order Taylor series expansion of $f_W(\vec{u})$ around $\vec{u} = \vec0$, i.e.,
\begin{align*}
    f_W(\vec{u}) 
        &= 1- \frac{2\pi^2}{N} \cdot \sum_{i=1}^N \langle \vec{w}_i, \vec{u} \rangle^2 + \frac{2 \pi^4}{3N} \sum_{i=1}^N \langle \vec{w}_i, \vec{u} \rangle^4 - \cdots \\
        &\geq 1- \frac{2\pi^2}{N} \cdot \sum_{i=1}^N \langle \vec{w}_i, \vec{u} \rangle^2 \\
        &\geq 1 - 2\pi^2 \cdot \|WW^T/N\| \cdot \|\vec{u}\|^2
    \; ,
\end{align*}
where $WW^T/N = \frac{1}{N} \sum_i \vec{w}_i \vec{w}_i^T \in \R^{n \times n}$ and $\|WW^T/N\|$ is the spectral norm. In particular, $f_W(\vec{u})$ will be large for short $\vec{u}$, provided that $WW^T/N$ has small spectral norm. One can show (again using the Poisson summation formula) that an honestly sampled witness $W$ will satisfy, say, $\|W W^T/N\| \leq 1$ with high probability. And, the verifier can of course efficiently check this because the spectral norm is efficiently computable. So, Aharonov and Regev used this simple test as their ``test of Gaussianity.''

Putting everything together, we see that by checking that $W$ consists of dual vectors, that $\|WW^T/N\| \leq 1$, and that, say, $f_W(\vec{t}) < 1/2$, the verifier will always reject when $\dist(\vec{t},\lat)$ is smaller than some constant in the CLOSE case, regardless of $W$. And, it will accept (with high probability over the choice of witness $W$) when $W$ is sampled honestly and $\dist(\vec{t},\lat) \geq \sqrt{n}$ in the FAR case. This yields the final approximation factor of $O(\sqrt{n})$ achieved in~\cite{aharonovLatticeProblemsNP2005}.

Notice, however, that by using this spectral-norm-based ``test of Gaussianity,'' Aharonov and Regev only achieved an approximation factor of $O(\sqrt{n})$, rather than the approximation factor $O(\sqrt{n/\log N})$ that they would have gotten if they could have somehow guarantee that the $W$ were sampled honestly. In particular, when $N = \poly(n)$, this costs a factor of roughly $\sqrt{\log n}$ in the approximation factor. (At a technical level, this factor of $\sqrt{\log n}$ is lost because the second-order approximation $\cos(x) \approx 1-x^2/2$ is of course only accurate when $|x|$ is bounded by some small fixed constant.)

Fixing this (again, rather maddening) loss of a $\sqrt{\log n}$ factor has been a major open problem ever since. More generally, it is not at all clear how to achieve even a slightly better approximation factor using these ideas, even if we are willing to increase $N$ and the running time of our verifier substantially. It seems relatively clear that a more demanding ``test of Gaussianity'' is needed.

A natural idea would be to approximate $f_W$ via a higher-order Taylor series approximation,
\[
    f_W^{(k)}(\vec{u}) := 1-\sum_{j=1}^{k} \frac{(2\pi)^{2j}}{(2j)!} \sum_{i=1}^N \langle \vec{w}_i, \vec{u} \rangle^{2j}/N \; .
\] 
It is not hard to see that $f_W^{(k)}$ is quite close to $f_W$ provided that $\vec{u}$ is not too long. Specifically,
\[
    |f_W^{(k)}(\vec{u}) - f_W(\vec{u})| \lesssim \frac{1}{N} \cdot \sum_{i=1}^N(\langle \vec{w}_i, \vec{u} \rangle/k)^{2k}
    \; .
\]
We know that when $W$ is sampled honestly, this error cannot be much larger than roughly $(\|\vec{u}\|^2/k)^{k}$ (with high probability). Therefore, when $W$ is sampled honestly, it must be the case that the \emph{moments} $\sum_{i=1}^N \langle \vec{w}_i, \vec{u} \rangle^{2j}/N$ of $W$ have some property that guarantees that $f_W^{(k)}(\vec{u}) \gtrsim e^{-\pi \|\vec{u}\|^2} - (\|\vec{u}\|^2/k)^{k}$. If we could somehow identify and test this property efficiently for sufficiently large $k$, then we could use this as our ``test of Gaussianity,'' and we would be done.

However, we do not know how to test this property efficiently, and it seems quite hard to do so in general. Even just for $j = 2$, it is in general computationally hard even to approximate, say,
\[
    \max_{\|\vec{u}\| \leq d} \frac{1}{N} \cdot \sum_{i=1}^N \langle \vec{w}_i, \vec{u} \rangle^{2j}
    \; ,
\]
as this is exactly the matrix two-to-four norm~\cite{BBH+HypercontractivitySumofsquaresProofs2012}.
(Compare this with the case of $j = 1$, which yields the easy-to-compute spectral norm.) And, bounding the specific sum $f_W^{(k)}$ that interests us seems significantly more complicated than bounding an individual term---perhaps particularly because it is an alternating sum. It is therefore entirely unclear how to efficiently certify that the sum defining $f_W^{(k)}(\vec{u})$ is bounded whenever $\vec{u}$ is bounded.

We solve this problem by asking for additional properties of our lattice $\lat$ in the FAR case that allow us to make this problem tractable. Specifically, we require that in the FAR case, not only do we have $\dist(\vec{t},\lat) \geq \sqrt{n}$, but we also have that $\lat$ has no non-zero vectors shorter than $\sqrt{n}$. (Intuitively, this means that ``the Gaussian peaks of $f(\vec{t})$ are well separated,'' as in the left example in \cref{fig:periodic_gaussian}.) Micciancio and Regev~\cite{MR04} considered this more restrictive promise problem (for roughly the same reason) and observed that $\gamma$-$\SVP$ can be reduced to it. (It is this additional requirement in the FAR case that prevents us from obtaining a protocol for $\CVP$, rather than $\SVP$.)

Via Fourier-analytic techniques, we show that this new requirement implies that in the FAR case, the \emph{moments} of the discrete Gaussian $D_{\lat^*}$, 
\[\expect_{\vec{w} \sim D_{\lat^*}}[w_1^{j_1} \cdots w_n^{j_n}]
\; ,
\]are extremely close to the corresponding moments of the continuous Gaussian distribution as long as the $j_i$ are non-negative integers such that $\sum j_i$ is not too large. (See \cref{lem:smooth_Hermite_multi-dim}.) We then observe that these moments for $\sum j_i \leq 2k$ completely characterize $f_W^{(k)}(\vec{u})$. 

So, while in general it seems to be difficult to determine whether a given witness $W$ has the property that $f_W^{(k)}(\vec{u})$ is not too small for all sufficiently short $\vec{u}$, we show that in our special use case, it suffices for the verifier to check that the sample moments
\[
    \frac{1}{N} \sum_{i=1}^N w_{i,1}^{j_1} \cdots w_{i,n}^{j_n}
    \; 
\]
are close to some specific known values for $\sum j_i \leq 2k$.

There are roughly $n^{2k}$ such moments to check, and each can be checked in time essentially $N$. If these checks pass, then we can use $f_W(\vec{t})$ to distinguish the CLOSE case from the FAR case as long as in the close case we have
\[
    1/\sqrt{N} +  (\dist(\vec{t},\lat)^2/k)^{k} \lesssim e^{-\pi \dist(\vec{t},\lat)^2}
    \; .
\]
In particular, by setting $N = n^{O(k)}$, we will not be fooled in the CLOSE case as long as $\dist(\vec{t},\lat) \lesssim \sqrt{k}/10$, which gives our approximation factor of roughly $\sqrt{n/k}$ in time roughly $n^{O(k)}$.\footnote{This description might suggest that we can take $N \leq 2^{O(k)}$, yielding a $n^{O(k)}$-time protocol with $2^{O(k)}$-sized witness. However, in this informal discussion we are ignoring the error that we incur from the fact that the sample moments $\frac{1}{N} \sum_i w_{i,1}^{j_1} \cdots w_{i,n}^{j_n}$ will deviate from their expectation. After accounting for this, we are forced to take $N \geq n^{\Omega(k)}$.} See \cref{sec:coNP}.

\subsubsection{A \texorpdfstring{$\mathsf{coMA}$}{coMA} protocol}
\label{sec:techniques_MA}

Our $\mathsf{coMA}$ protocol combines some of the beautiful ideas from~\cite{aharonovLatticeProblemsNP2005} with some of the equally beautiful ideas from~\cite{GG00}.

Indeed, recall that \cite{aharonovLatticeProblemsNP2005} and our generalization show how to generate a witness $W$ of size roughly $N$ such that, if $W$ is sampled honestly, it can be used to distinguish the case when $\dist(\vec{t},\lat) \geq \sqrt{n}$ from the case when $\dist(\vec{t},\lat) \lesssim \sqrt{\log N}$. Specifically, there is a simple function $f_W(\vec{t})$ that is large in the CLOSE case but small in the FAR case, provided that the witness $W$ is generated honestly. The difficulty, in both the original Aharonov and Regev protocol and in our version described above, is in how to handle dishonestly generated $W$, in which case $f_W$ might not have this property and might therefore lead Arthur to incorrectly think that $\vec{t}$ is far from the lattice when in fact it is close.

On the other hand, \cite{GG00} and our generalize work by either sampling $\vec{x}$ from a ball around a lattice point or sampling $\vec{x}$ from a ball around (a lattice shift of) $\vec{t}$. Then, in the FAR case, a random vector $\vec{x}$ sampled from a ball around a lattice point will always be closer to $\lat$ than a random vector $\vec{x}$ sampled from a ball around $\vec{t}$. So, in the FAR case, an honest Merlin can determine whether $\vec{x}$ was sampled from one distribution or the other by checking whether $\dist(\vec{x}, \lat)$ is large or small. On the other hand, in the CLOSE case, there is some overlap between the distributions, so that no matter how Merlin behaves, he will not be able to consistently distinguish between the two cases. (Recall \cref{figure:distributions}.)

Our idea is therefore to have Arthur ``use $W$ to simulate Merlin's behavior in the $\mathsf{coAM}$ protocol.'' In particular, the witness for our protocol is exactly the same $W$ that we use as a witness in our co-non-deterministic protocol (and therefore simply a larger version of the original~\cite{aharonovLatticeProblemsNP2005} witness). However, Arthur's verification procedure is quite different (and, of course, it is now randomized, which is why we obtain a $\mathsf{coMA}$ protocol). To verify Merlin's claim that $\dist(\vec{t},\lat) \geq \sqrt{n}$, Arthur repeatedly samples $\vec{x}_0$ from a ball of radius $r$ around $\vec0$ and $\vec{x}_1$ from a ball of radius $r$ around $\vec{t}$, where $r$ is to be set later. Arthur then computes $f_W(\vec{x}_0)$ and $f_W(\vec{x}_1)$ and rejects (i.e., guesses that he is in the CLOSE case) unless $f_W(\vec{x}_0)$ is large and $f_W(\vec{x}_1)$ is small.

Note that, at least at a high level, the completeness of our protocol in the FAR case follows from the analysis of \cite{aharonovLatticeProblemsNP2005}. In particular, if $W$ is sampled honestly, then $f_W(\vec{x}_0)$ will be large as long as $ r \lesssim \sqrt{\log N}$, and $f_W(\vec{x}_1)$ will be small as long as $\dist(\vec{t},\lat) - r \gtrsim \sqrt{n}$.
On the other hand, the soundness of our protocol in the CLOSE case follows from the analysis of \cite{GG00}. In particular, if $\dist(\vec{t}, \lat) \leq d \leq r$, then regardless of our choice of $f_W$, Arthur will reject with probability at least
\[
    p_{d,r}/2 \approx (1-d^2/(4r^2))^{n/2} \approx e^{-d^2 n/r^2}
    \; ,
\]
as in our discussion of the $\mathsf{coAM}$ protocol above.
By running this test, say, $n/p_{d,r}$ times, Arthur will reject with high probability in the CLOSE case.

Plugging in numbers, we take $r$ as large as we possibly can without violating completeness, so we take $r \approx \sqrt{\log N}$. Our final protocol then has a witness of size roughly $N$, an approximation factor of roughly $\sqrt{n}/d$, and a running time of roughly $N \cdot e^{d^2 n/r^2} \approx N \cdot e^{d^2 n/\log N}$. The most natural setting of parameters takes $d \approx \log N/\sqrt{n}$, which gives an approximation factor of $n/\log N$ in $\poly(N)$ time. (However, we note that the protocol is also potentially interesting in other parameter settings; e.g., one can obtain non-trivial approximation factors with relatively small communication size $N$ by allowing Arthur to run in more time. In contrast, our other protocols seem to require roughly as much communication as computation.) See \cref{sec:coMA}.

\subsubsection{Worst-case to average-case reductions for SIS}
\label{sec:techniques_SIS}

Our generalization of Micciancio and Regev's worst-case to average-case reduction for SIS comes nearly for free after all the work we did to develop our variant of Aharonov and Regev's $\mathsf{coNP}$ protocol (and our $\mathsf{coMA}$ protocol). In particular, Micciancio and Regev's worst-case to average-case reduction essentially shows how to use a $\SIS$ oracle to sample from $D_{\lat^*}$ (provided that $\lambda_1(\lat)$ is not too small). They then used this to generate the witness for Aharonov and Regev's protocol, allowing them to solve $\gamma$-$\SVP$. Our co-non-deterministic protocol also uses samples from $D_{\lat^*}$ as a witness (as does our $\mathsf{coMA}$ protocol). So, we are more-or-less able to use the exact same idea to obtain our generalization of Micciancio and Regev's result. (In fact we are able to work in the more general setting of Micciancio and Peikert \cite{micciancioHardnessSISLWE2013}, who showed a reduction that works for smaller moduli than~\cite{MR04}.)

However, our co-non-deterministic protocol is a bit more delicate than the protocol in~\cite{aharonovLatticeProblemsNP2005}. Specifically, our reduction really does need to produce samples from $D_{\lat^*}$ in the FAR case. In contrast,~\cite{MR04} (and, to our knowledge, all other worst-case to average-case reductions for $\SIS$) were only able to show how to use a $\SIS$ oracle to produce samples from some mixture of discrete Gaussian distributions with potentially different parameters (i.e., different standard deviations). At a technical level, this issue arises because the $\SIS$ oracle can potentially output vectors with different lengths, resulting in discrete Gaussian samples with different parameters.

We overcome this (annoying!) technical difficulty by showing how to control the parameter of the samples generated by the reduction, showing that a $\SIS$ oracle is in fact sufficient to produce samples from the distribution $D_{\lat^*}$ itself (provided that the smoothing parameter of $\lat^*$ is small enough). Our reduction mostly follows the elegant and well known reduction of Micciancio and Peikert \cite{micciancioHardnessSISLWE2013}. And, though the proof does not require substantial new ideas, we expect that the result will be useful in future work---as a reduction directly from discrete Gaussian sampling should be quite convenient. (See \cref{thm:DGStoSIS}.)

Finally, in order to get the best approximation factor that we can, we actually use our $\mathsf{coMA}$ protocol when the running time is large, rather than our co-non-deterministic protocol. E.g., our $\mathsf{coMA}$ protocol saves a factor of $\sqrt{\log n}$ in the approximation factor over our co-non-deterministic protocol in the important special case when the running time is $2^{\eps n}$. And, our worst-case to average-case reduction inherits this savings. See \cref{sec:SIS}.

\subsubsection{Worst-case to average-case reductions for LWE}

Recall that we show two worst-case to average-case reductions for $\LWE$. One is a \emph{quantum} reduction, following Regev~\cite{regevLatticesLearningErrors2009}. The other is a \emph{classical} reduction, following Peikert~\cite{peikertPublickeyCryptosystemsWorstcase2009}. In both cases, our modifications to prior work are surprisingly simple.

In the quantum case, the \emph{only} difference between our reduction and prior work is in a single step. Specifically, Regev's original quantum reduction is most naturally viewed as a reduction from $\BDD$ to $\LWE$. However, $\BDD$ is not nearly as well studied as $\SVP$. Regev therefore used elegant quantum computing tricks to obtain hardness directly from $\SVP$. However, Peikert~\cite{peikertPublickeyCryptosystemsWorstcase2009} and Lyubashevsky and Micciancio~\cite{conf/crypto/LyubashevskyM09} later showed a simple classical reduction from $\SVP$ to $\BDD$ that is perhaps best viewed as a version of the $\mathsf{coAM}$ protocol from~\cite{GG00} in which the $\BDD$ oracle is used to simulate Merlin. (This reduction was implicit in~\cite{peikertPublickeyCryptosystemsWorstcase2009} and made explicit in~\cite{conf/crypto/LyubashevskyM09}.) 

Using the same ideas that we used to generalize the $\mathsf{coAM}$ protocol from~\cite{GG00}, we show how to generalize this reduction from $\SVP$ to $\BDD$---showing that a better approximation factor is achievable if the reduction is allowed more running time. By composing this reduction with Regev's reduction from $\BDD$ to $\LWE$, we similarly show a time-approximation-factor tradeoff for $\LWE$.

To generalize Peikert's \emph{classical} reduction, we use the above idea and also make one other simple modification to the reduction. Specifically, Peikert showed how to use a sufficiently ``nice'' basis of the dual lattice $\lat^*$ to reduce $\BDD$ to $\LWE$, where the modulus of the $\LWE$ instance depends on how ``nice'' the basis is.\footnote{A larger modulus yields a larger approximation factor. This was fundamentally the reason why Peikert's classical reduction achieved a larger approximation factor than Regev's quantum reduction. Later work~\cite{brakerskiClassicalHardnessLearning2013} showed how to prove classical hardness of $\LWE$ with a smaller modulus. However, the \cite{brakerskiClassicalHardnessLearning2013} reduction works by reducing from the large-modulus case and incurs additional loss in the approximation factor. Therefore, it improves the modulus of Peikert's reduction without improving the approximation factor.} He then used the celebrated LLL algorithm~\cite{lenstraFactoringPolynomialsRational1982} to efficiently find a relatively nice basis. We simply plug in generalizations of \cite{lenstraFactoringPolynomialsRational1982} that obtain better bases in more time~\cite{schnorrHierarchyPolynomialTime1987,gamaFindingShortLattice2008}. (In fact, even in the polynomial-time regime, this improves on Peikert's approximation factor by a small polylogarthmic term See \cref{sec:LWE}.)

To achieve the best parameters, we rely on the ``direct-to-decision'' reduction of~\cite{PRS17} (in both the classical and quantum setting), allowing us to avoid the search-to-decision reductions that were used in work prior to~\cite{PRS17}. (Search-to-decision reductions that do not increase the approximation factor are known for some moduli, but for other moduli, the only known search-to-decision reductions incur a loss in the approximation factor.)

\section{Preliminaries}

We use boldface lower-case letters $\vec{z} \in \Z^n, \vec{x} \in \Q^n, \vec{y} \in \R^n$ to represent \emph{column} vectors. All logarithms are base $e$ unless specified otherwise. We write $\|\vec{x}\|$ for the $\ell_2$ norm of a vector, i.e., $\|\vec{x}\| := (x_1^2 + \cdots + x_n^2)^{1/2}$. When we work with other norms, we explicitly clarify this by writing $\|\vec{x}\|_K$.

\subsection{Lattices}

A lattice $\lat \subset \R^n$ is the set of all integer linear combinations of linearly independent basis vectors $\basis := (\vec{b}_1,\ldots, \vec{b}_n) \in \R^n$,
\[
    \lat = \{ z_1 \vec{b}_1 + \cdots + z_n \vec{b}_n \ : \ z_i \in \Z\}
    \; .
\]
We call $n$ the \emph{rank} of the lattice. When we wish to emphasize that $\lat$ is the lattice generated by $\basis$, we write $\lat = \lat(\basis)$.

We write $\lat_{\neq \vec0}$ for the set of lattice vectors excluding the zero vector, and we define
\[
    \lambda_1(\lat) := \min_{\vec{y} \in \lat_{\neq \vec0}} \|\vec{y}\|
\]
for the length of the shortest non-zero vector in the lattice.

Given a target vector $\vec{t} \in \R^n$, we write
\[
    \dist(\vec{t},\lat) := \min_{\vec{y} \in \lat} \|\vec{y} - \vec{t}\|
\]
for the distance between $\vec{t}$ and the lattice $\lat$.

The \emph{dual lattice} is
\[
    \lat^* := \{\vec{w} \in \R^n \ : \ \forall \vec{y} \in \lat,\ \langle \vec{w}, \vec{y} \rangle \in \Z \}
    \; ,
\]
i.e., the set of all vectors that have integer inner product with \emph{all} lattice vectors. It is not hard to see that $\lat^*$ is itself a lattice, that $(\lat^*)^* = \lat$, and that if $\lat = \lat(\basis)$, then $\lat^* = \lat((\basis^{-1})^T)$. In particular, for any $s > 0$, $(s \lat)^* = s^{-1} \lat^*$.

For a basis $\basis$ for a lattice $\lat \subset \R^n$, we write $\mathcal{P}(\basis) := \{\basis \vec{x} \ : \ \vec{x} \in [0,1)^n\}$ for the \emph{fundamental parallelepiped} defined by the basis.  For a point $\vec{t} \in \R^n$, we write $\vec{t}' := \vec{t} \bmod \mathcal{P}(\basis)$ for the unique element $\vec{t}' \in \mathcal{P}(\basis)$ such that $\vec{t} = \vec{t}' + \vec{y}$ for some lattice vector $\vec{y} \in \lat$. Equivalently, if $\vec{t} = \basis \vec{x}$, then $\vec{t} \bmod \mathcal{P}(\basis) = \basis(\vec{x} - \floor{\vec{x}})$, where $\floor{\cdot}$ represents the coordinate-wise floor function.

\subsection{Volume lower bounds}

\paragraph{Intersections of Euclidean Balls.}
We start by giving a nearly tight lower bound on the volume of the intersection of two unit Euclidean balls. We use calculations very similar to those appearing in~\cite{bcklp/lotsofnorms22}.

Let $\B_2^n := \set{\vec{x} \in \R^n : \norm{\vec{x}} \leq 1}$ denote the closed $n$-dimensional unit Euclidean ball; let $\Vball_n := \vol_n(\B_2^n) = (\pi^{n/2}/\Gamma(n/2 + 1))$, where $\vol_n$ denotes $n$-dimensional volume and $\Gamma$ is the $\Gamma$ function; and let $\Vcap_n(r, \theta)$ denote the volume of a spherical cap of angle $\theta \in [0, \pi/2]$ in an $n$-dimensional Euclidean ball of radius $r$. (The angle of a spherical cap of a Euclidean ball is the angle formed by (1) the line segment connecting the center of the  ball and the center of the cap and (2) a line segment connecting the center of the ball and a point on the boundary of the base of the cap.)
As shown, e.g., in~\cite{Li2010}, 
\begin{equation} \label{eq:Vcap}
\Vcap_n(r, \theta) = \Vball_{n-1} \cdot r^n \cdot \int_{0}^{\theta} \sin^n(\phi) \intd \phi \ \text{.}
\end{equation}
We next use a trick that we learned from Zhou~\cite{zhou} for lower bounding $\int_{0}^{\theta} \sin^n(\phi) \intd \phi$, which is as follows.
Because $\cos(\phi) \in [0, 1]$ for all $\phi \in [0, \pi/2]$, we have that for $\theta \in [0, \pi/2]$,
\begin{equation} \label{eq:Jn-lb}
\int_{0}^{\theta} \sin^n(\phi) \intd \phi \geq \int_{0}^{\theta} \sin^n(\phi) \cos(\phi) \intd \phi = \sin^{n+1}(\theta)/(n+1) \ \text{.}
\end{equation}
Combining \cref{eq:Vcap,eq:Jn-lb}, we get that for $r > 0$ and $\theta \in [0, \pi/2]$,
\begin{equation} \label{eq:Vcap-lb}
\Vcap_n(r, \theta) \geq \Vball_{n-1} \cdot r^n \cdot \sin^{n+1}(\theta)/(n+1) \ \text{.}
\end{equation}

\begin{lemma} \label{lem:normalized-l2-ball-intersection-vol}
For any $n \in \Z^+$, any $d$ satisfying $0 \leq d \leq 2$, and any $\vec{v} \in \R^n$ satisfying $\norm{\vec{v}} = d$,
\[
\frac{\vol_n(\B_2^n \cap (\B_2^n + \vec{v}))}{\vol_n(\B_2^n)} \geq \sqrt{1/(2 \pi n)} \cdot (1 - d^2/4)^{(n+1)/2} \ \text{.}
\]
\end{lemma}

\begin{proof}
Because the balls $\B_2^n$ and $\B_2^n + \vec{v}$ are congruent, the volume $\vol_n(\B_2^n \cap (\B_2^n + \vec{v}))$ is equal to the volume of the union of two disjoint, congruent spherical caps of $\B_2^n$. (See \cref{figure:distributions}.) Using the law of cosines, it is straightforward to show that these caps each have angle $\arccos(d/2)$, and therefore that
\begin{equation} \label{eq:intersection-as-caps}
\vol_n(\B_2^n \cap (\B_2^n + \vec{v})) = 2 \Vcap_n(1, \arccos(d/2)) \ \text{.}
\end{equation}
Combining \cref{eq:Vcap-lb,eq:intersection-as-caps}, and using the trigonometric identity $\sin(\arccos(x)) = \sqrt{1 - x^2}$ and the fact that $n \geq 1$, we additionally have that
\begin{equation} \label{eq:norm-ball-int-lb}
\frac{\vol_n(\B_2^n \cap (\B_2^n + \vec{v}))}{\vol_n(\B_2^n)} \geq \frac{2 \Vball_{n-1} \cdot \sin^{n+1}(\arccos(d/2))}{(n+1) \cdot \Vball_n} 
\geq \frac{\Vball_{n-1} \cdot (1 - d^2/4)^{(n+1)/2}}{n \cdot \Vball_n} \ \text{.}
\end{equation}
Furthermore, using Gautschi's inequality~\cite{gautschis-inequality}, we have that
\begin{equation} \label{eq:Vnm1-Vn-ratio-lb}
\frac{\Vball_{n-1}}{\Vball_{n}}
= \frac{\Gamma(n/2 + 1)}{\sqrt{\pi} \cdot \Gamma(n/2 + 1/2)} \geq \sqrt{n/(2\pi)} \ \text{.}
\end{equation}
The lemma then follows by combining \cref{eq:norm-ball-int-lb,eq:Vnm1-Vn-ratio-lb}.
\end{proof}

\paragraph{Intersections of pairs of arbitrary convex bodies.} A set $K \subseteq \R^n$ is called a \emph{centrally symmetric convex body} if $K$ is compact, convex, and symmetric (i.e., if $K = -K$). There is a one-to-one correspondence between centrally symmetric convex bodies $K$ and the norms $\norm{\cdot}_K$ that they induce. For $\vec{x} \in \R^n$, we define
\[
\norm{\vec{x}}_K := \min \set{r > 0 : \vec{x} \in r K} \ \text{.}
\]

In the following lemma, we prove an analog of \cref{lem:normalized-l2-ball-intersection-vol} for general centrally symmetric convex bodies $K$.
\begin{lemma} \label{lem:normalized-K-intersection-vol}
Let $n \in \Z^+$, let $K \subseteq \R^n$ be a centrally symmetric convex body, let $d$ satisfy $0 \leq d \leq 2$, and let $\vec{v} \in \R^n$ be such that $\norm{\vec{v}}_K = d$. Then
\[
\frac{\vol_n(K \cap (K + \vec{v}))}{\vol_n(K)} \geq (1 - d/2)^n \ \text{.}
\]
\end{lemma}

\begin{proof}
We claim that $(1 - d/2) K + \vec{v}/2 \subseteq K \cap (K + \vec{v})$, from which the lemma follows. Indeed, by shift invariance and scaling properties of volume, $\vol_n((1 - d/2) K + \vec{v}/2) = \vol_n((1 - d/2) K) = (1 - d/2)^n \cdot \vol_n(K)$.

It remains to prove the claim. Let $\vec{x} \in (1 - d/2) K + \vec{v}/2$. By applying triangle inequality twice, we have that 
$\norm{\vec{x}}_K \leq \norm{\vec{x} - \vec{v}/2}_K + \norm{\vec{v}/2}_K = (1 - d/2) + d/2 = 1$ and that $\norm{\vec{x} - \vec{v}} \leq \norm{\vec{x} - \vec{v}/2}_K + \norm{\vec{v}/2 - \vec{v}}_K = (1 - d/2) + d/2 = 1$. The claim follows.
\end{proof}

We remark that although \cref{lem:normalized-K-intersection-vol} is somewhat quantitatively weaker than \cref{lem:normalized-l2-ball-intersection-vol}---it lower bounds the normalized volume of the intersection of two bodies by $(1 - d/2)^n$ instead of roughly $(1 - d^2/4)^{n/2}$---it is in general tight. In particular, it is tight for all $d$ with $0 \leq d \leq 2$ when $K$ is a (hyper)cube, i.e., when $K$ is the $\ell_{\infty}$ unit ball.

\subsection{Gaussians and the discrete Gaussian}

		For $\vec{x} \in \R^n$ and $s > 0$, we write
		\[
			\rho_s(\vec{x}) = \exp(-\pi \|\vec{x}\|^2/s^2)
			\; .
		\]
		For a lattice $\lat \subset \R^n$ and shift $\vec{t} \in \R^n$, we write
		\begin{align*}
		 \rho_s(\lat) &= \sum_{\vec{x} \in \lat} \rho_s(\vec{x})\;\text{, and}\\
		 \rho_s(\lat - \vec{t}) &= \sum_{\vec{x} \in \lat - \vec{t}} \rho_s(\vec{x}) = \sum_{\vec{y} \in \lat} \rho_s(\vec{y} - \vec{t}) \;.
		\end{align*}
		We write
		\[
			f_{\lat, s}(\vec{t}) := \frac{\rho_s(\lat - \vec{t})}{ \rho_s(\lat)}
			\; .
		\]
		Unless there is confusion, we omit the parameter $\lat$ and simply write $f_s(\vec{t})$. Finally, we write $D_{\lat, s}$ for the discrete Gaussian, which is the distribution induced by the measure $\rho_s$ on $\lat$. I.e., for $\vec{y} \in \lat$,
		\[
		\Pr_{\vec{X} \sim D_{\lat, s}}[\vec{X} = \vec{y}] = \frac{\rho_s(\vec{y})}{\rho_s(\lat)}
		\; .
		\]
		
		When $s = 1$, we omit it and simply write $\rho(\vec{x})$, $\rho(\lat - \vec{t})$, $f(\vec{t})$, and $D_{\lat}$ respectively.
		
		For a lattice $\lat \subset \R^n$ and $\eps > 0$, the \emph{smoothing parameter} $\eta_\eps(\lat)$ (introduced in~\cite{MR04}) is the unique parameter satisfying $\rho_{1/\eta_\eps(\lat)}(\lat^*) = 1 + \eps$. Equivalently, $\rho_{1/s}(\lat^*) \leq 1+\eps$ if and only if $s \geq \eta_\eps(\lat)$.

		We will need the following two facts about the Gaussian, both proven by Banaszczyk~\cite{banaszczykNewBoundsTransference1993}.

		\begin{theorem}[{\cite{banaszczykNewBoundsTransference1993}}]
	\label{thm:banaszczyk_shift_big}
			For any lattice $\lat \subset \R^n$, shift $\vec{t} \in \R^n$, and parameter $s > 0$,
			\[
			\rho_s(\dist(\vec{t}, \lat)) \leq f_s(\vec{t})  \leq 1
			\; .
			\]
		\end{theorem}

		\begin{theorem}[{\cite{banaszczykNewBoundsTransference1993}}]
			\label{thm:banaszczyk_tail}
			For any lattice $\lat \subset \R^n$, shift $\vec{t} \in \R^n$, parameter $s > 0$, and $t \geq 1/\sqrt{2\pi}$,
			\[
			\sum_{\stackrel{\vec{x} \in \lat  - \vec{t}}{\|\vec{x}\| \geq st\sqrt{n}} } \rho_s(\vec{x}) \leq (2\pi e t^2)^{n/2} e^{-\pi t^2 n} \rho_s(\lat)
			\; .
			\]
		\end{theorem}
		
		We will actually be content with the following corollary, which is much easier to use.
		
		\begin{corollary}\label{cor:banaszczyk_tail}
		    For any lattice $\lat \subset \R^n$, shift $\vec{t} \in \R^n$, parameter $s > 0$, and radius $r \geq s \sqrt{n/(2\pi)}$,
		    \[
		        \Pr_{\vec{y} \sim D_{\lat, s}}[\|\vec{y}\| \geq r] \leq e^{-\pi x^2}
		        \; ,
		    \]
		    where $x := r/s - \sqrt{n/(2\pi})$.
		    
		    In particular, $\eta_{2^{-n}}(\lat) \leq \sqrt{n}/\lambda_1(\lat^*)$.
		\end{corollary}
		\begin{proof}
		    From \cref{thm:banaszczyk_tail}, we have
		    \begin{align*}
		        \Pr[[\|\vec{y}\| \geq r] 
		            &\leq (2\pi e r^2/(n s^2))^{n/2} e^{-\pi r^2/s^2} \\
		            &= e^{-\pi x^2 - \sqrt{2\pi n} x} \cdot (1+\sqrt{2\pi/n} \cdot x)^{n}\\
		           &\leq e^{-\pi x^2}
		           \; ,
		    \end{align*}
		    where we have used the simple identity $(2\pi e r^2/(n s^2)) = (1+ \sqrt{2\pi/n} \cdot x)^2$ and the inequality $(1+y) \leq e^y$.
		\end{proof}
		
		We will also need the following result from \cite{micciancioTrapdoorsLatticesSimpler2012}.
		
		\begin{lemma}[{\cite[Lemma 2.8]{micciancioTrapdoorsLatticesSimpler2012}}]
			\label{lem:subgaussianity}
			For any lattice $\lat \subset \R^n$ and parameter $s > 0$, the discrete Gaussian $D_{\lat, s}$ is subgaussian with parameter $s$. I.e., for any unit vector $\vec{v} \in S^{n-1}$ and $r > 0$,
			\[
			\Pr_{\vec{y} \sim D_{\lat, s}}[\inner{\vec{y}, \vec{v}} \geq r] \leq \exp(-\pi r^2/s^2)
			\; .
			\]
			Furthermore,
			\[
			\Pr_{\vec{y} \sim D_{\lat, s}}[|\inner{\vec{y}, \vec{v}}| \geq r] \leq 2\exp(-\pi r^2/s^2)
			\; .
			\]
		\end{lemma}

The next three results are useful for building algorithms that work with the discrete Gaussian distribution. The first is (a special case of) the convolution theorem of Micciancio and Peikert.

\begin{theorem}[{\cite[Theorem 3.3]{micciancioHardnessSISLWE2013}}]
\label{thm:convolution}
Let $\lat$ be an $n$-dimensional lattice. Fix $\eps > 0$ and  $\vec{z} \in \{-1, 0, 1\}^m$. For $1 \leq i \leq m$, let $s_i \geq \sqrt{2} \| \vec{z} \|_\infty \cdot \eta_\eps(\lat)$, and let $\vec{y_i}$ be independently sampled from $D_{\lat,  s_i}$.
Then the distribution of $\vec{y} := \sum_i z_i \vec{y_i}$ has statistical distance at most $\eps$ from $D_{\lat, s}$, where $s := \sqrt{\sum_i (z_i s_i)^2}$.
\end{theorem}
\begin{proof}
    The statement is a special case of \cite[Theorem 3.3]{micciancioHardnessSISLWE2013}, except that the guarantee on statistical distance has been made explicit. 
\end{proof}

		We will also need the following algorithm for sampling from discrete Gaussians, given a good basis of a lattice. The algorithm itself is originally due to Klein~\cite{kleinFindingClosestLattice2000}, and was first shown to obtain samples from the discrete Gaussian by Gentry, Peikert, and Vaikuntanathan~\cite{gentryTrapdoorsHardLattices2008}. We use the version from~\cite{brakerskiClassicalHardnessLearning2013}, which works for slightly better parameters and does not incur any error. Here, $\|\widetilde{\vec{B}}\|$ means the maximal norm of a vector in the Gram-Schmidt orthogonalization of the basis. 
            \begin{theorem}[{\cite[Lemma 2.3]{brakerskiClassicalHardnessLearning2013}}]\label{thm:BLP_sampler}
                There is a probabilistic polynomial-time algorithm that takes as input a basis $\basis$ for a lattice $\lat \subset \R^n$ and $s > \norm{\widetilde{\basis}}\sqrt{\log n}$ and outputs a vector that is distributed exactly as $D_{\lat ,s}$. 
            \end{theorem}

            The next algorithm
            can be obtained by combining the discrete Gaussian sampling algorithm described above with basis reduction algorithms, e.g., from \cite{gamaFindingShortLattice2008}. The particular version that we use is a special case of~\cite[Proposition 2.17]{aggarwalSolvingShortestVector2015}, though tighter versions of this result are possible.
		
		\begin{theorem}[{\cite[Proposition 2.17]{aggarwalSolvingShortestVector2015}}]
		    \label{thm:BKZ_plus_GPV}
		    There is an algorithm that takes as input a (basis for a) lattice $\lat \subset \Q^n$ , $2 \leq \beta \leq n$ (the block size), and a positive integer $N$ such that if $\lambda_1(\lat) \geq \beta^{n/\beta}$, then the algorithm outputs $N$ vectors $\vec{w}_1,\ldots, \vec{w}_N$ that are distributed exactly as independent samples from $D_{\lat^*}$. Furthermore, the algorithm runs in time $(N + 2^{O(\beta)}) \cdot \poly(n)$.
		\end{theorem}
		
				\subsubsection{Hermite polynomials and moments of the Gaussian}
		
		In general, the Poisson Summation Formula tells us that for any sufficiently nice function $g(\vec{x})$ (e.g., the Schwartz functions suffice for our purposes) with Fourier transform $\widehat{g}(\vec{w}) := \int_{-\infty}^\infty e^{2\pi i \inner{\vec{w}, \vec{x}}} g(\vec{x}) {\rm d} \vec{x}$, we have
		\[
		\sum_{\vec{y} \in \lat} g(\vec{y})  = \frac{1}{\det(\lat)} \sum_{\vec{w} \in \lat^*} \widehat{g}(\vec{w})
		\; .
		\]
		By applying this to $\rho_s(\lat - \vec{t})$ and $\rho_s(\lat)$, we get the identity 
		\begin{equation}
		\label{eq:PSF_Gaussian}
		f_s(\vec{t}) = \expect_{\vec{w} \sim D_{\lat^*,1/s}}[\cos(2\pi \inner{\vec{w}, \vec{t}})]
		\; .
		\end{equation}
		We will also need the ``continuous version'' of this identity
		\begin{equation}
		    \label{eq:Gaussian_fourier}
		    \rho(\vec{u}) = \expect_{\vec{w} \sim D}[\cos(2\pi \inner{\vec{w}, \vec{t}}]
		    \; ,
		\end{equation}
		where $D$ is the continuous Gaussian distribution given by probability density function $\rho$.
		By applying the Poisson Summation Formula to $x_1^{a_1} \cdots x_n^{a_n} \rho_s(\vec{x})$ with $a_i \in \Z_{\geq 0}$, we see that
		\begin{equation}
		\label{eq:Hermite_identity}
			\expect_{\vec{x} \sim D_{\lat}}[x_1^{a_1}\cdots x_n^{a_n}] = \expect_{\vec{w} \sim D_{\lat^*}}[H_{\vec{a}}(\vec{w})]
			\; .
		\end{equation}
		Here, $H_{\vec{a}}(\vec{x})$ is the $\vec{a}$th multivariate Hermite polynomial,\footnote{There are many different definitions of the Hermite polynomials that differ in normalization. We choose the normalization that makes \cref{eq:Hermite_identity} true.} given by
		\[
		    H_{\vec{a}}(\vec{x}) := (2\pi i)^{- (a_1 + \cdots + a_n)} e^{\pi \|\vec{x}\|^2} \cdot \frac{\partial^{a_1}}{\partial x_1^{a_1}} \cdots \frac{\partial^{a_n}}{\partial x_n^{a_n}} e^{-\pi \|\vec{x}\|^2}
		    \; .
		\]
		E.g., $H_{1,0,\ldots,0}(\vec{x}) = i x_1$, and $H_{2,0,\ldots,0}(\vec{x}) = x^2-1/(2\pi)$, etc.
		Notice that, with this definition, \cref{eq:Hermite_identity} follows from the Poisson Summation Formula together with the fact that the Fourier transform of $2\pi i x g(x)$ is $ \frac{\partial}{\partial w} \widehat{g}(w)$. 
		
		We will need some basic properties of the Hermite polynomials.
		
		\begin{fact}
		    \label{fact:hermite}
		    For any $\vec{a} \in \Z_{\geq 0}^n$ with $k := a_1 + \cdots + a_n$, 
		    \begin{enumerate}
		        \item $H_{\vec{a}}(\vec{x}) := \sum_{b_1,\ldots, b_n} c_{\vec{a}, b_1,\ldots, b_n} x_1^{b_1} x_2^{b_2} \cdots x_n^{b_n}$ is a polynomial with degree $k$;
		        \item \label{item:bound_on_hermite_coefficients} the coefficients $c_{\vec{a}, \vec{b}}$ of $H_{\vec{a}}$ satisfy $\sum_{\vec{b}} |c_{\vec{a}, \vec{b}}| \leq (k+1)!$; and
		       \item the constant term of $H_{\vec{a}}$ is \[
		       H_{\vec{a}}(\vec0) = \int_{\R^n} x_1^{a_1} \cdots x_n^{a_n} e^{-\pi \|\vec{x}\|^2}\intd \vec{x} = V_{a_1} \cdots V_{a_n}
		       \; ,
		       \] 
		       where $V_{a} := \int_{-\infty}^\infty x^{a} e^{-\pi \|\vec{x}\|^2} \intd x$.
		     In particular, for all non-negative integers $a$, $V_{2a+1} = 0$ and $V_{2a} =  (2a)!/((4\pi)^a a!)$
		    \end{enumerate}
		\end{fact}
		\begin{proof}
		    The only non-trivial statement is \cref{item:bound_on_hermite_coefficients}. To prove this, we assume without loss of generality that $a_1 \geq 1$, and set $\vec{a}^{-} := (a_1 - 1, a_2,\ldots, a_n)$. Let $\vec{b}^+ := (b_1+1,b_2,\ldots, b_n)$ and $\vec{b}^- := (b_1 - 1,b_2,\ldots, b_n)$, where we adopt the convention that $\vec{c}_{\vec{a}', \vec{b}^-} = 0$ if $b_1 = 0$. Then, we notice that 
		    \[
		        2\pi i c_{\vec{a}, \vec{b}} = (b_1 + 1) c_{\vec{a}^-, \vec{b}^+} - 2\pi c_{\vec{a}^-, \vec{b}^-}
		        \; .
		    \]
		    Therefore,
		    \[
		        \sum_{\vec{b}} |c_{\vec{a}, \vec{b}}| \leq \sum_{\vec{b}} \Big(\frac{b_1+1}{2\pi} \cdot |c_{\vec{a}^-,\vec{b}^+}| + |c_{\vec{a}^-, \vec{b}^-}|\Big) \leq \sum_{\vec{b}} \Big(\frac{k}{2\pi} \cdot |c_{\vec{a}^-,\vec{b}^+}| + |c_{\vec{a}^-, \vec{b}^-}|\Big) \leq (k+1) \sum_{\vec{b}} |c_{\vec{a}^-, \vec{b}}|
		        \; .
		    \]
		    The result then follows by induction on $k$ (together with the base case $H_{\vec{0}}(\vec{x}) = 1$).
		\end{proof}
		
		The following lemma shows that if  $\lambda_1(\lat)\geq \sqrt{n}$, then the moments of $D_{\lat^*}$ are very close to the moments of the continuous Gaussian.
		
		\begin{lemma}
			\label{lem:smooth_Hermite_multi-dim}
			For any lattice $\lat \subset \R^n$ with $\lambda_1(\lat) \geq \sqrt{n}$, and any $\vec{a} \in \Z_{\geq 0}^n$ with $k := a_1 + \cdots + a_n \leq n/10$, we have
			\[
			V_{\lat^*, \vec{a}} := \expect_{\vec{w} \sim D_{\lat^*}}[w_1^{a_1} \cdots w_n^{a_n}] = V_{\vec{a}} +\eps
			\; ,
			\]
			where 
			\[
			V_{\vec{a}} := \int_{\R^n} x_1^{a_1} \cdots x_n^{a_n} \exp(-\pi x^2) \intd \vec{x}
			\; ,
			\]
			 and $|\eps | \leq 2^{-n} \cdot n^{2k}$. 
		\end{lemma}
		\begin{proof}
			By \cref{eq:Hermite_identity}, we have
			\[
			V_{\lat^*, \vec{a}} = \expect_{\vec{y} \sim D_{\lat}}[H_{\vec{a}}(\vec{y})] = \expect_{\vec{y} \sim D_{\lat}}[H_{\vec{a}}'(\vec{y})] + V_{\vec{a}}
			\; ,
			\]
			where $H_{\vec{a}}$ is the $\vec{a}$th Hermite polynomial, and where we write $H_{\vec{a}}'$ for the $\vec{a}$th Hermite polynomial without its constant term, which is equal to $V_{\vec{a}}$ by \cref{fact:hermite}.
			Notice that for $\vec{y} \in \lat$, we have \[|y_1^{b_1} \cdots y_n^{b_n}| \leq \max_i | y_i|^{b_1 + \cdots + b_n} \leq \max_i |y_i|^k\] whenever $b_1 + \cdots + b_n \leq k$. (Here, we have used the fact that $\vec{y} \in \lat$ is either the zero vector or $\max_i |y_i| \geq 1$, since $\lambda_1(\lat) \geq \sqrt{n}$.) Therefore, applying \cref{fact:hermite}, we have
			\[
			|\eps| = |V_{\lat^*, \vec{a}} -  V_k| \leq \expect_{\vec{y} \sim D_{\lat}}[|H_{\vec{a}}'(\vec{y})|] \leq (k+1)! \expect_{\vec{y} \sim D_{\lat}}[\max_i |y_i|^k]
			\; .
			\]
			
			It remains to bound this expectation. Indeed, we have
			\[
			    \expect_{\vec{y} \sim D_{\lat}}[\max_i |y_i|^k] = k \int_0^\infty r^{k-1} \Pr_{\vec{y} \sim D_{\lat}}[\max_i |y_i| > r] \intd r 
			    \; .
			\]
			By \cref{lem:subgaussianity}, we have $\Pr[|y_i| > r] \leq 2e^{-\pi r^2}$. Applying union bound, we see that $\Pr[\max_i |y_i| > r] \leq 2n e^{- \pi r^2}$. We also trivially have $\Pr[\max_i |y_i| > r] \leq \Pr[\vec{y} \neq \vec0] \leq 2^{-n}$, where the last step uses \cref{cor:banaszczyk_tail}. Therefore,
			\[
			    \expect_{\vec{y} \sim D_{\lat}}[\max_i |y_i|^k] \leq k\int_0^\infty r^{k-1} \min\{ 2ne^{-\pi r^2}, 2^{-n}\} \intd r = 2^{-n} r_0^k + 2nk\int_{r_0}^{\infty} r^{k-1} e^{-\pi r^2} \intd r
			    \; ,
			\]
			where $r_0 := \sqrt{\log(n 2^{n+1})/\pi }$. Note that
			\begin{align*}
			    \int_{r_0}^{\infty} r^{k-1} e^{-\pi r^2} \intd r 
			        &= (1/\pi)^{k/2}/2 \cdot \int_{\pi r_0^2}^\infty u^{(k-2)/2} e^{-u} \intd u  \\
			        &\leq e^{-\pi r_0^2} r_0^{k-2}/(2\pi) \cdot  \int_0^\infty \exp(-u + u(k-2)/(2\pi r_0^2)) \intd u \\ 
			        &= e^{-\pi r_0^2} \cdot \frac{ r_0^{k-2}}{2\pi (1-(k-2)/(2\pi r_0^2))}\\
			        &=  \frac{r_0^k}{1-(k-2)/(2\pi r_0^2)} \cdot \frac{1}{2^{n+1}n \log(n 2^{n+1})}\\
			        &\leq \frac{2^{-n} r_0^k}{2kn}
			        \; ,
			\end{align*}
			where in the last step we have used the fact that $k \leq n/10$, which in particular implies that $(1-(k-2)/(2\pi r_0^2)) \cdot \log(n 2^{n+1}) \geq k$.
			
			Putting everything together, we see that
			\[
			    |\eps| \leq (k+1)! 2^{1-n} r_0^k \leq 2^{-n} n^{2k}
			    \; ,
			\]
			as claimed.
		\end{proof}

\subsection{Worst-case lattice problems}

The running time of lattice algorithms depends on (1) the rank $n$ of the lattice; and (2) the number of bits $\ell$ required to represent the input (e.g., the individual entries in the basis matrix; the threshold $d$; the target vector $\vec{t} \in \Q^n$; etc.). We adopt the practice, common in the literature on lattices, of completely suppressing factors of $\poly(\ell)$ in our running time. Similarly, we sometimes describe our algorithms as though they work with real numbers, though in reality they must of course work with a suitable discretization of the real numbers.

\begin{definition}
    For any $\gamma = \gamma(n) \geq 1$, $\gamma$-$\gapSVP$ is the promise problem defined as follows. The input is (a basis for) a lattice $\lat \subset \Q^n$ and a distance threshold $d > 0$. It is a YES instance if $\lambda_1(\lat) \leq d$ and a NO instance if $\lambda_1(\lat) > \gamma d$.
\end{definition}

\begin{definition}
    For any $\gamma = \gamma(n) \geq 1$, $\gamma$-$\gapCVP$ is the promise problem defined as follows. The input is (a basis for) a lattice $\lat \subset \Q^n$, a target $\vec{t} \in \Q^n$, and a distance threshold $d > 0$. It is a YES instance if $\dist(\vec{t},\lat) \leq d$ and a NO instance if $\dist(\vec{t},\lat) > \gamma d$.
\end{definition}

Notice that by rescaling the lattice (and target) appropriately, we may always assume without loss of generality that $d$ is any fixed value. E.g., $\gamma$-$\GapSVP$ is equivalent to the problem of distinguishing $\lambda_1(\lat) \leq 1$ from $\lambda_1(\lat) > \gamma$, or distinguishing $\lambda_1(\lat) \leq \sqrt{n}/\gamma$ from $\lambda_1(\lat) > \sqrt{n}$. So, we sometimes implicitly work with a particular choice of $d$ that is convenient.

We will also need the following problem, which is known to be at least as hard as $\gamma$-$\SVP$.
	
	\begin{definition}
			For any $\gamma = \gamma(n) \geq 1$, $\gamma$-$\CVP$' is the promise problem defined as follows. The input is (a basis for) a lattice $\lat \subset \Q^n$ and a target $\vec{t} \in \Q^n$.  It is a YES instance if $\dist(\vec{t}, \lat) \leq \sqrt{n}/\gamma$. It is a NO instance if $\dist(\vec{t}, \lat) > \sqrt{n}$ \emph{and} $\lambda_1(\lat) > \sqrt{n}$.
		\end{definition}
		
				The following reduction due to Goldreich, Micciancio, Safra, and Seifert shows that $\CVP'$ is at least as hard as $\SVP$~\cite{GMSS99}. In~\cite{GMSS99}, they describe this reduction as a reduction to $\CVP$, rather than $\CVP$'. So, we reproduce the proof below to show that in fact it works with $\CVP'$. (The same idea was also used by Micciancio and Regev~\cite{MR04}. They actually observe that the reductions works for an even easier version of $\CVP$ in which $\dist(k\vec{t},\lat)$ is large for all odd $k$ in the NO case.)
		
		\begin{theorem}[\cite{GMSS99}]
			\label{thm:GMSS}
			For any $\gamma \geq 1$, there is an efficient reduction that maps one $\gamma$-$\SVP$ instance in $n$ dimensions into $n$ instances of $\gamma$-$\CVP'$ in $n$ dimensions such that the $\gamma$-$\SVP$ instance is a YES if and only if at least one of the resulting $\gamma$-$\CVP'$ instances is a YES and the $\gamma$-SVP instance is a NO instance if and only if all resulting $\gamma$-$\CVP'$ instances are NO instances.
			\end{theorem}
		\begin{proof}
			On input a basis $\basis := (\vec{b}_1,\ldots, \vec{b}_n)$ for a lattice $\lat \subset \R^n$ and distance $d > 0$, the reduction behaves as follows. It first rescales the lattice so that we may assume without loss of generality that $d = \sqrt{n}/\gamma$. Let $\basis_i := (\vec{b}_1,\ldots, \vec{b}_{i-1}, 2\vec{b}_i, \vec{b}_{i+1}, \ldots, \vec{b}_n)$ be the original basis ``with $\vec{b}_i$ doubled.'' Let $\vec{t}_i = \vec{b}_i$. The reduction creates the $\gamma$-CVP' instances given by $(\basis_i, \vec{t}_i)$.
			
			It is immediate that the reduction is efficient, so we only need to prove correctness. To that end, suppose that the input instance is a YES. I.e., suppose that $\lambda_1(\lat) \leq \sqrt{n}/\gamma$. Let $\vec{v} = \sum z_i \vec{b}_i \in \lat_{\neq \vec0}$ with $\|\vec{v}\| \leq \sqrt{n}/\gamma$. We may assume without loss of generality that $z_i$ is odd for some $i$. (Otherwise, we may replace $\vec{v}$ by $\vec{v}/2 \in \lat$.) Then, $\dist(\vec{t}_i, \lat(\basis_i)) \leq \|\vec{v}\| \leq \sqrt{n}/\gamma$, and therefore the $i$th $\gamma$-$\CVP'$ instance is a YES instance.
			
			On the other hand, suppose that the input is a NO instance. I.e., suppose that $\lambda_1(\lat) > \sqrt{n}$. Let $\vec{y} \in \lat(\basis_i)$ be a closest vector to $\vec{t}_i = \vec{b}_i$, and notice that $\vec{y} - \vec{t}_i \in \lat_{\neq 0}$. Therefore, $\dist(\vec{t}_i, \lat(\basis_i)) = \|\vec{y} - \vec{t}_i\| \geq \lambda_1(\lat) > \sqrt{n}$. Furthermore, since $\lat(\basis_i) \subset \lat$, we must have $\lambda_1(\lat(\basis_i)) \geq \lambda_1(\lat) > \sqrt{n}$. So, all of the $\gamma$-$\CVP'$ instances must be NO instances, as needed.
		\end{proof}
		
\paragraph{Lattice problems in general norms.}
We will also consider lattice problems in general norms.
For a lattice $\lat \subset \R^m$, target vector $\vec{t} \in \R^m$, and centrally symmetric convex body $K \subset \R^m$, we define
\[
\dist_K(\vec{t}, \lat) := \min_{\vec{x} \in \lat} \norm{\vec{t} - \vec{x}}_K \ \text{.}
\]

\begin{definition} \label{def:cvp-general-norm}
    For positive integers $m \geq n$, $K \subset \R^m$ a centrally symmetric convex body, and $\gamma = \gamma(m, n) \geq 1$, $\gamma$-$\gapCVP_K$ is the promise problem defined as follows. The input is (a basis for) a lattice $\lat \subset \Q^m$ of rank $n$, a target $\vec{t} \in \Q^m$, and a distance threshold $d > 0$. It is a YES instance if $\dist_K(\vec{t},\lat) \leq d$ and a NO instance if $\dist_K(\vec{t},\lat) > \gamma d$.
\end{definition}
We call $m$ the \emph{ambient dimension} of the lattice $\lat$. We note that in the $\ell_2$ norm one may assume that $m = n$ essentially without loss of generality by projecting, but that this is not true in general norms or even for $\ell_p$ norms for $p \neq 2$.

For conciseness, we only define $\CVP_K$ and not $\SVP_K$ formally. Indeed, the primary protocol that we consider in general norms (\cref{thm:generalized-GG}) works for $\CVP_K$, which, because $\SVP_K$ reduces efficiently to $\CVP_K$ for any norm $K$, implies that it also works for $\SVP_K$. Indeed, the reduction from $\SVP$ to $\CVP$ in~\cite{GMSS99} works in arbitrary norms $\norm{\cdot}_K$.

\subsection{Average-case lattice problems}

    We now define the two most important average-case lattice problems. We note that our definition of $\LWE$ is the \emph{decision} version, and we only work with the decision version throughout.
		
		\begin{definition}
		For positive integers $n, m$, and $q > 2$ with $m > n \log_2 q$, the $(n, m, q)$-Short Integer Solutions problem $(\SIS_{n, m, q})$ is the \emph{average-case} computational problem defined as follows. The input is a matrix $A \sim \Z_q^{m \times n}$, and the goal is to output $\vec{z} \in \{-1, 0, 1\}^m$ with $\vec{z} \neq \vec{0}$ such that $A \vec{z} = \vec{0} \bmod q$. 
		\end{definition}
		
		\begin{definition}
		    \label{def:LWE}
		    For positive integers $n$, $m$, and $q > 2$ and noise parameter $\alpha \in (0,1)$, the $(n,m,q,\alpha)$-decision LWE problem ($\LWE_{n,m,q,\alpha}$) is the \emph{average-case} computational decision problem defined as follows. The input consists of a uniformly random matrix $\vec{A} \sim \Z_q^{m \times n}$ and a vector $\vec{b} \in \Z_q^m$. How $\vec{b}$ is distributed depends on a uniformly random bit $\mu \sim \{0,1\}$. If $\mu = 0$, then $\vec{b} = \vec{A} \vec{s} + \vec{e} \bmod q$, where $\vec{s} \sim \Z_q^n$ and $\vec{e} \sim D_{\Z^m,\alpha q}$. Otherwise, $\vec{b} \sim \Z_q^m$ is uniformly random. The goal is to output $\mu$.
		\end{definition}
		
		Regev showed that public-key encryption exists if $\LWE_{n,m,q,\alpha}$ is hard for any $m \geq (1+\eps) n \log_2 q$ and $\alpha < o(1/\sqrt{m \log n})$~\cite[Section 5]{regevLatticesLearningErrors2009}. And, following Ajtai~\cite{ajtaiGeneratingHardInstances1996}, Goldreich, Goldwasser, and Halevi~\cite{GGHCollisionFreeHashingLattice2011} showed that a collision-resistant hash function (and therefore, e.g., secret-key encryption) exists if $\SIS_{n,m,q}$ is hard for any $m \geq n \log_2 q + 1$.

\subsection{Some useful inequalities}

		\begin{claim}
		    \label{clm:cos_series}
		    For integer $k \geq 0$ and $x \in \R$, let
		    \[
		        C_k(x) := \sum_{i=0}^k (-x^2)^i/(2i)!
		    \]
		    be the $2k$th truncation of the Taylor series for $\cos(x)$ around $x = 0$. Then, 
		    \[
		        C_{2k+1}(x) \leq \cos(x) \leq C_{2k}(x)
		        ; .
		    \]
		\end{claim}
		\begin{proof}
		    Define
		    \[
		        C_k'(x) := \sum_{i=0}^\infty (-1)^i x^{2k+2i+2}/(2k+2i+2)!
		        \; ,
		    \]
		    so that $\cos(x) = C_k(x) + (-1)^{k+1} C_k'(x)$. It suffices to prove that $C_k'(x) \geq 0$. We prove this by induction on $k$. Indeed, for $k= 0$, the result is trivial (as it is simply equivalent to the inequality $\cos(x) \leq 1$). On the other hand, for $k \geq 1$, we have that $\frac{\partial^2}{ \partial x^2} C_k'(x) = C_{k-1}'(x)$. In particular, by induction, this implies that the second derivative of $C_k'(x)$ is non-negative for all $x$. Since $C_k'(0) = 0$ (for $k \geq 1$) and since the first derivative of $C_k'$ is zero at zero, it follows that $C_k'(x) \geq 0$ for all $x$.
		\end{proof}

		\begin{claim}
		    \label{clm:tensor}
		    For any $\vec{w}_1,\ldots, \vec{w}_m \in \R^n$, any $\vec{t} \in \R^n$, and any integer $k \geq 1$,
		    \begin{align*}	        \left|V_{2k}\|\vec{t}\|^{2k} - \frac{1}{m} \sum_{i=1}^m \langle \vec{t}, \vec{w}_i \rangle^{2k}\right|     &\leq \|\vec{t}\|^{2k} \sum_{\alpha_1 + \cdots + \alpha_n = 2k} \left| V_{\alpha_1} \cdots V_{\alpha_n} - \frac{1}{m}\sum_{i=1}^m w_{i,1}^{\alpha_1} \cdots w_{i,n}^{\alpha_n} \right|
		        \; ,
		    \end{align*}
		    where 
		    \[
		        V_i := \int_{-\infty}^\infty x^i \exp(-\pi x^2) \intd x
		        \; .
		    \]
		\end{claim}
		\begin{proof}
		    Notice that
		    \[
		        \int_{\R^n} \langle \vec{t}, \vec{x} \rangle^{2k} \rho(\vec{x}) \intd \vec{x} = \|\vec{t}\|^{2k} \cdot \int_{-\infty}^\infty x^{2k} \rho(x) \intd x  = \|\vec{t}\|^{2k} V_{2k}
		    \]
		    because the Gaussian is radially symmetric.
		    
		    On the other hand, by expanding out the inner product, we have
		    \begin{align*}
		        \int_{\R^n} \langle \vec{t}, \vec{x} \rangle^{2k} \rho(\vec{x}) \intd \vec{x} 
		            &= \sum_{\alpha_1 + \cdots + \alpha_n = 2k} \Big(\int_{-\infty}^\infty (t_{1} x_1)^{\alpha_1} \rho(x_1) \intd x_1 \Big) \cdots \Big(\int_{-\infty}^\infty (t_{n} x_n)^{\alpha_n} \rho(x_n) \intd x_n \Big) \\
		            &= \sum_{\alpha_1 + \cdots + \alpha_n = 2k} t_1^{\alpha_1} \cdots t_n^{\alpha_n} \cdot V_{\alpha_1}  \cdots V_{\alpha_n}
		        \; .
		    \end{align*}
		    Similarly, we have
		    \[
		        \sum_{i=1}^m \langle \vec{t}, \vec{w}_i \rangle^{2k} = \sum_{\alpha_1 + \cdots + \alpha_n = 2k} t_1^{\alpha_1} \cdots t_n^{\alpha_n} \sum_{i=1}^m w_{i,1}^{\alpha_1} \cdots w_{i,n}^{\alpha_n}
		         \; .
		    \]
		    Combining everything together, we see that
		    \begin{align*}
		        \left|V_{2k} \|\vec{t}\|^{2k} - \frac{1}{m} \sum_{i=1}^m \langle \vec{t}, \vec{w}_i \rangle^{2k}\right|
		        &=  \sum_{\alpha_1 + \cdots + \alpha_n = 2k}  t_1^{\alpha_1} \cdots t_n^{\alpha_n}\left( V_{\alpha_1} \cdots V_{\alpha_n} - \frac{1}{m}\sum_{i=1}^m w_{i,1}^{\alpha_1} \cdots w_{i,n}^{\alpha_n}   \right) \\
		        &\leq \left( \max_{\alpha_1 + \cdots + \alpha_n = 2k} t_1^{\alpha_1} \cdots t_n^{\alpha_n} \right)\cdot  \sum_{\alpha_1 + \cdots + \alpha_n = 2k}  \left| V_{\alpha_1} \cdots V_{\alpha_n} - \frac{1}{m}\sum_{i=1}^m w_{i,1}^{\alpha_1} \cdots w_{i,n}^{\alpha_n}   \right| \\
		        &\leq \|\vec{t}\|^{2k} \cdot \sum_{\alpha_1 + \cdots + \alpha_n = 2k}  \left| V_{\alpha_1} \cdots V_{\alpha_n} - \frac{1}{m}\sum_{i=1}^m w_{i,1}^{\alpha_1} \cdots w_{i,n}^{\alpha_n}   \right|\; . &&\qedhere
		    \end{align*}
		\end{proof}
		
			\begin{proposition}
	    \label{prop:stupid_annoying_thing_about_moments_and_stuff}
	    Let $k_{\max}$ be a positive odd integer. Suppose that $\vec{w}_1,\ldots, \vec{w}_m$ satisfy the inequality
	    \[
	        \max_{\alpha_1 + \cdots + \alpha_n = 2k}\left| V_{\alpha_1}\cdots V_{\alpha_n} - \frac{1}{m}\sum_{i=1}^m w_{i,1}^{\alpha_1} \cdots w_{i,n}^{\alpha_n}   \right| \leq \eps
	    \]
	    for some $\eps > 0$ for all $0 \le k \leq k_{\max}$.
	    Then, for every $\vec{u} \in \R^n$ with $\|\vec{u}\| \leq r$, we have
	    \[
	        f_{W}(\vec{u}) \geq e^{-\pi r^2} - (10 r^2/k_{\max})^{k_{\max}+1} - \eps n^{2 k_{\max}} e^{2\pi r}
	        \; ,
	    \]
	    where
	    \[
	        f_W(\vec{u}) := \frac{1}{m} \cdot \sum_{i=1}^m \cos(2\pi \inner{\vec{w}_i, \vec{u}})
	        \; .
	    \]
	\end{proposition}
	\begin{proof}
	    By \cref{clm:cos_series}, we have 
		\begin{align*}
		f_W(\vec{u}) &:= \frac{1}{m} \cdot \sum_{i=1}^m \cos(2\pi \inner{\vec{w}_i, \vec{u}}) \\
		&\geq  \sum_{k=0}^{k_{\max}} \frac{(-4\pi^2)^k}{(2k)!} \cdot \frac{1}{m} \sum_{i=1}^m \inner{\vec{w}_i, \vec{u}}^{2k} 
		\; .
		\end{align*}
		Applying \cref{clm:tensor}, we see that 
		\[
		    \Big| V_{2k}\|\vec{u}\|^{2k} - \frac{1}{m} \sum_{i=1}^m \langle \vec{w}_i, \vec{u} \rangle^{2k} \Big|
		    \leq \eps n^{2k}\|\vec{u}\|^{2k} \leq \eps n^{2k} r^{2k}
		    \; .
		\]
		Therefore, 
		\begin{align*}
		   f_W(\vec{u}) 
		   \geq \sum_{k=0}^{k_{\max}} \frac{(-4\pi^2 \|\vec{u}\|^2)^k}{(2k)!} V_{2k} - \eps n^{2k_{\max}} \cdot \sum_{k=0}^{k_{\max}} \frac{(2\pi r)^{2k}}{(2k)!} \geq \sum_{k=0}^{k_{\max}} \frac{(-4\pi^2\|\vec{u}\|^2)^k}{(2k)!} V_{2k}  - \eps n^{2k_{\max}}  \exp(2 \pi r)
		   \; .
		\end{align*}
		
		On the other hand, using \cref{eq:Gaussian_fourier} and applying \cref{clm:cos_series} again, we have
		\[
		\rho(\vec{u}) = \expect_{\vec{w} \sim D}[\cos(2\pi \langle \vec{w}, \vec{u} \rangle)] \leq \sum_{k=0}^{k_{\max}+1} \frac{(-4\pi^2 \|\vec{u}\|^2)^k}{(2k)!} \cdot V_{2k} \leq \sum_{k=0}^{k_{\max}} \frac{(-4\pi^2\|\vec{u}\|^2)^k}{(2k)!} V_{2k} + \frac{(2\pi r)^{2k_{\max}+2}}{(2k_{\max}+2)!}\cdot V_{2(k_{\max}+1)}
		\; .
		\]
		Combining the two inequalities gives
		\[
		f_W(\vec{u})  \geq \rho(\vec{u}) - \frac{(2\pi r)^{2k_{\max}+2}}{(2k_{\max}+2)!}\cdot V_{2(k_{\max}+1)} - \eps n^{2k_{\max}}e^{2\pi r} \geq e^{-\pi r^2} - (10 r^2/k_{\max})^{k_{\max} + 1} - \eps n^{2k_{\max}} e^{2\pi r}
		\; ,
		\]
		as needed.
	\end{proof}
		
		\subsection{The Chernoff-Hoeffding bound}

\begin{lemma}
\label{lemma:chernoff}
    If $X_1,\ldots, X_N$ are independent identically distributed random variables with $|X_i| \leq r$, then for any $\delta \geq 0$,
    \[
        \Pr\left[ \Big|\frac{1}{N} \sum_{i=1}^N X_i - \expect[X_1]\Big| \geq \delta \right] \leq 2 \exp(-\delta^2 N/(2r^2))
        \; .
    \]
\end{lemma}

\subsection{Exponential Time Hypotheses}
\label{sec:prelims_ETH_stuff}

In the $k$-SAT problem, given a $k$-CNF formula~$\phi$, the task is to check if~$\phi$ has a satisfying assignment. The complement of $k$-SAT, the $k$-TAUT problem, is to decide if all assignments to the variables of a given $k$-DNF formula satisfy~it. Impagliazzo, Paturi, and Zane~\cite{IPZ98,IP99} introduced the following two hypotheses on the complexity of $k$-SAT, which now lie at the heart of the field of fine-grained complexity. (See~\cite{VW18} for a survery.)
\begin{definition}[Exponential time hypothesis 
    (ETH)~\cite{IPZ98,IP99}]
        There exists a constant $\eps>0$ such that no algorithm solves $3$-SAT on formulas with $n$~variables in time $2^{\eps n}$.
\end{definition}
\begin{definition}[Strong exponential time hypothesis 
    (SETH)~\cite{IPZ98,IP99}]
        For every constant $\eps>0$, there exists a constant~$k \geq 3$ such that no algorithm solves $k$-SAT on formulas with $n$~variables in time $2^{(1-\eps)n}$.
\end{definition}

Impagliazzo, Paturi, and Zane~\cite{IPZ98} proved the following sparsification lemma which, in particular, is used to show that the Strong exponential time hypothesis implies the Exponential time hypothesis.
\begin{theorem}[Sparsification lemma~\cite{IPZ98}]\label{thm:sparsification}
For every $k \geq 3$ and $\eps > 0$ there exists a constant $c = c(k, \eps)$ such that every $k$-SAT formula $\phi$ with $n$ variables can be expressed as $\phi = \lor_{i=1}^r \psi_{i}$ where $r \leq 2^{\eps n}$ and each $\psi_i$ is a $k$-SAT formula with at most $cn$ clauses. Moreover, all $\psi_i$ can be computed in $2^{\eps n}$-time.
\label{prop:sparsification}
\end{theorem}

Given the lack of progress on refuting these conjectures, one may propose even stronger hypotheses by considering more powerful classes of algorithms such as $\coNP, \coMA, \coAM, \IP$. \cite{CGIMPS16} studied co-non-deterministic versions of the above conjectures. 
\begin{definition}[Non-deterministic ETH
    (NETH)~\cite{CGIMPS16}]
        There exists an $\eps>0$ such that no non-deterministic algorithm solves $3$-TAUT on formulas with $n$~variables in time $2^{\eps n}$.
\end{definition}
\begin{definition}[Non-deterministic SETH (NSETH)~\cite{CGIMPS16}]
For every $\eps>0$, there exists~$k$ such that no non-deterministic algorithm solves $k$-TAUT on formulas with $n$~variables in time~$2^{(1-\eps)n}$.
\end{definition}
Jahanjou, Miles, and Viola~\cite{JMV15} showed that a refutation of ETH or SETH would imply super-linear lower bounds against general or series-parallel Boolean circuits for a problem in $\Eclass^{\NP}$, respectively. \cite{CGIMPS16} extended this result by showing that refuting NETH or NSETH is also sufficient for super-linear lower bounds against general or series-parallel Boolean circuits.

\paragraph{Probabilistic Proof Systems for $k$-TAUT.}
One can further strengthen the above conjectures by allowing the algorithms to use both non-determinism and randomness. We will consider private-coin and public-coin probabilistic proof systems for $k$-TAUT. 

By a constant-round $\IP$ protocol we denote a protocol where a computationally bounded probabilistic verifier and a computationally unbounded prover exchange a constant number of messages, and the verifier makes a decision by running a probabilistic algorithm on the input and transcript of their communication. In this work, we only consider constant-round $\IP$ protocols. (Note that the classical definition of the complexity class $\IP$ allows polynomially many rounds.)
By an $\MA$ protocol we mean a protocol where a computationally unbounded prover (Merlin) sends a message to a computationally bounded verifier (Arthur), and Arthur makes a decision by running a probabilistic algorithm on the input and Merlin's message. We will also consider the public-coin version of $\IP$ protocols---$\AM$ protocols---where all Arthur's messages are sequences of random bits, and Arthur is not allowed to use other random bits except for the ones revealed in his messages. By a constant-round $\AM$ protocol we mean a protocol where a computationally bounded probabilistic Arthur and computationally unbounded Merlin exchange a constant number of messages, and then Arthur makes a decision by running a deterministic algorithm on the input, and Arthur's and Merlin's messages.

By the complexity of an $\IP$, $\MA$, or $\AM$ protocol we will mean the maximum of (i) the total communication between the prover and verifier and (ii) the total running time of the verifier. We say that an ($\IP$, $\MA$, or $\AM$) protocol is a protocol for a given language if the prover can cause the verifier to accept every input in the language with probability at least $2/3$, while even if the prover behaves maliciously, the verifier still rejects every input not in the language with probability at least $2/3$. We write $\mathsf{MATIME}[T]$,  $\mathsf{AMTIME}[T]$, and  $\mathsf{IPTIME}[T]$ for the set of languages with such (constant-round) protocols with complexity bounded by $T$.

In the case of (constant-round) protocols with \emph{polynomial complexity}, it is known that $\MATIME[\poly] \subseteq \AMTIME[\poly] = \IPTIME[\poly]$, and the total number of messages sent by the prover and verifier can be reduced to two for $\AM$ and $\IP$ protocols~\cite{GS86,BM88}. But these results are not known to extend to the case of fine-grained complexity as both the public-coin simulation of private coins and the round-reduction procedure incur polynomial overhead in complexity (i.e., if the complexity of the original protocol is $T$, the complexity of the protocol after the transformation will be $\poly(T)$). For example, it is not known whether $\MATIME[T]\subseteq\AMTIME[T]$ or $\IPTIME[T]\subseteq\AMTIME[T]$, or whether $\AMTIME[T]$ is equivalent to its two-round variant.

Below we state two versions of ETH: one for $\MA$ protocols, and one for $\AM$ and $\IP$ protocols.
\begin{definition}[MAETH]
 There exists an $\eps>0$ such that no \MA{} protocol of complexity $2^{\eps n}$ solves $3$-TAUT on formulas with $n$~variables.
\end{definition}
\begin{definition}[AMETH]
 There exists an $\eps>0$ such that no constant-round \AM{} protocol of complexity $2^{\eps n}$ solves $3$-TAUT on formulas with $n$~variables.
\end{definition}
 By the results mentioned above, we can replace the constant-round $\AM$ protocols in the definition of AMETH by two-round \AM{} protocols or two-round $\IP$ protocols  without changing the definition. (The transformation only affects the unspecified constant $\eps$.)

Regarding the strong versions of the above hypotheses, Williams~\cite{W16} gave an \MA{} protocol of complexity $2^{n/2}$ solving $k$-TAUT (this result was later improved in~\cite{ACJRW22} to a protocol of complexity $2^{n/2-n/O(k)}$). 
Therefore, a natural version of SETH for the case of probabilistic protocols is whether there exist constant-round probabilistic protocols for $k$-TAUT of complexity significantly less than $2^{n/2}$. In this work, we will be using the two-round private-coin version of this conjecture.

\begin{definition}[IPSETH]\label{def:ipseth}
For every $\eps>0$, there exists~$k$ such that no two-round \IP{} protocol of complexity $2^{(1/2-\eps)n}$ solves $k$-TAUT on formulas with $n$~variables.
\end{definition}
We remark that while \cite{W16} and \cite{ACJRW22} give an $\MA$ protocol for $k$-TAUT of complexity roughly $2^{n/2}$, the known transformations of such a protocol into a \emph{two-round} \AM{} protocol~\cite{BM88} suffer a quadratic overhead in complexity, resulting in a protocol of trivial complexity roughly $2^n$.\footnote{By a two-round \AM{} protocol we mean a protocol where Arthur sends a sequence of public coins, Merlin responds with a message, and then Arthur makes a decision by running a \emph{deterministic} algorithm on the input and the transcript. The fact that Arthur's final verification procedure must be deterministic is the reason why the $\mathsf{MA}$ protocol in \cite{W16} does not trivially imply a two-round $\mathsf{AM}$ protocol with the same complexity.} Therefore, in the public-coin version of the above conjecture, one may also hypothesize that there is no two-round \AM{} protocol for $k$-TAUT with complexity significantly smaller than $2^n$. Using the standard techniques of simulating randomness by non-uniformity, a two-round \AM{} protocol of complexity $2^{(1-\eps)n}$ would also refute NUNSETH (Non-uniform non-deterministic SETH, see \cite[Definition~4]{CGIMPS16}). 
While faster than known probabilistic proof systems (or non-uniform non-deterministic algorithms) for TAUT are not known to imply circuit lower bounds, such protocols would still greatly improve the current state of the art, and some of these hypotheses are used as complexity assumptions (see, e.g., \cite{CGIMPS16,BGKMS23}).

\section{A \texorpdfstring{$\mathsf{coAM}$}{coAM} protocol}
\label{sec:coAM}

We next present a generalization of the Goldreich-Goldwasser $\coAM$ protocol for $\gamma$-$\GapCVP$~\cite{GG00} with $\gamma = O(\sqrt{n/\log n})$. 
Our protocol generalizes theirs in that we consider a general time-approximation-factor tradeoff beyond polynomial running times. In particular, we get a $\coAM$ protocol for constant-factor approximate $\GapCVP$ in any norm $\norm{\cdot}_K$ that runs in $2^{\eps n}$ time. This implies a barrier to proving fine-grained hardness of approximation results for $\GapCVP$, as we show in \cref{sec:limitations}.

In fact, we will give our protocol as a private-coin protocol and then apply a general transformation to convert it into a public-coin, $\coAM$ protocol.

\begin{theorem} \label{thm:generalized-GG}
Let $K \subseteq \R^n$ be a centrally symmetric convex body represented by a membership oracle, and let $\gamma = \gamma(n) \geq 1 + 1/n$. Then there is a two-round private-coin interactive proof (in fact, an honest-verifier perfect zero knowledge proof) for the complement of $\gamma$-$\GapCVP_K$ that runs in time $N \cdot \poly(n)$ for
\[
N := 10 n (1 - 1/\gamma)^{-n} \leq (1 - 1/\gamma)^{-n} \cdot \poly(n) \ \text{.}
\]
Furthermore, for $K = \B_2^n$ (i.e., in the $\ell_2$ norm) there is such a protocol that runs in time $N_2 \cdot \poly(n)$ for
\[
N_2 := 10 n^{3/2} \cdot (1 - 1/\gamma^2)^{-(n+1)/2} \leq (1 - 1/\gamma^2)^{-n/2} \cdot \poly(n) \ \text{.}
\]
\end{theorem}

\begin{proof}
Let $(\basis, \vec{t}, d)$ be the input instance of the complement of $\gamma$-$\GapCVP_K$, and let $\lat = \lat(\basis)$.
The interactive proof works by performing the following procedure $N$ times in parallel:
\begin{itemize}
    \item Arthur samples a uniformly random bit $b \in \bit$ and a uniformly random vector $\vec{s} \in r K + b \vec{t}$ for $r := \gamma d/2$.
He then sends $\vec{v} := \vec{s} \bmod \Par(\basis)$ to Merlin, and asks Merlin for the value of $b$.
    \item Merlin then sends Arthur a bit $b' \in \bit$ in response.

\end{itemize}
Arthur accepts if $b = b'$ for all of the $N$ trials, and rejects otherwise.

To sample a uniformly random point from a convex body in polynomial time, Arthur can use, e.g.,~\cite{DFKRandomPolynomialtimeAlgorithm1991}.
So, it is clear that the protocol runs in $N \cdot \poly(n)$ time, and it remains to show correctness. 

First, assume that the input instance is a YES instance, i.e., that $\dist_K(\vec{t}, \lat) > \gamma d$.
If $b = 0$, then
\[
\dist_K(\vec{v}, \lat) = \dist_K(\vec{s}, \lat) \leq \gamma d/2 \ \text{,}
\]
and if $b = 1$,
\[
\dist_K(\vec{v}, \lat) = \dist_K(\vec{s}, \lat) \geq \dist_K(\vec{t}, \lat) - \norm{\vec{s}-\vec{t}}_K > \gamma d - \gamma d/2 = \gamma d/ 2 \ \text{,}
\]
where the first inequality holds by triangle inequality. So, in either case Merlin can determine the value of $b$ with probability $1$, and will therefore send the correct bit $b' = b$ to Arthur in all $N$ trials, as needed.

Next, assume that the input instance is a NO instance, and let $\vec{u} \in \lat$ be a closest vector to $\vec{t}$ so that $\norm{\vec{t} - \vec{u}}_K = \dist_K(\vec{t}, \lat) \leq d$.
Notice that the distributions $(\vec{u} + \vec{s}) \bmod \Par(\basis)$ and $\vec{s} \bmod \Par(\basis)$ for $\vec{s} \sim r K$ are identical, and so for analysis assume that Arthur samples $\vec{s} \sim r K + \vec{u}$ (as opposed to $\vec{s} \sim r K$) if $b = 0$ and $\vec{s} \sim r K + \vec{t}$ if $b = 1$. 
Additionally notice that if $\vec{s} \in (rK + \vec{u}) \cap (rK + \vec{t})$ then (information theoretically) Merlin cannot correctly guess the value of $b$ given $\vec{v} = \vec{s} \bmod \Par(\basis)$ with probability greater than $1/2$.
Furthermore, Arthur samples such a vector $\vec{s} \in (rK + \vec{u}) \cap (rK + \vec{t})$ with probability
\begin{align}
p &:= \frac{\vol_n((rK + \vec{u}) \cap (rK + \vec{t}))}{\vol_n(rK + \vec{u})} \nonumber \\
&= \frac{\vol_n((K + \frac{1}{r} \vec{u}) \cap (K + \frac{1}{r} \vec{t}))}{\vol_n(K)} \nonumber \\
&= \frac{\vol_n(K \cap (K + \frac{2}{\gamma d} (\vec{t} - \vec{u})))}{\vol_n(K)} \nonumber \\
&\geq \frac{\vol_n(K \cap (K + \frac{2}{\gamma \norm{\vec{t} - \vec{u}}_K} (\vec{t} - \vec{u})))}{\vol_n(K)} \nonumber \\ 
&\geq (1 - 1/\gamma)^n  \label{eq:p-lb-gg-proof} \ \text{,}
\end{align}
where the first inequality uses the fact that $\norm{\vec{t} - \vec{u}}_K \leq d$ and the second inequality (i.e., \cref{eq:p-lb-gg-proof}) uses \cref{lem:normalized-K-intersection-vol}.
Therefore, the probability that Merlin answers correctly on all $N$ trials is less than
\[
(1 - p/2)^N \leq \exp(-pN/2) \leq \exp(-n) \ \text{,}
\]
as needed.

Now, consider the case where $K = \B_2^n$ and run the same protocol as for general $K$ except $N_2$ times in parallel instead of $N$.
Applying the stronger $\ell_2$-ball intersection lower bound from \cref{lem:normalized-l2-ball-intersection-vol} to the left-hand side of \cref{eq:p-lb-gg-proof}, we get
\[
p \geq \sqrt{1/(2 \pi n)} \cdot (1 - 1/\gamma^2)^{(n + 1)/2} \ \text{.}
\]
We then similarly get that the probability that Merlin answers correctly on all $N_2$ trials is less than
\[
(1 - p/2)^{N_2} \leq \exp(-p N_2/2) \leq \exp(-n) \ \text{,}
\]
as needed.

Finally, we note that the protocol is honest verifier perfect zero knowledge, as all bits sent by Merlin in the YES case are known to Arthur in advance.
\end{proof}

\begin{corollary}
\label{cor:ip2protocol}
Let $K \subseteq \R^n$ be a centrally symmetric convex body represented by a membership oracle. Then for every $\delta>0$ there is $\eps>0$ and a two-round $\IP$ protocol of complexity $2^{n(1-\eps)/2}$ for the complement of $\gamma$-$\GapCVP_K$ for $\gamma=2+\sqrt{2}+\delta$.
Furthermore, in the $\ell_2$ norm there is such a protocol for the complement of $\gamma$-$\GapCVP_2$ for $\gamma=\sqrt{2}+\delta$.
\end{corollary}
\begin{corollary}
Let $K \subseteq \R^n$ be a centrally symmetric convex body represented by a membership oracle. Then for every $\gamma=\gamma(n)\geq 1$, there is a two-round $\coAM$ protocol for $\gamma$-$\GapCVP_K$ running in time $2^{O(n/\gamma)}$.
Furthermore, in the $\ell_2$ norm there is such a protocol running in time $2^{O(n/\gamma^2)}$.
\end{corollary}
\begin{proof}
This result follows from \cref{thm:generalized-GG} and the general transformation of a two-round $\IP$ protocol of complexity $T$ into an $\cc{MAMAM}$ protocol of complexity $\poly(T)$~\cite{GS86}, and the latter protocol into a two-round $\AM$ protocol of complexity $\poly(T)$~\cite{BM88}.
\end{proof}

\section{A reduction from \texorpdfstring{$\GapSVP$}{GapSVP} to \texorpdfstring{$\BDD$}{BDD}}
\label{sec:gapsvp-bdd}

We next give a generalized reduction from $\gamma$-$\GapSVP$ to $\alpha$-$\BDD$ with an explicit tradeoff between $\gamma$, $\alpha$, and the running time of the reduction. The reduction itself is a straightforward adaptation of the proof of~\cite[Theorem 7.1]{conf/crypto/LyubashevskyM09}---which in turn adapts a reduction implicit in~\cite{peikertPublickeyCryptosystemsWorstcase2009}---but instantiated with the more fine-grained volume lower bound in \cref{lem:normalized-l2-ball-intersection-vol} and allowed to run in super-polynomial time. The reduction uses similar ideas to those in the Goldreich-Goldwasser protocol in \cref{thm:generalized-GG}.

\begin{theorem} \label{thm:gapsvp-to-bdd}
For any $\gamma = \gamma(n) \geq 1$ and $\alpha = \alpha(n) \in (0, 1/2)$ satisfying $\alpha \gamma \geq 1/2 + 1/n$, there is a randomized, dimension-preserving Turing reduction from $\gamma$-$\GapSVP$ on lattices of dimension $n$ to $\alpha$-$\BDD$ that makes at most
\[
N := 10 n^{3/2} \cdot \Big(1 - \frac{1}{(2 \alpha \gamma)^2}\Big)^{-(n+1)/2} \leq \Big(1 - \frac{1}{(2 \alpha \gamma)^2}\Big)^{-n/2} \cdot \poly(n)
\]
queries to its $\alpha$-$\BDD$ oracle and runs in $N \cdot \poly(n)$ time overall.
\end{theorem}

\begin{proof}
Let $(\basis, d)$ for $\basis \in \Q^{n \times n}$ and $d > 0$ be the input instance of $\gamma$-$\GapSVP$, and let $\lat = \lat(\basis)$.
The reduction repeats the following procedure $N$ times. First, it samples a uniformly random vector $\vec{s}$ in $r \B_2^n$ for $r := \alpha \gamma d$ and outputs $\vec{t} := \vec{s} \bmod \Par(\basis)$. It then calls its $\alpha$-$\BDD$ oracle on ($\basis$, $\vec{t}$), receiving as output a vector $\vec{v}$.
If $\vec{v} \neq \vec{t} - \vec{s}$ for some trial, the reduction outputs YES.
Otherwise, if $\vec{v} = \vec{t} - \vec{s}$ for all $N$ trials, the reduction outputs NO.

It is clear that each trial is efficient, and calls its $\alpha$-$\BDD$ oracle once on a lattice of rank $n$ (i.e., on the lattice $\lat$ in the input $\gamma$-$\GapSVP$ instance). 
So, it is clear that the reduction is dimension-preserving, makes $N$ oracle calls, and runs in $N \cdot \poly(n)$ time overall, as needed.

It remains to show correctness. First, suppose that the input is a NO instance of $\GapSVP$. Then $\vec{t} - \vec{s} \in \lat$ and
\[
\dist(\vec{t}, \lat) \leq \norm{(\vec{t} - \vec{s}) - \vec{t}} = \norm{\vec{s}} \leq \alpha \gamma d < \alpha \lambda_1(\lat) \ \text{.}
\]
Therefore, ($\basis$, $\vec{t}$) is a valid $\alpha$-$\BDD$ instance, and, because $\alpha < 1/2$, $\vec{t} - \vec{s}$ must be the only lattice vector within distance $\alpha \lambda_1(\lat)$ of $\vec{t}$.
So, on input ($\basis$, $\vec{t}$) the $\alpha$-$\BDD$ oracle must output $\vec{v} = \vec{t} - \vec{s}$.

Now, suppose that the input is a YES instance of $\GapSVP$. Notice that for the $\alpha$-$\BDD$ oracle to return $\vec{v} = \vec{t} - \vec{s}$ with probability $1$ it is information theoretically necessary for $\vec{s}$ to be the unique preimage of $\vec{t}$ according to the map $f : r \B_2^n \to \Par(\basis)$, $f: \vec{s} \mapsto \vec{s} \bmod \Par(\basis)$.
Otherwise, if there exist distinct $\vec{s}, \vec{s}' \in r \B_2^n$ with $\vec{s} \bmod \Par(\basis) = \vec{s}' \bmod \Par(\basis)$, then the probability that the $\alpha$-$\BDD$ oracle returns $\vec{v} = \vec{t} - \vec{s}$ is at most $1/2$.

Let $\vec{u}$ be a shortest non-zero vector in $\lat$. Because the input is a YES instance of $\GapSVP$, $\norm{\vec{u}} \leq d$. Moreover, note that if $\norm{\vec{s}} \leq r$ and $\norm{\vec{s} - \vec{u}} \leq r$ then $\vec{s}, \vec{s} - \vec{u}$ are both preimages of $f(\vec{s}) = \vec{s} \bmod \Par(\basis)$.

Let 
\[
p := \Pr_{\vec{s} \sim r \B_2^n}[\text{there exists } \vec{s}' \in r \B_2^n \setminus \set{\vec{s}} \text{ such that } \vec{s} \bmod \Par(\basis) = \vec{s}' \bmod \Par(\basis)]\ \text{.}
\]
Then
\begin{align*}
p &\geq \frac{\vol_n(r \B_2^n \cap (r \B_2^n + \vec{u}))}{\vol_n(r \B_2^n)} \\
&= \frac{\vol_n(\B_2^n \cap (\B_2^n + \frac{1}{\alpha \gamma d} \vec{u}))}{\vol_n(\B_2^n)} \\
&\geq \frac{\vol_n(\B_2^n \cap (\B_2^n + \frac{1}{\alpha \gamma \norm{\vec{u}}} \vec{u}))}{\vol_n(\B_2^n)} \\
&\geq \sqrt{1/(2 \pi n)} \cdot (1 - 1/(2 \alpha \gamma)^2)^{(n + 1)/2} \ \text{,}
\end{align*}
where the last inequality uses \cref{lem:normalized-l2-ball-intersection-vol}.
So, the probability that the $\alpha$-$\BDD$ oracle succeeds and outputs $\vec{v} = \vec{t} - \vec{s}$ on all $N$ trials is at most
\[
(1 - p/2)^N \leq \exp(-pN/2) \leq \exp(-n) \ \text{,}
\]
as needed.
\end{proof}\section{Worst-case to average-case reductions for LWE}
\label{sec:LWE}

We now show how to obtain better approximation factors in Regev's~\cite{regevLatticesLearningErrors2009} quantum worst-case to average-case reduction and Peikert's~\cite{peikertPublickeyCryptosystemsWorstcase2009} classical worst-case to average-case reductions for $\LWE$ by allowing the reduction to run in more time. Since hardness of $\LWE_{n,m,q,\alpha}$ implies secure public-key encryption for $\alpha < o(1/\sqrt{m \log n})$~\cite{regevLatticesLearningErrors2009}, our reductions immediately imply the existence of secure public-key cryptography from various forms of hardness of $\SVP$.

\subsection{A classical reduction}

For our classical reduction, we first recall the following result, which can be viewed as one of the main technical results in~\cite{regevLatticesLearningErrors2009} and~\cite{peikertPublickeyCryptosystemsWorstcase2009}. We will actually use a strengthening from~\cite{PRS17}, which allows us to work directly with decision LWE.  (I.e., we work with LWE as we have defined it in \cref{def:LWE}, whose hardness is known to directly imply public-key encryption. Work prior to~\cite{PRS17} used instead the associated search version of the problem and then relied on delicate search-to-decision reductions to prove hardness of the decision problem, and thus to prove security of public-key encryption.)

\begin{theorem}[{\cite[Lemma 5.4]{PRS17}}]
    \label{thm:PRS_classical}
    For any positive integers $n$, $m$, and $q \geq 2$ with $m > n \log_2 q$, and any noise parameter $\alpha \in (0,1)$, there is a polynomial-time (classical) algorithm with access to a (decision) $\LWE_{n,m,q,\alpha}$ oracle that takes as input a $\alpha'$-$\BDD$ instance over a lattice $\lat \subset \Q^n$ with the promise that $\sqrt{2n} q \leq \lambda_1(\lat) \leq 2\sqrt{2n} q$ and polynomially many samples from $D_{\lat^*}$, and solves the input $\alpha'$-$\BDD$ instance (with high probability), where 
    $\alpha' := \alpha/(4\sqrt{n})$.
\end{theorem}

This next theorem then follows from \cref{thm:PRS_classical} together with \cref{thm:BKZ_plus_GPV}, which shows us how to generate the relevant samples. (The parameter $\beta$ corresponds roughly to the block size in a basis-reduction algorithm.)

\begin{theorem}
    \label{thm:BDD_to_LWE}
    For any $2 \leq \beta \leq n/3$, any positive integers $n$, $m \leq \poly(n)$, and $q \geq \beta^{n/\beta}$ with $m > n \log_2 q$, and any noise parameter $\alpha \in (0,1)$, there is a (classical) reduction from $\alpha'$-$\BDD$ on a lattice with rank $n$ to (decision) $\LWE_{n,m,q,\alpha}$ that runs in time $2^{O(\beta)} \cdot \poly(n)$
    where $\alpha' := \alpha/(4\sqrt{n})$. Furthermore, the reduction makes only polynomially many calls to the $\LWE_{n,m,q,\alpha}$ oracle.
    
    Therefore, public-key cryptography exists if $\alpha'$-$\BDD$ requires classical time
    \[
        ((\alpha')^2 n^{3} \log n )^{\omega((\alpha')^2 n^3 \log n)}
        \; 
    \]
    for $1/(n^{3/2} \sqrt{\log \log n}) \leq \alpha' \leq 1/(n \log n)$ and in particular, if $o(1/(n \log n))$-$\BDD$ is $2^{\Omega(n)}$ hard.
\end{theorem}
\begin{proof}
    By trying many different rescalings of the input lattice $\lat$, we may assume that $\lat$ has been scaled so that
    \[
        \sqrt{2n} q \leq \lambda_1(\lat) \leq 2 \sqrt{2n} q
        \; .
    \]
    Notice that this also implies that $\lambda_1(\lat) \geq \beta^{n/\beta}$.
     The reduction first runs the procedure from \cref{thm:BKZ_plus_GPV} to generate samples $\vec{w}_1,\ldots, \vec{w}_N \sim D_{\lat^*}$ for $N := \poly(n)$ in time $2^{O(\beta)} \cdot \poly(n)$. The reduction then uses our LWE oracle to run the procedure from \cref{thm:PRS_classical} to try to solve the input $\BDD$ instance and outputs the result.
    
    It is clear that the reduction runs in the claimed time.
    And, by \cref{thm:BKZ_plus_GPV}, the $\vec{w}_i$ are in fact distributed as independent samples from $D_{\lat^*}$. So, by \cref{thm:PRS_classical}, the reduction will succeed.
    
    The result about public-key cryptography follows by recalling that Regev proved that public-key encryption exists if $\LWE_{n,m,q,\alpha}$ is hard for $m \geq (1+\eps) n \log_2 q$ and any $\alpha < o(1/\sqrt{m \log n})$. Setting $q := \ceil{2^{\eta/(n^2 (\alpha')^2 \log n)}}$ for some $1/(\log \log n) < \eta < o(1)$ and $m := 2\ceil{n \log q} = O(\eta / (n (\alpha')^2 \log n))$, we see that $\alpha = 4\sqrt{n} \alpha' < o(1/\sqrt{m \log n})$. So, to obtain secure public-key cryptography, it suffices to prove that $\LWE_{n,m,q,\alpha}$ is hard for these parameters. Indeed, setting $\beta := \Theta(\kappa \log \kappa)$, where $\kappa := (\alpha')^2 n^3 \log(n)/\eta^2$ gives a reduction from $\alpha'$-$\BDD$ that runs in time $2^{O(\beta)} \cdot \poly(n) = ((\alpha')^2 n^3 \log n)^{\Omega((\alpha') n^3 \log n /\eta^2)}$, where we may take $\eta < o(1)$ to be an arbitrarily large subconstant function. The result follows.
\end{proof}

The main classical result of this section then follows as an immediate corollary of \cref{thm:BDD_to_LWE} and \cref{thm:gapsvp-to-bdd}.

\begin{theorem}
    \label{thm:LWE_classical}
    For any positive integers $n$, $m \leq \poly(n)$, and $q$ with $m > n \log_2 q$, noise parameter $\alpha \in (0,1)$, and approximation factor $10\sqrt{n} /\alpha \leq  \gamma \leq n/(10\alpha)$, if $q \geq (n/(\alpha \gamma))^{(\alpha \gamma)^2/n}$, then there is a (classical) reduction from $\gamma$-$\gapSVP$ on a lattice with rank $n$ to (decision) $\LWE_{n,m,q,\alpha}$ that runs in time
    \[
        2^{O(n^2/(\alpha \gamma)^2)}\cdot \poly(n)
         \; .
    \]
    Therefore, public-key encryption exists if $\gamma$-$\gapSVP$ requires time \[
    2^{\omega((n^2/\gamma) \cdot \sqrt{\log(n) \log(n^2 \sqrt{\log n}/\gamma)})}
    \; 
    \]
    to solve with a classical computer for $ \omega(n \log n) < \gamma \leq O(n^2 \sqrt{\log \log n/\log n})$.
    
    In particular, for every constant $\eps > 0$, there exists $q = q(n) \leq \poly_\eps(n)$, such that there is a $2^{\eps n}$-time (classical) reduction from $O_\eps(\sqrt{n}/\alpha)$-$\gapSVP$ to $\LWE_{n,m,q,\alpha}$. Therefore, (exponentially secure) public-key cryptography exists if $\omega(n \log n)$-$\gapSVP$ is $2^{\Omega(n)}$ hard.
\end{theorem}
\begin{proof}
    The reduction works by composing the reduction from $\gamma$-$\SVP$ to $\alpha'$-$\BDD$ in \cref{thm:gapsvp-to-bdd} with the reduction from $\alpha'$-$\BDD$ to $\LWE_{n,m,q,\alpha}$ in \cref{thm:BDD_to_LWE}, where $\alpha' := \alpha/(4\sqrt{n})$ and $\beta := n^2/(2\alpha \gamma)^2$. Specifically, it runs the reduction from $\gamma$-$\SVP$ to $\alpha'$-$\BDD$ but simulates all queries to the $\alpha'$-$\BDD$ oracle using the reduction from $\alpha'$-$\BDD$ to $\LWE_{n,m,q,\alpha}$. The running time is as claimed because the first reduction runs in time $2^{O(\beta)}\cdot \poly(n) = 2^{O(n^2/(\alpha \gamma)^2)} \cdot \poly(n)$ time and makes $\poly(n)$ calls to the $\alpha'$-$\BDD$ oracle, each of which is simulated by the $\gamma$-$\SVP$ to $\alpha'$-$\BDD$ reduction, which runs in time
    \[
        \Big( 1- \frac{1}{(2\alpha' \gamma)^2} \Big)^{-n/2} \cdot \poly(n) = 2^{O(n^2/(\alpha \gamma)^2)} \cdot \poly(n)
        \; ,
    \]
    as needed.
    
    For the public-key cryptography result, we set $\alpha$ such that
    \[\alpha^{-2} = \eta\gamma \cdot \sqrt{\log (n) \log(n^2 \sqrt{\log n}/\gamma)}
    \; ,
    \]
    for \emph{any} $\eta = \eta(n)$ with $\omega(1) < \eta < \log \log n$.
    We then set
    \[
        q := \ceil{(n/(\alpha \gamma))^{\eta \cdot (\alpha \gamma)^2/n}}
         = 2^{O((\gamma/n) \cdot \sqrt{\log(n^2 \sqrt{\log n}/\gamma)/\log n})}
         \; ,
    \]
    and
    \[
        m := \ceil{2n \log_2 q} = O\big(\,\gamma \cdot \sqrt{\log(n^2 \sqrt{\log n}/\gamma)/\log n}\, \big)
        \; .
    \] 
    Notice that
    \[
        \alpha^{-2} > \omega\big(\,\gamma \cdot \sqrt{\log n \log(n^2 \sqrt{\log n}/\gamma)}\,\big) = \omega(m \log n)
        \; .
    \]
    Therefore, hardness of $\LWE_{n,m,q,\alpha}$ with these parameters implies secure public-key encryption. And, with these parameters, our reduction runs in time
    \[
        2^{O(n^2/(\alpha \gamma)^2)} \cdot \poly(n) = 2^{O( \eta \cdot (n^2/\gamma) \cdot \sqrt{\log(n) \log(n^2 \sqrt{\log n}/\gamma)} )}
        \; ,
    \]
    where we may take $\eta > \omega(1)$ to be an arbitrarily small superconstant function. The result follows.
\end{proof}

\subsection{A quantum reduction}

Our quantum reduction follows by combining Regev's worst-case to average-case reduction from $\BDD$ to $\LWE$ with our reduction from $\gapSVP$ to $\BDD$. Again, we use the version from \cite{PRS17}, which allows us to work directly with the decision version of the problem. 

\begin{theorem}[{\cite[Section 5]{PRS17}}]
    \label{thm:PRS_quantum}
    For any positive integers $n$, $m$, and $q$ with $m > n \log_2 q$ and noise parameter $\alpha \in (0,1)$ with $\alpha q \geq 2\sqrt{n}$, there is a polynomial-time quantum reduction from $\alpha'$-$\BDD$ to (decision) $\LWE_{n,m,q,\alpha}$, where $\alpha' := (\alpha/(4\sqrt{n}))$.
\end{theorem}
\begin{proof}
    This follows by combining \cite[Theorem 5.1]{PRS17} (which shows how to generate discrete Gaussian samples using a quantum computer with access to an $\LWE$ oracle) with \cite[Lemma 5.4]{PRS17} (which is our \cref{thm:PRS_classical}).
\end{proof}

Our main quantum reduction then follows immediately by combining the above with \cref{thm:gapsvp-to-bdd}.

\begin{theorem}
    \label{thm:LWE_quantum}
    For any positive integers $n$, $m < \poly(n)$, and $q$ with $m > n \log_2 q$, noise parameter $\alpha \in (0,1)$ with $\alpha q \geq 2\sqrt{n}$, and approximation factor $\gamma \geq 10 \sqrt{n}/\alpha$, there is a quantum reduction from $\gamma$-$\gapSVP$ to (decision) $\LWE_{n,m,q,\alpha}$, that runs in time
    \[
        \Big( 1 - \frac{4n}{(\alpha \gamma)^2}\Big)^{-n/2} \cdot \poly(n)
        \; .
    \]
    Therefore, public-key encryption exists if $\gamma$-$\gapSVP$ requires time
    \[
        2^{\omega(n^3 \log^2 n/\gamma^2)} \cdot \poly(n)
    \]
    to solve with a quantum computer for any $\gamma > \omega(n \log n)$.
    
    In particular, there is a $2^{\eps n}$-time quantum reduction from $O_\eps(\sqrt{n}/\alpha)$-$\gapSVP$ to $\LWE_{n,m,q,\alpha}$, and therefore (exponentially secure) public-key encryption exists if $\omega(n \log n)$-$\gapSVP$ requires $2^{\Omega(n)}$ time to solve with a quantum computer.
\end{theorem}
\begin{proof}
    The reduction follows immediately from \cref{thm:gapsvp-to-bdd,thm:PRS_quantum}. To obtain the public-key encryption result, we can take, e.g., $q = n^2$, $m = \ceil{2 n \log_2 q} = O(n \log n)$, and any $\alpha < o(1/\sqrt{m \log q}) = o(1/\sqrt{n} \log n))$.
\end{proof}\section{A co-non-deterministic Protocol}
\label{sec:coNP}

Our main technical result in this section is the following. Notice that one can view this as a verifier for (co-)$(4\sqrt{n/k})$-$\GapCVP'$. This result will immediately imply the existence of our co-non-deterministic protocol, and it will also be useful in \cref{sec:SIS}, where we use it to prove worst-case hardness of $\SIS$. (We did not attempt to optimize the constants in \cref{thm:GapCVP'_verifier,thm:coNP}.)

\begin{theorem}
    \label{thm:GapCVP'_verifier}
    There exists a deterministic algorithm that takes as input a (basis for a) lattice $\lat \subset \Q^n$ and target $\vec{t} \in \Q^n$, together with an odd integer $100 \leq k \leq n/(10\log n)$, and a witness $W := (\vec{w}_1,\ldots, \vec{w}_N) \in (\lat^*)^N$ with $N := (20k^2 n^2 \log n)^{2k+1} \leq n^{8k+4}$ and behaves as follows. 
    \begin{enumerate}
        \item If
        $\dist(\vec{t}, \lat) \leq \sqrt{k}/4$,
        then the algorithm always outputs CLOSE. 
        \item \label{item:correctness} If (1) the $\vec{w}_i$ are sampled independently from $D_{\lat^*}$; (2) $\dist(\vec{t},\lat) > \sqrt{n}$; and (3) $\lambda_1(\lat) > \sqrt{n}$, then the algorithm outputs FAR  with probability at least $1-3/N$.
        \item The running time of the algorithm
        is 
        \[\poly(n) \cdot (5n/k)^{2k} \cdot N = (100 k n^3 \log n)^{2k + O(1)} \leq (10n)^{8k + O(1)} 
        \; .
        \]
        \end{enumerate}
\end{theorem}

Before we prove the above, we observe that it immediately implies a co-non-deterministic protocol.

\begin{theorem}
    \label{thm:coNP}
    For any $100 \leq k \leq n/(10 \log n)$,
    $4 \sqrt{n/k}$-$\GapSVP$
    is in $\mathsf{coNTIME}[T]$ for running time
$T := (100 kn^3 \log n)^{2k +O(1)} \leq (10 n)^{8k + O(1)}$.
    In particular, for any constant 
    $C > 4 \sqrt{10}$,
    $(C\sqrt{\log n})$-$\GapSVP$ is in 
    $\mathsf{coNTIME}[e^{128 n/C^2 + o(n)}]$.
\end{theorem}
\begin{proof} 
    We may assume without loss of generality that $k$ is an odd integer. It follows from \cref{thm:GapCVP'_verifier} that
     $4 \sqrt{ n/k }$-$\gapCVP$'
    is in
    $\mathsf{coNTIME}[T]$.
    Specifically, vectors sampled from a discrete Gaussian $D_{\lat^*}$ yield a witness with high probability. The result then follows immediately from \cref{thm:GMSS}, which reduces $\gamma$-$\gapSVP$ to $\gamma$-$\gapCVP'$.
\end{proof}

\subsection{Proof of \texorpdfstring{\cref{thm:GapCVP'_verifier}}{the theorem}}

We now present the algorithm claimed in \cref{thm:GapCVP'_verifier} and prove that it satisfies the desired properties. On input $(\lat, \vec{t})$ together with an odd integer $k \geq 1$, and a witness $W := (\vec{w}_1,\ldots, \vec{w}_N) \in (\lat^*)^N$, the algorithm performs the following three checks. 
\begin{enumerate}
    \item It checks that $\vec{w}_i \in \lat^*$ for all $i$.
    \item It computes 
\[
    f_W(\vec{t}) := \frac{1}{N} \sum_{i=1}^N \cos(2\pi \langle \vec{w}_i, \vec{t} \rangle)
    \; ,
\]
and checks that $f_W(\vec{t}) < 20\sqrt{\log (N)/N}$. 
\item Finally, it checks that
\begin{equation}
		        \label{eq:moment_check_alg}
		         \Big|V_{\alpha_1} \cdots V_{\alpha_n} - \frac{1}{N} \sum_{i=1}^N w_{i,1}^{\alpha_1} \cdots w_{i,n}^{\alpha_n} \Big| \leq \eps
		        \; .
\end{equation}
for all $(\alpha_1,\ldots, \alpha_n) \in \Z_{\geq 0}^n$ with $\alpha_1 + \cdots + \alpha_n \leq 2k$, where $V_{2a-1} = 0$ for all integers $a$,
\[
    V_{2a} := \int_{-\infty}^\infty x^{2a} e^{-
    \pi x^2} \intd x = \frac{(2a)!}{(4\pi)^a \cdot a!}
    \; ,
\]
and
\[
    \eps := 20\log^{k}(2nN) \cdot \sqrt{k \log(n)/N}
    \; .
\]
\end{enumerate}
If all three checks pass, then it outputs FAR. Otherwise, it outputs CLOSE.

Notice that \cref{eq:moment_check_alg} consists of
\begin{align*}
    |\{(\alpha_1,\ldots \alpha_n) \in \Z_{\geq 0}^n \ : \ \alpha_1 + \cdots + \alpha_n \leq 2k\}| 
    &= |\{(\alpha_1,\ldots, \alpha_n,\alpha_{n+1}) \in \Z_{\geq 0}^n \ : \ \alpha_1 + \cdots + \alpha_{n+1} = 2k\}| \\
    &= \binom{n+2k}{2k} \\
    &\leq (e \cdot (n + 2k)/(2k))^{2k} \leq (5n/k)^{2k}.
\end{align*}
inequalities, each of which can be checked in time $\poly(n) \cdot N$. The other two checks can be performed far more efficiently.
It follows that the running time is as claimed. Of course, the hard part is proving correctness and soundness.

\subsubsection{Correctness of the algorithm in the FAR case}

In this section, we prove that the algorithm is correct in the FAR case. We assume that the $\vec{w}_i$ are sampled independently from $D_{\lat^*}$, and that $\min\{\dist(\vec{t},\lat), \lambda_1(\lat)\}\geq \sqrt{n}$, and we wish to show that the algorithm outputs FAR with probability at least $1-3/N$.

Indeed, it is immediate that the $\vec{w}_i$ are in fact in the dual lattice. Furthermore, by \cref{eq:PSF_Gaussian}, we have
\[
    \expect_{\vec{w} \sim D_{\lat^*}}[\cos(2\pi \langle \vec{w}, \vec{t} \rangle)]  = f(\vec{t}) \leq 2^{-n}
    \; ,
\]
where the inequality is by \cref{thm:banaszczyk_tail}.
By the Chernoff-Hoeffding bound (\cref{lemma:chernoff}), it follows that
\[
    f_W(\vec{t}) := \frac{1}{N} \sum_{i=1}^N \cos(2\pi \langle \vec{w}_i, \vec{t} \rangle) \leq 2^{-n} + 10\sqrt{\log(N)/N} < 20 \sqrt{\log (N)/N}
\]
except with probability at most $1/N$. Therefore, the second test passes with high probability.

It remains to show that the $\vec{w}_i$ satisfy \cref{eq:moment_check_alg} except with probability at most $2/N$. To that end, first notice that by \cref{lem:subgaussianity} and a union bound, we have that $|w_{i,j}| \leq \sqrt{2\log(2nN)/\pi}$ for all $i,j$, except with probability at most $1/N$. Conditioned on this, we have by the Chernoff-Hoeffding bound (\cref{lemma:chernoff}) that
\[
    \Pr\Big[ \Big|\frac{1}{N} \sum_{i=1}^N w_{i,1}^{\alpha_1} \cdots w_{i,n}^{\alpha_n} - \expect_{\vec{w} \sim D_{\lat^*}}[w_{1}^{\alpha_1} \cdots w_{n}^{\alpha_n}]  \Big| \geq \delta \Big] \leq \exp(- \delta^2 N/ (2\log^{2k}(2nN)))
    \; .
\]
In particular, if we take 
$ \delta := 10\log^k(2nN) \sqrt{k \log(n)/N}$, then we see that 
\[
    \Big|\frac{1}{N} \sum_{i=1}^N w_{i,1}^{\alpha_1} \cdots w_{i,n}^{\alpha_n} - \expect_{\vec{w} \sim D_{\lat^*}}[w_{1}^{\alpha_1} \cdots w_{n}^{\alpha_n}]  \Big| < \delta
\]
for any fixed choice of $(\alpha_1,\ldots,\alpha_n)$ with $\alpha_1 + \cdots + \alpha_n \leq 2k$ except with probability at most $n^{-10k} < n^{-2k}/N$. After taking a union bound over all $\binom{n+2k}{2k} \leq n^{2k}$ choices of $(\alpha_1,\ldots, \alpha_n)$, we see that this holds simultaneously for all such choices except with probability at most $1/N$.

Finally, by \cref{lem:smooth_Hermite_multi-dim}, we have that
\[
    \Big|\expect_{\vec{w} \sim D_{\lat^*}}[w_{1}^{\alpha_1} \cdots w_{n}^{\alpha_n}] - V_{\alpha_1} \cdots V_{\alpha_n}\Big| \leq n^{2k} 2^{-n}
    \; .
\]
Putting everything together, we see that except with probability at most $2/N$, we have
\[
    \Big|\frac{1}{N} \sum_{i=1}^N w_{i,1}^{\alpha_1} \cdots w_{i,n}^{\alpha_n} - V_{\alpha_1} \cdots V_{\alpha_n}  \Big| < 10 \log^k(2nN) \sqrt{k \log(n)/N} + n^{2k} 2^{-n} \leq \eps \; ,
\]
as needed.

\subsubsection{Correctness of the algorithm in the CLOSE case}
It remains to prove that the algorithm is sound in the CLOSE case. To that end, we assume that $\dist(\vec{t},\lat) \leq r$ where 
$r := \frac{\sqrt{k}}{4}$
and that $\vec{w}_i \in \lat^*$ are dual vectors satisfying \cref{eq:moment_check_alg}, and we prove that this implies that $f_W(\vec{t}) \geq 20 \sqrt{\log(N)/N}$.

Indeed, since the $\vec{w}_i$ are dual vectors, notice that $f_W(\vec{t}) = f_W(\vec{t} - \vec{y})$ for any $\vec{y} \in \lat$. Taking $\vec{y} \in \lat$ with $\|\vec{y} - \vec{t}\| \leq r$ and defining $\vec{u} := \vec{t} - \vec{y}$, we see that it suffices to prove that
\[
    f_W(\vec{u}) \geq 20 \sqrt{\log(N)/N}
\]
for all $\vec{u} \in \R^n$ with $\|\vec{u}\| \leq r$. In fact, we will show that $f_W(\vec{u}) \geq e^{-\pi r^2}/2 \gg 20 \sqrt{\log N/N}$.
Indeed, by \cref{prop:stupid_annoying_thing_about_moments_and_stuff}, we see that
\begin{align*}
    f_W(\vec{u}) 
        &\ge e^{-\pi r^2} - (10 r^2/k)^{k+1} - \eps n^{2k} e^{2\pi r}\\
        &= e^{-\pi r^2} - (10r^2/k)^{k+1} - 20 \log^{k}(2nN) \sqrt{k\log (n)/N} \cdot n^{2k} \cdot e^{2\pi r} \;.
\end{align*}
Notice that we have chosen our parameters so that 
\[
    20 \log^k(2nN) \sqrt{k \log(n)/N} \cdot n^{2k} \cdot  e^{2\pi r} \leq (10 k \log n)^{k+1/2}/(20 k^2 \log n)^{k+1/2} \leq k^{-k}/4 < e^{-\pi r^2}/4\; ,
\]
and 
\[
    (10 r^2/k)^{k+1} < (16/10)^{-k} < (e^{\pi / 16})^{-k} / 4 = e^{- \pi r^2} / 4
\]
The result follows.

\section{A \texorpdfstring{$\mathsf{coMA}$}{coMA} Protocol}
\label{sec:coMA}

Our main technical result in this section is the following. This can be viewed as a randomized verifier for $\gamma$-$\coGapCVP$ with approximation factor $\gamma = (1+\alpha)/(\alpha \beta) \leq 2/(\alpha \beta)$. From this, we will immediately derive \cref{thm:coMA}, which gives a a $\mathsf{coMA}$ protocol for $\gamma$-$\CVP$. In \cref{sec:SIS}, we will also use \cref{thm:GapCVP_randomized_verifier} to give a worst-case to average-case reduction for $\SIS$. (We did not attempt to optimize the constants in \cref{thm:GapCVP_randomized_verifier,thm:coMA}.)

\begin{theorem}
    \label{thm:GapCVP_randomized_verifier}
    For any $\beta \le \alpha < 1/3$, there exists a randomized algorithm that takes as input a (basis for a) lattice $\lat \subset \Q^n$ and target $\vec{t} \in \Q^n$, and a witness $W := (\vec{w}_1,\ldots, \vec{w}_N) \in (\lat^*)^N$, where $N := 2^{10\alpha^2 n}$, and behaves as follows. 
    \begin{enumerate}
        \item If
        $\dist(\vec{t}, \lat) \leq \alpha \beta \sqrt{n}$,
        then the algorithm outputs CLOSE with probability at least $1 - 2^{-{\beta^2 n}}$. 
        \item \label{item:correctness_MA} If the $\vec{w}_i$ are sampled independently from $D_{\lat^*}$ and  $\dist(\vec{t},\lat) > (1+\alpha) \sqrt{n}$, then the algorithm outputs FAR with probability at least  $1 - 2^{-{\alpha^2 n}}$.
        \item The running time of the algorithm
        is $2^{ n(10 \alpha^2 + 2\beta^2)} \cdot \poly(n)$.
        \end{enumerate}
\end{theorem}

Before we prove \cref{thm:GapCVP_randomized_verifier}, we observe that it immediately implies the following result.

\begin{theorem}
    \label{thm:coMA}
    For any $\gamma = \gamma(n) \geq 20$,
    $\gamma$-$\GapCVP$
    is in $\mathsf{coMATIME}[2^{O(n/\gamma)} \cdot \poly(n)]$.\footnote{The way that we have described the protocol gives (very good but) imperfect completeness. This can be converted generically into a protocol with perfect completeness (at the expense of an additional constant in the exponent), by first repeating the protocol roughly $2^{\beta^2 n}$ times to decrease the completeness error so that it is small enough to allow us to apply a union bound over Arthur's coins---i.e., the coins used to sample $\vec{v}$ as described below. (This does require us to specify a finite grid from which $\vec{v}$ is sampled, but we ignore such details.) One can also observe that the function $f_W$ is Lipschitz and use this to argue directly that it suffices for Merlin to consider a finite set of choices for $\vec{v}$---formally, an appropriately sized net of an appropriately sized ball---as was done in~\cite[Section 6.2]{aharonovLatticeProblemsNP2005}. However, we will simply content ourselves with imperfect completeness.}
\end{theorem}
\begin{proof} 
Let $\alpha = \beta = \sqrt{2/\gamma}$. By possibly rescaling the input lattice $\lat$ and target $\vec{t}$, we may assume that the lattice is scaled so that in the YES case $\dist(\vec{t},\lat) > (1+\alpha)n$ and in the NO case $\dist(\vec{t},\lat) \leq \alpha \beta \sqrt{n}$. Then, in the protocol, Merlin chooses the witness $W$ that maximizes the probability in \cref{item:correctness_MA} of~\cref{thm:GapCVP_randomized_verifier}. Arthur simply runs the algorithm in~\cref{thm:GapCVP_randomized_verifier}. It is immediate that the protocol is sound and complete and that the running time is $2^{12\alpha^2 n} \cdot \poly(n) = 2^{24n/\gamma} \cdot \poly(n)$, and the gap between the YES and NO instance is $\frac{1+\alpha}{\alpha \beta} \le \frac{2}{\alpha \beta} = \gamma$.
\end{proof}

\subsection{Proof of\texorpdfstring{~\cref{thm:GapCVP_randomized_verifier}}{the theorem}}

We now present the algorithm claimed in \cref{thm:GapCVP_randomized_verifier} and prove that it satisfies the desired properties. On input $(\lat, \vec{t})$ and a witness $W := (\vec{w}_1,\ldots, \vec{w}_N) \in (\lat^*)^N$, the algorithm does the following. 
\begin{enumerate}
    \item It checks that $\vec{w}_i \in \lat^*$ for all $i$.
    \item It repeats the following steps $2^{2\beta^2 n}$ times. 
    \begin{enumerate}
    \item It samples $\vec{v}$ uniformly from $\B(0, \alpha \sqrt{n})$.
    \item It computes  \label{item:check_v}
\[
    f_W(\vec{v}) := \frac{1}{N} \sum_{i=1}^N \cos(2\pi \langle \vec{w}_i, \vec{t} \rangle)
    \; ,
\]
and checks that $f_W(\vec{v}) > e^{-\pi \alpha^2 n}/2$. 
 \item It computes \label{item:check_v_plus_t}
\[
    f_W(\vec{v}+\vec{t}) := \frac{1}{N} \sum_{i=1}^N \cos(2\pi \langle \vec{w}_i, \vec{t} + \vec{v}\rangle)
    \; ,
\]
and checks that $f_W(\vec{v}+\vec{t}) \le e^{-\pi \alpha^2 n}/2$. 
\end{enumerate}
\item If all checks pass, the algorithm outputs FAR. Otherwise, it outputs CLOSE. 
\end{enumerate}
The running time of the algorithm is immediate. 

\subsubsection{Correctness of the algorithm in the FAR case}

We now prove correctness in the case when $\dist(\vec{t},\lat) > (1+\alpha) \sqrt{n}$. 
Indeed, since $\|\vec{v}\|< \alpha \sqrt{n}$, we have by triangle inequality that $\dist(\vec{t} + \vec{v},\lat) >  \sqrt{n}$.  By \cref{eq:PSF_Gaussian}, we have that
\[
    \expect_{\vec{w} \sim D_{\lat^*}}[\cos(2\pi \langle \vec{w}, \vec{t}+\vec{v} \rangle)]  = f(\vec{t}+\vec{v}) \leq 2^{-n}
    \; ,
\]
where the inequality is by \cref{cor:banaszczyk_tail}.
By the Chernoff-Hoeffding bound (\cref{lemma:chernoff}), it follows that
\[
    f_W(\vec{t}+\vec{v}) := \frac{1}{N} \sum_{i=1}^N \cos(2\pi \langle \vec{w}_i, \vec{t}+\vec{v} \rangle) \leq 2^{-n} + 10\sqrt{\log(N)/N} < e^{-\pi \alpha^2 n}/2 \; .
\]
except with probability at most $1/N$.

Similarly, we have that
\[
    \expect_{\vec{w} \sim D_{\lat^*}}[\cos(2\pi \langle \vec{w}, \vec{v} \rangle)]  = f(\vec{v}) \geq e^{-\pi \alpha^2 n}
    \; ,
\]
where the inequality is \cref{thm:banaszczyk_shift_big}. By the Chernoff-Hoeffding bound again (\cref{lemma:chernoff}), it follows that
\[
    f_W(\vec{v}) := \frac{1}{N} \sum_{i=1}^N \cos(2\pi \langle \vec{w}_i, \vec{v} \rangle) \geq e^{-\pi\alpha^2 n} - 10\sqrt{\log(N)/N} > e^{-\pi\alpha^2 n}/2 \;. 
\]
except with probability at most $1/N$. 

By a union bound over all $2^{2\beta^2 n+1}$ checks made by the algorithm, all checks succeed except with probability at most $\frac{2^{2\beta^2 n+1}}{N} < 2^{-\alpha^2 n}$.

\subsubsection{Correctness of the algorithm in the CLOSE case} 

We now prove correctness in the case when $\dist(\vec{t},\lat) \leq \alpha \beta \sqrt{n}$. We may assume that the $\vec{w}_i$ are dual lattice vectors, since otherwise correctness is immediate. Notice that this immediately implies that $f_W(\vec{t} + \vec{v}) = f_W(\vec{u} + \vec{v})$, where $\vec{u} := \vec{t} - \vec{y}$ for any $\vec{y} \in \lat$. In particular, taking $\vec{y}$ to be a closest lattice vector to $\vec{t}$, we see that $\|\vec{u}\| = \dist(\vec{t},\lat) \leq \alpha \beta \sqrt{n}$.

Define $S:= \B(\vec{0}, \alpha \sqrt{n})\cap \B(\vec{u}, \alpha \sqrt{n})$. By \cref{lem:normalized-l2-ball-intersection-vol}, the ratio of the volume of $S$, and that of $\B(\vec{0}, \alpha \sqrt{n})$ is greater than $\delta := \sqrt{1/(2 \pi n)} \cdot (1 - \beta^2/4)^{(n+1)/2} \geq 2^{-\beta^2 n}$. 

Let $p$ be the probability that for $f_W(\vec{v}') > e^{-\pi \alpha^2}/2$ for $\vec{v}'\sim S$. Notice that the probability that \cref{item:check_v} fails (which is what we want to happen!) is 
\[
    \Pr_{\vec{v} \sim \B(\vec{0},\alpha \sqrt{n})}[f_W(\vec{v}) \leq e^{-\pi \alpha^2}/2] \geq \delta \Pr_{\vec{v}' \sim S}[f_W(\vec{v}') \leq e^{-\pi \alpha^2}/2] = \delta (1-p)
    \; .
\]
On the other hand, the probability that \cref{item:check_v_plus_t} fails is
\[
    \Pr_{\vec{v} \sim \B(\vec{0},\alpha \sqrt{n})}[f_W(\vec{v} + \vec{t}) > e^{-\pi \alpha^2}/2] = \Pr_{\vec{v} \sim \B(\vec{0},\alpha \sqrt{n})}[f_W(\vec{v} + \vec{u}) > e^{-\pi \alpha^2}/2] \geq \delta \Pr_{\vec{v}' \sim S}[f_W(\vec{v}') > e^{-\pi \alpha^2}/2] = \delta p
    \; .
\]
Therefore, the probability that either \cref{item:check_v} or \cref{item:check_v_plus_t} fails is at least $\delta \cdot \max\{p,1-p\} \geq \delta/2$. After iterating the test $2^{2\beta^2n}$ times, the probability that the algorithm correctly outputs CLOSE is at least
\[
    1 - (1-\delta/2)^{2^{2\beta^2 n}} \geq 1-e^{-\delta 2^{2\beta^2 n-1}} \gg 1-2^{-\beta^2 n}
    \; ,
\]
as needed.

\section{Worst-case to average-case reductions for SIS}
\label{sec:SIS}

The main result of this section is as follows.

\begin{theorem}
    \label{thm:SIS}
    For any constant $\alpha > 0$, any positive integers $q = q(n)$ and $m = m(n)$ satisfying $q \geq n^\alpha \sqrt{m}$ and $m > n \log_2 q$, there is a reduction from $\gamma$-$\gapSVP$ to $\SIS_{n,m,q}$ that runs in time
    \[
        T := 2^{O(\min \{\sqrt{m} n/\gamma, nm \log n/\gamma^2 \} )}
        \; .
    \]
    Therefore, collision-resistant hash functions exist if $\gamma$-$\SVP$ is $2^{\Omega(\min\{ n^{1.5} \sqrt{\log n}/\gamma, n \log^{3/2} n/\gamma^2 \})}$.
    
    In particular, there is a $2^{\eps n}$-time reduction from $O_\eps(\sqrt{m})$-$\gapSVP$ to $\SIS_{n,m,q}$, and collision-resistant hash functions exist if $O(\sqrt{n \log n})$-$\gapSVP$ is $2^{\Omega(n)}$ hard.
\end{theorem}

\subsection{From discrete Gaussian sampling to \texorpdfstring{$\SIS$}{SIS}}

To prove our main theorem, we first show how to use a $\SIS$ oracle to sample from the discrete Gaussian distribution (above the smoothing parameter). In fact, it is well known that a $\SIS$ oracle can in some sense be used to sample from a discrete Gaussian. Indeed, this is implicit already in the reduction of Micciancio and Regev~\cite{MR04} (who introduced the use of the discrete Gaussian in this context).  However, to our knowledge, no prior work actually directly proves that there is a reduction from $\sDGS$ as defined below to $\SIS$.
		
		\begin{definition}
		    For any approximation factor $\gamma = \gamma(n) \geq 1$, $\eps = \eps(n) > 0$, and error parameter $\delta = \delta(n) \in (0,1)$ the $(\gamma,\eps, \delta)$-approximate smooth Discrete Gaussian Sampling problem ($\sDGS_{\gamma, \eps, \delta}$) is a sampling problem defined as follows. The input is (a basis for) a lattice $\lat \subset \Q^n$ and a parameter $s \geq \gamma \eta_\eps(\lat)$. The goal is to output $\vec{y} \in \lat$ whose distribution is within statistical distance $\delta$ of $D_{\lat,s}$.
		\end{definition}
		
		Micciancio and Peikert come closest to reducing $\sDGS$ to $\SIS$~\cite{micciancioHardnessSISLWE2013}, but their (very elegant) reduction only works for a variant of $\sDGS$ in which the algorithm can output a single sample from $D_{\lat,s'}$ (together with the parameter $s'$ itself) for any $s'$ in some range $s/\gamma \leq s' \leq s$. 
	    As far as we know, this was sufficient for all known use cases prior to this work. (For example, in \cite{MR04}, they implicitly work with this version of discrete Gaussian sampling.) However, our reductions in the next section seem to inherently need many samples $\vec{y}_1,\ldots, \vec{y}_N$ from $D_{\lat,s'}$, all with \emph{the same} parameter $s'$ (because we need much more precise estimates on the function $f_W(\vec{t})$).
		
		At a technical level, this lack of control over the parameter $s$ in prior work arises because of a lack of control over the length of the vector output by the $\SIS$ oracle. We could resolve this issue by observing that there are only polynomially many possible values for $s'$ in the~\cite{micciancioHardnessSISLWE2013} reduction (since there are only polynomially many possible values for the length of the vector returned by the $\SIS$ oracle) and then use a pigeonhole argument to show that we can find $N$ such samples by running the~\cite{micciancioHardnessSISLWE2013} reduction $\poly(n) \cdot N$ times.

		Instead, we show how to modify the~\cite{micciancioHardnessSISLWE2013} reduction to output a sample with a fixed chosen parameter $s >0$ (provided that $s$ is not too small).
		At a technical level, the only new idea in our proof is the observation that the reduction can find a parameter $r$ such that with non-negligible probability the $\SIS$ oracle outputs a valid solution $\vec{z} \in \{-1,0,1\}^m$ that has length exactly $r$. 
		Since the proof of \cref{thm:DGStoSIS} follows \cite{micciancioHardnessSISLWE2013} quite closely, we defer it to \cref{app:DGStoSIS}. Nevertheless, we expect that this new reduction will be useful in other contexts as well.

		\begin{theorem}
                \label{thm:DGStoSIS}
		    For any constant $\alpha > 0$, any positive integers $q = q(n)$ and $m = m(n)$ satisfying $q \geq n^\alpha \sqrt{m}$ and $m > n \log_2 q$, and any $2^{-m} < \eps = \eps(n) < m^{-\omega(1)}$, there is a reduction from $\sDGS_{2\sqrt{m}, \eps, \delta}$ to $\SIS_{n, m, q}$ that runs in time $\poly(m,\log q)$, where $\delta \leq \poly(m) \cdot \eps$.
		\end{theorem}

\subsection{From \texorpdfstring{$\SVP$}{SVP} to discrete Gaussian sampling (two ways)}

\cref{thm:SIS} follows immediately from \cref{thm:DGStoSIS} together with the following two results. These results are themselves relatively straightforward corollaries of the main results in \cref{sec:coNP,sec:coMA}. In particular, we observe that Arthur could generate the witness himself in the protocols from \cref{sec:coNP,sec:coMA} if he could generate discrete Gaussian samples from the dual lattice. (This idea originally appeared in~\cite{MR04}.)

\begin{theorem}
    \label{thm:SVP_to_DGS_NP}
    For any approximation factors, $\gamma' = \gamma'(n) \geq 1$ and $\gamma = \gamma(n) \geq 1$ with $\Omega(\sqrt{\log n}) \leq \gamma/\gamma' \leq O(\sqrt{n})$, there is a reduction from $\gamma$-$\gapSVP$ on a lattice with rank $n$ to $\sDGS_{\gamma', 2^{-n}, \delta}$ on a lattice with rank $n$ that runs in time $n^{O(n (\gamma'/\gamma)^2)}$, for $\delta \leq n ^{-O(n (\gamma'/\gamma)^2)}$.
\end{theorem}
\begin{proof}
    Let $k$ be the nearest odd integer to $20n (\gamma'/\gamma)^2$. By \cref{thm:GMSS}, it suffices to show a reduction from $\gamma$-$\gapCVP'$. To that end, we use the procedure from \cref{thm:GapCVP'_verifier}. Specifically, on input a (basis for a) lattice $\lat \subset \Q^n$ and target $\vec{t} \in \Q^n$, our reduction first sets $N := (20k^2 n^2 \log n)^{2k+1}$ and uses its $\sDGS_{\gamma, \eps, \delta}$ oracle on the dual lattice $\lat^*$ to generate $N$ independent samples from the dual lattice $\vec{w}_1,\ldots, \vec{w}_N \in \lat^*$ with parameter $s := \gamma$.
    
    Let $\lat' := \gamma' \lat$, $\vec{t}' := \gamma' \vec{t}$, and $\vec{w}_i' := \vec{w}_i/\gamma'$, and note that $\vec{w}_i' \in (\lat')^*$. Finally, the reduction runs the verification procedure from \cref{thm:GapCVP'_verifier} on input $\lat'$, $\vec{t}'$, $k$, and $\vec{w}_1',\ldots, \vec{w}_N'$ and outputs YES if and only if the procedure outputs CLOSE.
    
    It is clear that the reduction runs in the claimed time. To prove correctness, we first notice that in the YES case in which $\dist(\vec{t},\lat) \leq \sqrt{n}/\gamma$, we have 
    \[\dist(\vec{t}',\lat') = \gamma' \dist(\vec{t},\lat) \leq \gamma' \sqrt{n}/\gamma \leq \sqrt{k}/4
    \; .
    \] It follows immediately from \cref{thm:GapCVP'_verifier} that the reduction always outputs YES in this case. 
    
    We now turn to the NO case when $\lambda_1(\lat) > \sqrt{n}$ and $\dist(\vec{t},\lat) > \sqrt{n}$. By \cref{cor:banaszczyk_tail}, we have that $\eta_{2^{-n}}(\lat^*) \leq \sqrt{n}/\lambda_1(\lat) < 1$. Therefore, since we took $s = \gamma'$, the input to the $\sDGS$ instance satisfies the promise that $s \geq \gamma' \eta_\eps(\lat^*)$ for $\eps = 2^{-n}$. It follows that the joint distribution of $\vec{w}_1,\ldots, \vec{w}_N \in \lat^*$ is within statistical distance $\delta N$ of $D_{\lat^*,s}^N$. Notice that this also implies that the $\vec{w}_i'$ are statistically close to independent samples from $D_{(\lat')^*}$. Furthermore, since $\gamma' \geq 1$, we clearly have $\lambda_1(\lat') = \gamma' \lambda_1(\lat) > \sqrt{n}$ and similarly $\dist(\vec{t}',\lat') = \gamma' \dist(\vec{t},\lat) > \sqrt{n}$. It follows from \cref{thm:GapCVP'_verifier} that the reduction outputs NO with probability at least $1-3/N - \delta N \gg 1/3$. This probability can then be boosted to, say, $1-2^{-n}$ by running the reduction $O(n)$ times.
\end{proof}

Our next result uses the same idea, but replaces the protocol from co-non-deterministic protocol from \cref{sec:coNP} with the $\mathsf{coMA}$ protocol from \cref{sec:coMA}. The same idea works because the witness for both protocols is the same. This achieves better approximation factors for running times larger than roughly $2^{n/\log n}$, and in particular in the important special case when the running time is $2^{\eps n}$.

\begin{theorem}
    \label{thm:SVP_to_DGS_MA}
    For any $\gamma' = \gamma'(n)$ and $\gamma = \gamma(n)$ with $1 \leq \gamma' \leq \gamma/20$, there is a reduction from $\gamma$-$\gapSVP$ on a lattice with rank $n$ to $\sDGS_{\gamma', 2^{-n},\delta}$ on a lattice with rank $n$ that runs in time $2^{24  n\gamma' /\gamma } \cdot \poly(n)$, and $\delta := 2^{-20n\gamma' /\gamma}$.
\end{theorem}
\begin{proof}
    The reduction is quite similar to the reduction in~\cref{thm:SVP_to_DGS_NP}.  By \cref{thm:GMSS}, it suffices to show a reduction from $\gamma$-$\gapCVP'$. To that end, we use the procedure from \cref{thm:GapCVP_randomized_verifier}. Let $\alpha := \beta := \sqrt{2 \gamma'/\gamma}$. Then, on input a (basis for a) lattice $\lat \subset \Q^n$ and target $\vec{t} \in \Q^n$, our reduction first sets $N := 2^{10 \alpha^2 n}$ and uses its $\sDGS_{\gamma', \eps, \delta}$ oracle on the dual lattice $\lat^*$ to generate $N$ independent samples from the dual lattice $\vec{w}_1,\ldots, \vec{w}_N \in \lat^*$ with parameter $s := \gamma'$.
    
    As in the proof of \cref{thm:SVP_to_DGS_NP}, let $\lat' := \gamma' \lat$, $\vec{t}' := \gamma' \vec{t}$, and $\vec{w}_i' := \vec{w}_i/\gamma'$, and note that $\vec{w}_i' \in (\lat')^*$. Finally, the reduction runs the verification procedure from \cref{thm:GapCVP'_verifier} on input $\lat'$, $\vec{t}'$, and $\vec{w}_1',\ldots, \vec{w}_N'$ and outputs YES if and only if the procedure outputs CLOSE.
    
    The running time of the reduction is clearly $2^{12 \alpha^2 n } \cdot \poly(n) = 2^{24 n \gamma'/\gamma} \cdot \poly(n)$ by \cref{thm:GapCVP_randomized_verifier}.
    Similarly to the proof in \cref{thm:SVP_to_DGS_NP}, in the YES case we have 
    \[
        \dist(\vec{t}',\lat') = \gamma' \dist(\vec{t},\lat) \leq \gamma' \sqrt{n}/\gamma = \alpha \beta \sqrt{n}
        \; ,
    \]
    and our algorithm therefore outputs YES correctly with high probability, regardless of the distribution of the $\vec{w}_i$.
    
    On the other hand, in the NO case, $\lambda_1(\lat) > \sqrt{n}$ and $\dist(\vec{t},\lat) > \sqrt{n}$. By \cref{cor:banaszczyk_tail}, we have that $\eta_{2^{-n}}(\lat^*) \leq \sqrt{n}/\lambda_1(\lat) < 1$. Therefore, since we took $s = \gamma'$, the input to the $\sDGS$ instance satisfies the promise that $s \geq \gamma' \eta_\eps(\lat^*)$ for $\eps = 2^{-n}$. It follows that the joint distribution of $\vec{w}_1',\ldots, \vec{w}_N' \in (\lat')^*$ is within statistical distance $\delta N$ of $D_{(\lat')^*}^N$. Furthermore, since $\gamma' \geq 1$, we clearly have $\lambda_1(\lat') = \gamma' \lambda_1(\lat) > \sqrt{n}$ and similarly $\dist(\vec{t}',\lat') = \gamma' \dist(\vec{t},\lat) > \sqrt{n}$. It follows from \cref{thm:GapCVP'_verifier} that the reduction outputs NO with probability at least $1-3/N - \delta N \gg 9/10$. This probability can then be boosted to, say, $1-2^{-n}$ by running the reduction $O(n)$ times.
\end{proof}

Finally, we provide the proof of \cref{thm:SIS}, which simply works by combining the above with \cref{thm:DGStoSIS}.

\subsection{Finishing the proof}

\begin{proof}[Proof of \cref{thm:SIS}]
    By \cref{thm:DGStoSIS}, there is an efficient reduction from $\sDGS_{\gamma', 2^{-n}, \delta}$ on a lattice with rank $n$ to $\SIS_{n,m,q}$ for $\gamma' = 2\sqrt{m}$ and $\delta \leq \poly(m) 2^{-n}$. And, by \cref{thm:SVP_to_DGS_NP,thm:SVP_to_DGS_MA}, we can reduce $\gamma$-$\gapSVP$ to $\sDGS_{\gamma', 2^{-n},\delta}$ in time 
    \[
        \min\{ 2^{O(n \gamma'/\gamma)}, n^{O(n (\gamma'/\gamma)^2)} \}  = \min\{ 2^{O(\sqrt{m} n/\gamma)}, n^{O(n m /\gamma^2)}\} 
        \; .
    \]
    Combining the reductions gives the main result. The result about collision-resistant hash functions follows by recalling that collision-resistant hash functions exist if $\SIS_{n,m,q}$ is hard for $m \geq n \log_2 q + 1$. So, we can set, e.g., $q := n^2$ and $m := \ceil{4 n \log_2 n}$.
\end{proof}\newcommand{\YES}{\mathsf{YES}}
\newcommand{\NO}{\mathsf{NO}}
\newcommand{\MAYBE}{\mathsf{MAYBE}}
\newcommand{\cA}{\mathcal{A}}

\section{Limitations to fine-grained hardness of \texorpdfstring{$\gapCVP$}{gapCVP} and \texorpdfstring{$\gapSVP$}{gapSVP}}
\label{sec:limitations}

In this section, we show that our various protocols do in fact imply barriers against proving hardness of $\gapCVP$ and $\gapSVP$. E.g., we show that if $\gamma$-$\gapCVP$ has a $2^{\eps n}$-time (co-non-deterministic, $\mathsf{coMA}$ or $\mathsf{coAM}$) protocol for some $\gamma$, then a suitably fine-grained reduction from $k$-SAT to $\gamma$-$\gapCVP$ implies a $2^{\eps n}$-time protocol for $k$-SAT. This might sound obvious, since one might simply think that we can run the protocol for $\gamma$-$\gapCVP$ on the instance(s) generated by the reduction, but it is not immediate because (1) $\gamma$-$\gapCVP$ and $\gamma$-$\gapSVP$ are promise problems; and (2) we need to account for Turing reductions, i.e., reductions that make (possibly a large number of) oracle queries. In particular, the reduction might produce many instances of $\gamma$-$\gapCVP$, some of which are $\YES$ instances, some $\NO$ instances, and some $\MAYBE$ instances (i.e., neither $\YES$ nor $\NO$), and it is not immediately clear what to do with the $\YES$ and $\MAYBE$ instances. Furthermore, the instances might be generated adaptively, \emph{and} additional subtleties could arise if the reduction could be caused to fail if the oracle's responses to $\MAYBE$ instances can depend on the content of previous oracle queries.

Aharonov and Regev encountered the same problem in~\cite{aharonovLatticeProblemsNP2005}, but showed in~\cite[Lemma B.1]{aharonovLatticeProblemsNP2005} that if a promise problem $\Pi=(\Pi_{\YES}, \Pi_{\NO})$ is in $\coNP$, and the problem $\Pi'=(\Pi_{\YES} \cup \Pi_{\MAYBE}, \Pi_{\NO})$ is in $\NP$ (where $\Pi_{\MAYBE}$ is the set of all instances that are neither in $\Pi_{\YES}$ nor in $\Pi_{\NO}$), then $\Pi$ cannot be $\NP$-hard (even) under Cook reductions unless the polynomial hierarchy collapses. In this section, we use the same idea to generalize this result to show the impossibility of fine-grained hardness of $\gapCVP$ and $\gapSVP$ under plausible conjectures. In particular, we account for Turing reductions that may run in super-polynomial time (recall that a Cook reduction is a polynomial-time Turing reduction). 

Here, we work with deterministic reductions (i.e., not randomized reductions). Clearly, if the reduction uses randomness then we will not, e.g., be able to use it to get a {co-non-deterministic} protocol, since these by definition cannot be randomized. A similar issue arises with $\mathsf{MA}$ protocols. However, we could work with randomized reductions for both $\mathsf{coAM}$ and $\mathsf{coIP}$ protocols, by having Arthur send the random coins used for the reduction in his first message.

In the following, we suppress factors that are polynomial in the input size, as is common in the literature.

\begin{lemma}\label{lem:arnp}
Let $g: \mathbb{N}\times\mathbb{N} \to \mathbb{N}$ and $f: \mathbb{N} \times \mathbb{N}\to \mathbb{N}$ and $\gamma = \gamma(n) \geq 1$. If there is a co-non-deterministic protocol for $\gamma$-$\gapCVP$ running in time $T(n)$ on lattices with rank $n$ \emph{and} a $g(n',m')$-time Turing reduction from $k$-$\mathrm{SAT}$ on $n'$ variables and $m'$ clauses to $\gamma$-$\gapCVP$ on lattices with rank at most $n=f(n',m')$, then $k$-$\mathrm{SAT}$ is in $\coNTIME[g(n',m')\cdot T(f(n',m'))]$.

The same result holds when $\gapCVP$ is replaced by $\gapSVP$.
\end{lemma}
\begin{proof}
We give the proof for $\CVP$ only. The proof for $\SVP$ is nearly identical.

The reduction, say $\cA$, makes at most $g(n',m')$ calls to an oracle for $\gamma(n)$-$\gapCVP$. %
The reduction is guaranteed to output the correct answer provided that the oracle gives the correct answer on all $\YES$ or $\NO$ instances, regardless of the output of the oracle on any $\MAYBE$ instances.

Let $V_1$ be the given $\coNTIME[T(n)]$ verifier for $\gamma$-$\gapCVP$, i.e., for every $\NO$ instance there exists a valid witness causing $V_1$ to output $\YES$ and for every $\YES$ instance, there is no witness that will cause $V_1$ to output $\YES$. Let $V_2$ be the natural $\mathsf{coNP}$ verifier for $\CVP$ that given a $\gamma$-$\gapCVP$ instance $(\vec{B}, \vec{t}, d)$  and a witness $\vec{v}$, accepts if $\vec{v} \in \lat(\vec{B})$ and $\|\vec{v} - \vec{t}\| \le \gamma(n) \cdot d$, i.e., there exists a witness if and only if the given instance is a $\YES$ instance \emph{or} a $\MAYBE$ instance of $\gamma$-$\gapCVP$. 

We will show that there is a verifier $V$ such that for any instance $\phi$ of $k$-SAT that is a $\NO$ instance, there is a witness such that the verifier accepts, and for any $\YES$ instance and for any witness, the verifier rejects. 

For a given input $k$-SAT instance $\phi$, let $x_1,\ldots, x_\ell$ be the sequence of queries to the $\gamma$-$\gapCVP$ oracle made by $\cA$ if (1) on all $\YES$ or $\NO$ instances the oracle answers correctly; and (2) on all $\MAYBE$ instances, the oracle answers $\YES$. (Notice that the sequence of queries made by the reduction might in general depend on the oracle's responses. The sequence $x_i$ is well defined because we have fixed the oracle's output on all instances---including $\MAYBE$ instances.)

An honest witness consists of $\ell$ pairs, one for each $x_i$. If $x_i$ is a $\NO$ instance, then we have the pair (``NO", $w_i$) where $w_i$ is a valid witness for $x_i$ for which $V_1$ accepts. If $x_i$ is a $\YES$, or $\MAYBE$ instance, then we have the pair (``YES", $w_i$), where $w_i$ is a valid witness for $x_i$ for which $V_2$ accepts. 

Finally, the verifier $V$ does the following. It simulates the reduction $\cA$. Whenever the reduction makes a call to $x_i$, if the $i$-th witness pair is (``NO", $w_i$), then the verifier checks if $V_1$ accepts on input $x_i, w_i$, and if it does, then the verifier responds the oracle query with $\NO$. If the $i$-th witness pair is (``YES", $w_i$), then the verifier checks if $V_2$ accepts on input $x_i, w_i$, and if so, then the verifier responds the oracle query with a $\YES$. Finally, the verifier $V$ accepts if all calls to $V_1$ and $V_2$ accept \emph{and} the simulated reduction $\cA$ outputs NO. Otherwise, the verifier $V$ rejects. 

Notice that the protocol clearly has the desired complexity. Completeness of the above protocol follows from the fact that if $\phi$ is a $\NO$ instance, then for each oracle call $x_i$, whether to a $\YES$ instance, a $\NO$ instance, or a $\MAYBE$ instance, there exists a valid witness $w_i$ that will cause $V$ to treat the oracle call as a $\YES$ when it is a $\YES$, a $\NO$ when it is a $\NO$, and a $\YES$ when it is a $\MAYBE$. Since the reduction is guaranteed to output the correct answer in such cases, $\cA$ will always reject in this case, causing $V$ to accept as needed.

For soundness, we see that if $\phi$ is a $\YES$ instance, then for $V$ to accept, we must have a valid witness corresponding to each oracle call $x_i$. This implies that $\cA$ receives a correct response for each oracle query, and hence must accept since $\cA$ is a correct reduction and $\phi$ is a $\YES$ instance. Thus, $V$ must reject. 
\end{proof}

Now we use \cref{thm:coNP} to obtain a barrier against proving $2^{\eps n}$-time hardness of $\gamma$-$\gapSVP$ for $\gamma \gtrsim \sqrt{\log n}$.
\begin{corollary}
    \label{cor:NETH}
    Unless NETH is false, for every constant $C_1\geq1$ there exist constants $C_2 \geq 1$ and $\eps > 0$ such that for any $\gamma \geq C_2 \sqrt{\log n}$ there is no $2^{\eps( n+m)}$-time Turing reduction from $3$-$\mathrm{SAT}$ on $n$ variables and $m$ clauses to $\gamma$-$\gapSVP$ on lattices with rank at most $C_1 (n+m)$.
\end{corollary}
\begin{proof}
Assuming NETH, there exists a constant $\eps_1>0$ such $3$-SAT on $n$ variables is not in $\coNTIME[2^{\eps_1 n}]$. Given an instance of $3$-SAT on $n$ variables, we use the sparsification lemma  (\cref{thm:sparsification}) to deterministically and efficiently reduce $3$-SAT on $n$ variables to $2^{\eps_1 n/3}$ instances of $3$-SAT on $n$ variables and $m=cn$ clauses for a constant $c=c(\eps_1)>0$. 

Setting $\eps=\eps_1/(3(c+1))>0$, by \cref{thm:coNP} there exists $C_2=O(\sqrt{C_1/\eps})$, such that $\gamma$-$\gapSVP$ for lattices of rank $C_1(n+m)$ is in $\coNTIME[2^{\eps (n+m)}]$ for $\gamma=C_2\sqrt{\log(n+m)}$. By \cref{lem:arnp}, a $2^{\eps (n+m)}$-time reduction from $3$-SAT on $n$ variables and $m$ clauses to $\gamma$-$\gapSVP$ on a lattice with rank $C_1 (n+m)$, implies that $3$-SAT on $n$ variables and $m$ clauses is in $\coNTIME[2^{2\eps (n+m)}]$.

Given the above deterministic reduction from $3$-SAT on $n$ variables to $2^{\eps_1 n/3}$ instances of $3$-SAT on $n$ variables and $cn$ clauses, and that $3$-SAT on $n$ variables and $cn$ clauses is in $\coNTIME[2^{2\eps (c+1)n)}]=\coNTIME[2^{2\eps_1/3}]$, we conclude that $3$-SAT on $n$ variables is in  $\coNTIME[2^{\eps_1 n}]$, which leads to a contradiction.
\end{proof}

A simple argument shows that this can be extended to $\coAMTIME$, $\coMATIME$, and $\coIPTIME$.

\begin{lemma}
\label{lem:aram}
Let $g: \mathbb{N} \times \mathbb{N} \to \mathbb{N}$ and $f: \mathbb{N} \times \mathbb{N} \to \mathbb{N}$, and let $\gamma = \gamma(n) \geq 1$. Suppose that $\gamma$-$\gapCVP$ can be solved in $\coAMTIME[T(n)]$ (respectively, $\coMATIME[T(n)]$, $\coIPTIME[T(n)]$), where $n$ is the rank of the lattice. If there is a $g(n',m')$-time Turing reduction from $k$-$\mathrm{SAT}$ on $n'$ variables and $m'$ clauses to $\gamma$-$\gapCVP$ on lattices with rank at most $n=f(n',m')$, then $k$-$\mathrm{SAT}$ is in $\coAMTIME[T(f(n',m')) \cdot g(n',m')]$ (respectively, $\coMATIME[T(f(n',m')) \cdot g(n',m')]$, $\coIPTIME[T(f(n',m')) \cdot g(n',m')]$). Furthermore, this transformation preserves the number of rounds in the protocol.

The same result holds when $\gapCVP$ is replaced by $\gapSVP$.
\end{lemma}

\begin{proof}
We only sketch where the proof is different. Now, the verifier $V_1$ is allowed randomness, and for each $\NO$ instance, there exists a witness such that the verifier $V_1$ accepts with probability $2/3$, and for a $\YES$ instance, no witness can make the verifier accept with probability more than $1/3$. Note here that the witness may or may not depend on the choice of randomness depending on whether we have that $\gapCVP$ is in $\coAMTIME[T(n)]$, $\coMATIME[T(n)]$, or $\coIPTIME[T(n)]$. The proof remains the same in either case. 

At the cost of an additional $\poly(n')$ factor in the running time, the failure probability can be amplified to $\frac{1}{3 g(n',m')}$. The result then follows by the same argument as the previous lemma, and a union bound. 
\end{proof}

Similarly to \cref{cor:NETH}, \cref{lem:aram} and \cref{thm:coMA} imply the following barrier against proving hardness of $\gamma$-$\gapCVP$ for a constant $\gamma$.
\begin{corollary}
    \label{cor:MAETH}
    Unless MAETH is false, for every constant $C_1\geq1$ there exist constants $C_2 > 1$ and $\eps > 0$ such that there is no $2^{\eps (n+m)}$-time Turing reduction from $3$-$\mathrm{SAT}$ on $n$ variables and $m$ clauses to $\gamma$-$\gapCVP$ on lattices with rank at most $C_1 (n+m)$ for any $\gamma \geq C_2$.
\end{corollary}

Finally, \cref{lem:aram}, together with \cref{cor:ip2protocol}, implies a barrier against proving hardness of $\gamma$-$\gapCVP_K$ for an \emph{explicit} constant $\gamma$.
\begin{corollary}
    \label{cor:IPSETH}
    Unless IPSETH is false, for every constant $\delta > 0$, there exist constants $\eps > 0$ and $k \geq 3$ such that for all (families of) centrally symmetric convex bodies $K$ there is no $2^{\eps n}$-time Turing reduction from $k$-SAT on $n$ variables to $\gamma$-$\gapCVP_K$ on lattices in $n$ dimensions for $\gamma = 2+\sqrt{2} + \delta$. Furthermore, in the case of the $\ell_2$ norm, $\gamma = \sqrt{2} +\delta$.
\end{corollary}
\begin{proof}
By \cref{cor:ip2protocol}, for every $\delta>0$, there exists $\eps_1>0$ such that $\gamma$-$\gapCVP_K$ is in $\coIPTIME[2^{(1-\eps_1)n/2}]$. Let $\eps=\eps_1/4>0$.
By \cref{lem:aram}, a $2^{\eps n}$-time reduction from $k$-SAT to $\gamma$-$\gapCVP_K$ implies that $k$-SAT is in $\coIPTIME[2^{n/2-\eps n}]$. Assuming IPSETH, there exists a constant $k\geq 3$ such that $k$-SAT is not in $\coIPTIME[2^{n/2-\eps n}]$, which concludes the proof of the corollary.
\end{proof}

\subsection{Comparison with existing fine-grained hardness results}
Here we note that \cref{cor:MAETH,cor:IPSETH} in particular present barriers to substantially improving the fine-grained hardness of approximation results for $\SVP_p$ and $\CVP_p$ (i.e., $\SVP$ and $\CVP$ in $\ell_p$ norms) shown in~\cite{bennettQuantitativeHardnessCVP2017,ASGapETHHardness2018,ABGSFinegrainedHardnessCVP2021,bennettImprovedHardnessBDD2022}.
Those papers rule out algorithms for approximating $\SVP_p$ and $\CVP_p$ with various values of $p$ to within small, non-explicit constant factors $\gamma = \gamma(p) > 1$ under variants of the Gap-ETH and Gap-SETH assumptions.
Somewhat more specifically, the results in~\cite{bennettQuantitativeHardnessCVP2017,ASGapETHHardness2018} rule out $2^{o(n)}$-time algorithms for $\gamma$-$\SVP_p$ and $\gamma$-$\CVP_p$ on lattices of rank $n$ for every $p \geq 1$ (including the important Euclidean case of $p = 2$) with some non-explicit constant $\gamma$ slightly greater than $1$ assuming (variants of) Gap-ETH. Moreover, the results in~\cite{bennettQuantitativeHardnessCVP2017,ABGSFinegrainedHardnessCVP2021} rule out $2^{(1-\eps)n}$-time algorithms for $\gamma$-$\CVP_p$ for any constant $\eps > 0$ and any $p \notin 2\Z$ under Gap-SETH, again for a constant $\gamma$ slightly greater than $1$. (See \cite{ASGapETHHardness2018,bennettImprovedHardnessBDD2022} for similar results for $\SVP_p$.)

We note that, by \cref{cor:MAETH}, assuming MAETH one cannot hope to strengthen the results of~\cite{bennettQuantitativeHardnessCVP2017,ASGapETHHardness2018} to rule out $2^{o(n)}$-time algorithms for $\gamma$-$\SVP_2$ or $\gamma$-$\CVP_2$ for arbitrarily large constant $\gamma \geq 1$ under (variants of) Gap-ETH.
Moreover, by \cref{cor:IPSETH}, assuming IPSETH one cannot hope to rule out $2^{(1-\eps)n}$-time algorithms for $\gamma$-$\SVP_K$ or $\gamma$-$\CVP_K$ for any norm $\norm{\cdot}_K$ and any constant $\eps > 0$ for constant $\gamma > 2 + \sqrt{2}$ under Gap-SETH.

\cref{cor:IPSETH} in particular implies that assuming IPSETH one cannot hope to strengthen the results of~\cite{ABGSFinegrainedHardnessCVP2021,bennettImprovedHardnessBDD2022} and rule out $2^{(1-\eps)n}$-time algorithms for $\gamma$-$\SVP_p$ or $\gamma$-$\CVP_p$ for such $\gamma$ under Gap-SETH. There \emph{is} a subtlety with this: the lower bounds shown for $\SVP_p$ and $\CVP_p$ in~\cite{ABGSFinegrainedHardnessCVP2021,bennettImprovedHardnessBDD2022} for general $\ell_p$ norms rule out algorithms whose running time depends on the \emph{rank}  $n$ of the underlying lattice, whereas the IP protocol in \cref{cor:ip2protocol} has running time depending on the \emph{ambient dimension} $m \geq n$ of the underlying lattice (and so the barrier result in \cref{cor:IPSETH} rules out algorithms whose running time depends on the ambient dimension of the underlying lattices). However, because \cref{cor:IPSETH} holds for general norms, we can assume essentially without loss of generality that our lattices $\lat$ have full rank (or nearly full rank) by working in the $K$-norm for $K := \B_p^m \cap \lspan(\lat)$ in the case of $\SVP$ and $K := \B_p^m \cap \lspan(\lat, \vec{t})$ in the case of $\CVP$, where $\B_p^m$ is the unit $\ell_p$ ball in $m$ dimensions. 
(We additionally note that~\cite{ABGSFinegrainedHardnessCVP2021,bennettImprovedHardnessBDD2022} do not show SETH-type, $2^{Cn}$-hardness with explicit $C > 0$ for $\SVP$ and $\CVP$ in the $\ell_2$ norm.)

\section*{Acknowledgements} 
ZB was supported by the Israel Science Foundation (Grant No.\ 3426/21), and by the European Union Horizon 2020 Research and Innovation Program via ERC Project REACT (Grant 756482).
NSD and SP were supported in part by the NSF under Grant No.~CCF-2122230.
VV was supported by DARPA under Agreement No. HR00112020023, NSF CNS-2154149, MIT-IBM Watson AI, Analog Devices, a
Microsoft Trustworthy AI grant and a Thornton Family Faculty Research Innovation Fellowship. Any opinions, findings and conclusions or recommendations expressed in this material are those of the author(s) and do not necessarily reflect the views of the United States Government or DARPA.

NSD thanks Chris Peikert for responding to oracle queries in less than unit time. DA, HB, AG, and NSD would like to thank the organizers of the Bertinoro \href{https://www.wisdom.weizmann.ac.il/\~robi/Bertinoro2019_FineGrained/index.html}{program on Fine Grained Approximation Algorithms and Complexity} at which some of this work was completed. DA, HB, RK, SP, NSD, and VV would like to thank the Simons Institute and the organizers of the \href{https://simons.berkeley.edu/programs/extended-reunion-lattices2022}{Lattices and Beyond Summer Cluster} at which some of this work was done. Finally, HB, SP, and NSD would like to thank the \href{https://www.quantumlah.org/}{National University of Singapore's Centre for Quantum Technologies} for graciously supporting their visit, at which, again, some of this work was done.

\newcommand{\etalchar}[1]{$^{#1}$}

\appendix

	   \section{A reduction from discrete Gaussian sampling to \texorpdfstring{$\SIS$}{SIS}}
	   \label{app:DGStoSIS}
	   
	   We now prove \cref{thm:DGStoSIS}.
	    To prove the theorem, we need the following lemma. The proof is a modification of the proof of \cite[Lemma 3.7]{micciancioHardnessSISLWE2013}.
	    
        \begin{lemma}
        \label{lemma:pigeonhole-version-3.6}
            Let $m = m(n)$ and $q = q(n) > 2$
            be positive integers with $m > n \log_2 q$. 
            Let $0 < \eps = \eps(n) < 1/m^{\omega(1)}$.
            There is an algorithm with access to an $\SIS_{n, m, q}$ oracle $\mathcal{O}_{\SIS}$, that on input (1) a basis $\basis$ of a lattice $\lat \subset \Q^n$, (2) a number $r > 0$ with $r^2 \in [m]$ such that 
            \[\Pr_{A \sim \Z_q^{m \times n}}[\vec{z} \gets \mathcal{O}_\SIS(A);\ \|\vec{z}\| = r,\ A\vec{z} = \vec{0} \bmod q] \geq 1/\poly(m)
            \]
            and (3) $\poly(m)$ samples $\vec{y} \sim D_{\lat, s}$ for some $s \geq \sqrt{2} q \cdot \eta_\eps(\lat)$,
            runs in time $\poly(m, \log q)$ 
            and outputs a vector
            $\vec{w} \in \lat$ 
            whose distribution is within statistical distance 
            $2 \eps$
            of $D_{\lat, sr / q}$.
        \end{lemma}
        \begin{proof}
            The algorithm repeats the following steps until it succeeds, or until it has tried $\poly(m)$ times.
            First, it takes $m$ unused samples of its input, $\vec{y}_1,\ldots, \vec{y}_m$, from $D_{\lat, s}$ and writes each $\vec{y}_i$ as $\vec{y}_i = \basis \vec{a}_i \bmod q \lat$ for some $\vec{a}_i \in \Z_q^n$. 
            It sets $A := (\vec{a}_1,\ldots, \vec{a}_m)$.
            Since $s \geq \sqrt{2} q \cdot \eta_\eps(\lat)$, and 
            $\eps$ is negligible,
            the distribution of $\vec{a_i}$ is within negligible statistical distance of uniform over $\Z_q^n$, 
            and so the distribution of the full matrix $A$ is also within negligible statistical distance of the uniform over $\Z_q^{m \times n}$. 
            Finally, the algorithm calls $\mathcal{O}_\SIS$ on input $A$. By hypothesis, with probability at least $1/\poly(m)$, it receives as output
            a nonzero $\vec{z} \in \{-1, 0, 1\}^m$ with $A\vec{z} = \vec{0} \bmod q$ and $\|\vec{z}\| = r$. 
    
            Assuming it succeeds, it outputs $\vec{w}$, where $\vec{w} = \sum_i z_i \vec{y_i} / q$. 
            Notice that after conditioning on any fixed $\vec{a}_i$, we have that $\vec{y}_i \sim D_{q \lat + \basis \vec{a_i}, s}$ is distributed as a discrete Gaussian over a coset of $q \lat$.
            It is clear that $\sum_i z_i \vec{a_i}$ is divisible by $q$, so $\basis (\sum_i z_i \vec{a_i}) \in q \lat$. Thus, $\vec{w} \in \lat$, and
            by \cref{thm:convolution}, 
            the distribution of $\vec{w}$ has statistical distance at most $\eps$ from $D_{\lat, s'}$, where $s' := \sqrt{\sum_i (s z_i)^2}/q = (s / q) \| \vec{z} \| = sr / q$. 
            
            In the unlikely case that the reduction did not succeed in obtaining $\vec{z}$, which occurs only with probability $\exp(-\Theta(m)) < \eps$, it outputs $\vec{0}$. This costs it only an additive $\eps$ in statistical distance.
        \end{proof}
            \begin{corollary}\label{cor:reduce_s_sqrt_m}
            Let $m = m(n)$ and $q = q(n) > 2$ be positive integers with $m > n \log_2 q$ and $q \geq \sqrt{m}$,
            and let $c \geq 1$ be a constant integer. Let $0 < \eps = \eps(n) < 1 / m^{\omega(1)}$. 
            There is an algorithm with access to a $\SIS_{n, m, q}$ oracle $\mathcal{O}_{\SIS}$, that on input (1) any basis $\basis$ of a lattice $\lat \subset \Q^n$,  (2) a number $r > 0$ with $r^2 \in [m]$ such that 
            \[\Pr_{A \sim \Z_q^{m \times n}}[\vec{z} \gets \mathcal{O}_{\SIS}(A);\ \|\vec{z}\| = r,\ A\vec{z} = \vec{0} \bmod q] \geq 1 / \poly(m)
            \; ,
            \]
            and (3) $\poly(m)$ samples $\vec{y} \sim D_{\lat, s}$ for some $s > 0$ such that $s(r / q)^{c-1} > \sqrt{2} q \cdot \eta_\eps(\lat)$,
            outputs a vector
            $\vec{w} \in \lat$ 
            whose distribution is within statistical distance
            $\poly(m) \cdot \eps $ of $D_{\lat, s(r / q)^c}$ and runs in time $\poly(m, \log q)$.
            \end{corollary}
            \begin{proof}
              This follows from running \cref{lemma:pigeonhole-version-3.6} 
              in $c$ rounds; each round reduces the parameter of its input by a factor of $r/q$.
              The accumulated statistical distance is upper bounded by $\poly(m) \cdot \eps$.
            \end{proof}
        
		\begin{proof}[Proof of \cref{thm:DGStoSIS}]
	       Let $\lat:= \lat(\basis)$ and $s$ be the input of the $\sDGS$ instance. 
		   Without loss of generality, we may assume that 
            $\norm{\widetilde{\basis}} \leq 2^n \eta_\eps(\lat)$ by running the LLL algorithm on the input basis.
            
            The reduction starts by making calls to $\mathcal{O}_\SIS$ with independently and uniformly generated inputs $A$ until the oracle has succeeded $m^3$ times. (A success is when $\vec{z} \gets \mathcal{O}_\SIS(A)$ is non-zero with $\vec{z} \in \{-1, 0, 1\}^m$ and $A \vec{z} = \vec{0} \bmod q$.) The reduction sets $r$ as the most common value of $\| \vec{z} \|$ for the successful outputs $\vec{z}$ (breaking ties arbitrarily).
            By the Chernoff-Hoeffding bound (\cref{lemma:chernoff}), 
            \[\Pr_{A \sim \Z_q^{m \times n}}[\vec{z} \gets \mathcal{O}_{\SIS(A)}; \|\vec{z}\| = r, A\vec{z} = \vec{0} \bmod q] \geq 1 / \poly(m)
            \; ,
            \]
            except with probability, say, $2^{-m} \leq \eps$.
            
            The reduction then runs in $j := 0, 1, \ldots, n$ phases. 
            In phase $j$, the input is a basis $\basis_j$ of $\lat$ where initially $\basis_0 = \basis$. Let $c := \ceil{\log_{q/r}(n)} \leq 1/\alpha + 1.$
            If $s(q/r)^c \geq \norm{\widetilde{\basis}_j}\cdot \sqrt{\log n}$, then the reduction samples $\poly(m)$ samples from $D_{\lat, s(q/r)^c}$ using the discrete Gaussian sampling algorithm from \cref{thm:BLP_sampler} and applies \cref{cor:reduce_s_sqrt_m} to obtain a sample of $D_{\lat, s}$.

            Otherwise, let $s' := \norm{\widetilde{\basis}_j}\cdot \sqrt{\log n}$. The reduction samples $\poly(m)$ samples from $D_{\lat, s'}$ using the discrete Gaussian sampling algorithm \cref{thm:BLP_sampler} and applies \cref{cor:reduce_s_sqrt_m} $n^2$ times to obtain $n^2$ independent samples of $D_{\lat, s' \cdot (r/q)^c}$. Using these sampled vectors, the algorithm computes a better basis $\basis_{j+1}$ just as \cite[Lemma 3.7]{micciancioHardnessSISLWE2013} did and proceeds to the next phase.

            We show that the reduction terminates after at most $n$ phases. In particular, consider the $n^2$ independent samples of $D_{\lat, s' \cdot (r/q)^c}$ we obtain from the $j$th phase. With overwhelming probability these sampled vectors all have length bounded by 
            $s' \cdot (r/q)^c \cdot \sqrt{n} \leq \| \widetilde{\basis}_j \| 
            \sqrt{\log n} / \sqrt{n} \leq \| \widetilde{\basis}_j \| / 2$ and contain $n$ linearly independent vectors. 
            These vectors can be transformed into the new basis $\basis_{j+1}$ with the guarantee that $\norm{\widetilde{\basis}_{j+1}} \leq \norm{\widetilde{\basis}_j}/2$. Hence, the reduction must terminate after at most $n$ phases given that $\norm{\widetilde{\basis}_0} \leq 2^n\eta_\eps(\lat)$ (since $s > \eta_\eps(\lat)$ by definition of $\sDGS)$. The guarantee of the statistical distance follows directly from \cref{cor:reduce_s_sqrt_m}.
		\end{proof}
\end{document}